\documentclass[preprint,1p,sort&compress]{elsarticle}

\usepackage{amsmath,amsfonts,amssymb,amsthm}
\usepackage{prooftree}
\usepackage{graphicx}
\usepackage{proof}
\usepackage[all]{xy}
\usepackage[only,llbracket,rrbracket]{stmaryrd}

\newtheorem{theorem}{Theorem}[section]
\newtheorem{lemma}[theorem]{Lemma}
\newtheorem{remark}[theorem]{Remark}

\newtheorem{definition}[theorem]{Definition}
\newtheorem{example}[theorem]{Example}

\newcommand{\subst}[2]{[{#1}/{#2}]}
\newcommand{\lvec}{\ensuremath{\lambda^{\!\!\textrm{vec}}}}
\newcommand{\canon}[1]{[#1]}
\newcommand{\cocanon}[1]{\{#1\}}

\newcommand{\Bool}{\mathcal{B}}
\newcommand{\True}{\mathcal{T}}
\newcommand{\False}{\mathcal{F}}
\newcommand{\true}{{\bf true}}
\newcommand{\false}{{\bf false}}
\newcommand{\eg}{{e.g.}~}
\newcommand{\ie}{{i.e.}~}

\newcommand{\FV}[1]{\ensuremath{FV}(#1)}
\newcommand\ve[1]{\ensuremath{\mathbf{#1}}}
\newcommand{\Sc}{\ensuremath{\mathsf{S}}}
\newcommand{\sui}[1]{\sum_{i=1}^{#1}}
\newcommand{\suj}[1]{\sum_{j=1}^{#1}}
\newcommand{\suk}[1]{\sum_{k=1}^{#1}}
\newcommand{\recap}[3]{\noindent\textbf{#1 \ref{#3}} (#2)\textbf{.}}
\newcommand{\xrecap}[2]{\noindent\textbf{#1 \ref{#2}}}
\newcommand{\ket}[1]{{|{#1}\rangle}}
\DeclareMathAlphabet{\mathpzc}{OT1}{pzc}{m}{n}\newcommand{\varu}[1]{\ensuremath{\mathpzc{#1}}}
\newcommand{\vara}[1]{\ensuremath{\mathbb{#1}}}
\newcommand{\V}{\mathcal{V}}
\newcommand{\Neutral}{\mathcal{N}}
\newcommand{\CT}{\Lambda_0}
\newcommand{\SN}{{\it SN}_0}
\newcommand{\Red}{{\rm Red}}
\newcommand{\RC}{\mathsf{RC}}
\newcommand{\RCn}{{\bf RC}}
\newcommand{\CC}{{\bf CC}}
\newcommand{\bcal}[1]{\mathsf{#1}}
\newcommand{\denot}[1]{\ensuremath{\llbracket{#1}\rrbracket}}
\newcommand{\lmat}{{\it Mat}}
\newcommand{\termdenot}[1]{\denot{#1}^{{\rm term}}}
\newcommand{\typedenot}[1]{\denot{#1}^{{\rm type}}}
\def\vddots{\raisebox{.7ex}{.}\raisebox{0ex}{.}\raisebox{-.7ex}{.}}

\allowdisplaybreaks

\journal{Information and Computation}

\begin{document}
\begin{frontmatter}

  \title{The Vectorial $\lambda$-Calculus\tnoteref{thanks}}
  \tnotetext[thanks]{Partially supported by the STIC-AmSud project FoQCoSS and PICT-PRH 2015-1208.}

  \author[Marseille]{Pablo Arrighi}
  \ead{pablo.arrighi@univ-amu.fr}
  \author[UNQ]{Alejandro D\'iaz-Caro}
  \ead{alejandro.diaz-caro@unq.edu.ar}
  \author[CSup]{Beno\^it Valiron}
  \ead{benoit.valiron@monoidal.net}

  \address[Marseille]{Aix-Marseille Universit\'e, CNRS, LIF, Marseille and
  IXXI, Lyon, France}
  \address[UNQ]{Universidad Nacional de Quilmes \& CONICET, B1876BXD Bernal,
  Buenos Aires, Argentina}
  \address[CSup]{LRI, CentraleSupélec, Univ. Paris Sud, CNRS, Universit\'e
  Paris-Saclay, 91405 Orsay Cedex, France}

  \begin{abstract}
    We describe a type system for the linear-algebraic  $\lambda$-calculus. The
    type system accounts for the linear-algebraic aspects of this extension of
    $\lambda$-calculus: it is able to statically describe the linear
    combinations of terms that will be obtained when reducing the programs.
    This gives rise to an original type theory where types, in the same way as
    terms, can be superposed into linear combinations.  We prove that the
    resulting typed $\lambda$-calculus is strongly normalising and features
    weak subject reduction. Finally, we show how to naturally encode matrices
    and vectors in this typed calculus.  
  \end{abstract}
\end{frontmatter}

\section{Introduction}
\subsection{(Linear-)algebraic $\lambda$-calculi}
A number of recent works seek to endow  the $\lambda$-calculus with a vector space structure. This agenda has emerged simultaneously in two different contexts.
\begin{itemize} 
  \item The field of \emph{Linear Logic} considers a logic of resources where the propositions themselves stand for those resources -- and hence cannot be discarded nor copied. When seeking to find models of this logic, one obtains a particular family of vector spaces and differentiable functions over these. It is by trying to capture these mathematical structures back into a programming language that Ehrhard and Regnier have defined the {\em differential $\lambda$-calculus} \cite{EhrhardRegnierTCS03}, which has an intriguing differential operator as a built-in primitive and an algebraic module of the $\lambda$-calculus terms over natural numbers. Vaux \cite{VauxMSCS09} has focused his attention on a `differential $\lambda$-calculus without differential operator', extending the algebraic module to positive real numbers. He obtained a confluence result in this case, which stands even in the untyped setting. More recent works on this {\em algebraic $\lambda$-calculus} tend to consider arbitrary scalars \cite{TassonTLCA09,EhrhardLICS10,AlbertiJFLA13}.
  \item The field of \emph{Quantum Computation} postulates that, as computers are physical systems, they may behave according to quantum theory. It proves that, if this is the case, novel, more efficient algorithms are possible \cite{ShorSIAM97,GroverSTOC96} -- which have no classical counterpart. Whilst partly unexplained, it is nevertheless clear that the algorithmic speed-up arises by tapping into the parallelism granted to us `for free' by the {\em superposition principle}, which states that if $\ve{t}$ and $\ve{u}$ are possible states of a system, then so is the formal linear combination of them $\alpha\cdot\ve{t}+\beta\cdot\ve{u}$ (with $\alpha$ and $\beta$ some arbitrary complex numbers, up to a normalizing factor). The idea of a module of $\lambda$-terms over an arbitrary scalar field arises quite naturally in this context. This was the motivation behind the {\em linear-algebraic $\lambda$-calculus}, or {\em Lineal} for short, by Dowek and one of the authors~\cite{ArrighiDowekLMCS17}, who obtained a confluence result which holds for arbitrary scalars and again covers the untyped setting. 
\end{itemize}
These two languages are rather similar: they both merge higher-order computation, be they terminating or not, in its simplest and most general form (namely the untyped $\lambda$-calculus) together with linear algebra also in its simplest and most general form (the axioms of vector spaces). In fact they can simulate each other \cite{AssafDiazcaroPerdrixTassonValironLMCS14}. Our starting point is the second one, {\em Lineal}, because its confluence proof allows arbitrary scalars and because one has to make a choice. Whether the models developed for the first language, and the type systems developed for the second language, carry through to one another via their reciprocal simulations, is a topic of future investigation.

\subsection{Other motivations to study (linear-)algebraic $\lambda$-calculi}  The two languages are also reminiscent of other works in the literature:
\begin{itemize}
  \item \emph{Algebraic and symbolic computation.} The functional style of programming is based on the $\lambda$-calculus together with a number of extensions, so as to make everyday programming more accessible. Hence since the birth of functional programming there have been several theoretical studies on extensions of the $\lambda$-calculus in order to account for basic algebra (see for instance Dougherty's algebraic extension \cite{DoughertyIC92} for normalising terms of the $\lambda$-calculus) and other basic programming constructs such as pattern-matching \cite{CirsteaKirchnerLiquoriFOSSACS01,ArbiserMiquelRiosJFP09}, together with sometimes non-trivial associated type theories \cite{PetitTLCA09}. Whilst this was not the original motivation behind (linear-)algebraic $\lambda$-calculi, they could still be viewed as an extension of the $\lambda$-calculus in order to handle operations over vector spaces  and make programming more accessible with them.
    The main difference in approach is that the $\lambda$-calculus is not seen here as a control structure which sits on top of the vector space data structure, controlling which operations to apply and when. Rather, the $\lambda$-calculus terms themselves can be summed and weighted, hence they actually are vectors, upon which they can also act.
  \item \emph{Parallel and probabilistic computation.} The above intertwinings of concepts are essential if seeking to represent parallel or probabilistic computation as it is the computation itself which must be endowed with a vector space structure. The ability to superpose $\lambda$-calculus terms in that sense takes us back to Boudol's parallel $\lambda$-calculus \cite{BoudolIC94} or de Liguoro and Piperno's work on non-deterministic extensions of $\lambda$-calculus \cite{deLiguoroPipernoIC95}, as well as more recent works such as \cite{PaganiRonchidellaroccaFI10,BucciarelliEhrhardManzonettoAPAL12,DiazcaroManzonettoPaganiLFCS13}. It may also be viewed as being part of a series of works on probabilistic extensions of calculi, \eg~\cite{BournezHoyrupRTA03,HerescuPalamidessiFOSSACS00} and \cite{DipierroHankinWiklickyJLC05,DalLagoZorziRAIRO12,DiazcaroDowekDCM13} for $\lambda$-calculus more specifically.
\end{itemize}
Hence (linear-)algebraic $\lambda$-calculi can be seen as a platform for various applications, ranging from algebraic computation, probabilistic computation, quantum computation and resource-aware computation.

\subsection{The language}
The language we consider in this paper will be called the {\em
vectorial $\lambda$-calculus}, denoted by $\lvec$. It is
derived from {\it Lineal}~\cite{ArrighiDowekLMCS17}. This language
admits the regular constructs of $\lambda$-calculus: variables
$x,y,\ldots$, $\lambda$-abstractions $\lambda x.{\ve s}$ and application
$(\ve{s})\,\ve{t}$. But it also admits linear combinations of terms:
$\ve{0}$, ${\ve s}+{\ve t}$ and $\alpha\cdot {\ve s}$ are terms, where the
scalar $\alpha$ ranges over a ring. As in~\cite{ArrighiDowekLMCS17}, it behaves in a
call-by-value oriented manner, in the sense that $(\lambda x.{\ve r})\,({\ve s}+{\ve t})$ first reduces to $(\lambda x.{\ve
r})\,{\ve s}+(\lambda x.{\ve r})\,{\ve t}$ until {\em basis terms} (i.e.~values) are reached, at which point beta-reduction applies.

The set of the normal forms of the terms can then be interpreted as a module and the
term $(\lambda x.{\ve r})\,{\ve s}$ can be seen as the application of
the linear operator $(\lambda x.{\ve r})$ to the vector~${\ve s}$.
The goal of this paper is to give a formal account of linear operators
and vectors at the level of the type system.

\subsection{Our contributions: The types}

Our goal is to characterize the vectoriality of the system of terms,
as summarized by the slogan:
\begin{quote}
  If $\ve s:T$ and $\ve t:R$ then $\alpha\cdot\ve s + \beta\cdot\ve t
  : \alpha\cdot T + \beta\cdot R$.
\end{quote}

\noindent
In the end we achieve a type system such that:
\begin{itemize}
  \item The typed language features a slightly weakened subject
    reduction property (Theorem~\ref{thm:subjectreduction}).
  \item The typed language features strong normalization (cf.~Theorem~\ref{th:SN}).
  \item In general, if $\ve t$ has type $\sum_i\alpha_i\cdot U_i$, then it must reduce to a $\ve
    t'$ of the form $\sum_{ij}\beta_{ij}\cdot\ve b_{ij}$, where: the $\ve b_{ij}$'s
    are basis terms of unit type $U_i$, and $\sum_{ij} \beta_{ij}=\alpha_{i}$. (cf.~Theorem~\ref{thm:termcharact}).
  \item In particular finite vectors and matrices and tensorial products can be encoded within \lvec. In this case, 
    the type of the encoded expressions coincides with the result of the expression (cf.~Theorem~\ref{thm:matrixsound}).
\end{itemize}

\noindent Beyond these formal results, this work constitutes a first attempt to
describe a natural type system with type constructs $\alpha\cdot$ and $+$ and
to study their behaviour.

\subsection{Directly related works}
This paper is part of a research
path~\cite{tonder04lambda,AltenkirchGrattageLICS05,ArrighiDowekLMCS17,ValironQPL10,BuirasDiazcaroJaskelioffLSFA11,ArrighiDiazcaroLMCS12,DiazcaroPetitWoLLIC12}
to design a typed language where terms can be linear combinations of terms
(they can be interpreted as probability distributions or
quantum superpositions of data and programs) and where the types
capture some of this additional structure
(they provide the propositions for a probabilistic or quantum logic via Curry-Howard).

Along this path, a first step was accomplished in~\cite{ArrighiDiazcaroLMCS12} with scalars in the type system. If $\alpha$ is a scalar and
$\Gamma\vdash\ve t:T$ is a valid sequent, then $\Gamma\vdash\alpha\cdot\ve t:\alpha\cdot T$ is a valid sequent. When the scalars are taken
to be positive real numbers, the developed language actually provides a static analysis
tool for {\it probabilistic} computation. However, it fails to address
the following issue: without sums but with negative numbers, the term
representing ``${\bf true}-{\bf false}$'', namely $\lambda x.\lambda y.x-\lambda x.\lambda y.y$, is typed with $0\cdot(X\to(X\to
X))$, a type which fails to exhibit the fact that we have a superposition of terms.

A second step was accomplished in~\cite{DiazcaroPetitWoLLIC12} with sums in the
type system. In this case, if $\Gamma\vdash
\ve s:S$ and $\Gamma\vdash\ve t:T$ are two valid sequents, then $\Gamma\vdash
\ve s+\ve t:S+T$ is a valid sequent. However, the language considered is only the {\it
additive} fragment of {\it Lineal}, it leaves scalars out of the picture. For instance,  $\lambda x.\lambda y.x-\lambda x.\lambda y.y$, does not have a type, due to its minus sign. Each of these two contributions required renewed, careful and lengthy proofs about their type systems, introducing new techniques.

The type system we propose in this paper builds upon these two approaches: it includes both scalars and sums of types, thereby reflecting the vectorial structure of the terms at the level of types.
Interestingly, combining the two separate features of
\cite{ArrighiDiazcaroLMCS12,DiazcaroPetitWoLLIC12} raises subtle novel
issues, which we identify and discuss (cf. Section~\ref{sec:vectorial}). 
Equipped with those two vectorial type constructs, the type system is indeed able to capture some fine-grained information about the vectorial structure of the terms. Intuitively, this means keeping track of both the `direction' and the `amplitude' of the terms.  

A preliminary version of this paper has appeared in \cite{ArrighiDiazcaroValironDCM11}.

\subsection{Plan of the paper}
In Section~\ref{sec:language}, we present the language. We discuss the
differences with the original language {\it Lineal}~\cite{ArrighiDowekLMCS17}.
In Section~\ref{sec:vectorial}, we explain the problems arising from the
possibility of having linear combinations of types, and elaborate a type system
that addresses those problems.  Section~\ref{sec:sr} is devoted to subject
reduction. We first say why the standard formulation of subject reduction does
not hold. Second we state a slightly weakened notion of the subject reduction
theorem, and we prove this result.  In Section~\ref{sec:SN}, we prove strong
normalisation.  Finally we close the paper in Section~\ref{sec:examples} with
theorems about the information brought by the type judgements, both in the
general and the finitary cases (matrices and vectors).

\section{The terms}\label{sec:language}

We consider the untyped language \lvec\ described in Figure~\ref{fig:Vec}. It
is based on {\em Lineal} \cite{ArrighiDowekLMCS17}: terms come in two flavours,
basis terms which are the only ones that will substitute a variable in a
$\beta$-reduction step, and general terms. We use Krivine's
notation~\cite{Krivine90} for function application: The term $(\ve s)\,\ve t$
passes the argument $\ve t$ to the function $\ve s$.

In addition to $\beta$-reduction, there are fifteen rules stemming from the
oriented axioms of vector spaces \cite{ArrighiDowekLMCS17}, specifying the
behaviour of sums and products. We divide the rules in groups: Elementary (E),
Factorisation (F), Application (A) and the Beta reduction (B).  Essentially,
the rules E and F, presented in \cite{arrighi05}, consist in capturing the
equations of vector spaces in an oriented rewrite system. For example, $0\cdot
\ve s$ reduces to $\ve 0$, as $0\cdot\ve s = \ve 0$ is valid in vector spaces.
It should also be noted that this set of algebraic rule is confluent, and does
not introduce loops. In particular, the two rules stating $\alpha\cdot(\ve
t+\ve r)\to \alpha\cdot\ve t + \alpha\cdot\ve r$ and $\alpha\cdot\ve t +
\beta\cdot\ve t\to (\alpha+\beta)\cdot\ve t$ are not inverse one of the other
when $\ve r = \ve t$. Indeed,
\[
  \alpha\cdot(\ve t+\ve t) 
  \to
  \alpha\cdot\ve t+\alpha\cdot\ve t
  \to
  (\alpha + \alpha)\cdot\ve t
\]
but not the other what around.

The group of A rules formalize the fact that a general term $\ve t$ is thought
of as a linear combination of terms $\alpha\cdot\ve r+\beta\cdot\ve r'$ and the
face that the application is distributive on the left {\em and} on the right.
When we apply $\ve s$ to such a superposition, $(\ve s)~\ve t$ reduces to
$\alpha\cdot(\ve s)~\ve r + \beta\cdot(\ve s)~\ve r'$.  The term $\ve 0$ is the
empty linear combination of terms, explaining the last two rules of Group A.

Terms are considered modulo associativity and commutativity of the
operator~$+$, making the reduction into an {\em AC-rewrite system}
\cite{JouannaudKirchnerSIAM86}.  Scalars (notation $\alpha,\beta,\gamma,\dots$)
form a ring $(\Sc,+,\times)$, where the scalar $0$ is the unit of the addition
and $1$ the unit of the multiplication. We use the shortcut notation
$\ve{s}-\ve{t}$ in place of $\ve{s}+(-1)\cdot\ve{t}$.  Note that although the
typical ring we consider in the examples is the ring of complex numbers, the
development works for any ring: the ring of integer $\mathbb{Z}$, the finite
ring $\mathbb{Z}/2\mathbb{Z}$\ldots

The set of free variables of a term is defined as usual: the only operator
binding variables is the $\lambda$-abstraction.  The operation of substitution
on terms (notation~$\ve{t}[{\ve b}/x]$) is defined in the usual way for the
regular $\lambda$-term constructs, by taking care of variable renaming to avoid
capture. For a linear combination, the substitution is defined as follows:
$(\alpha\cdot\ve t+\beta\cdot\ve r)[{\ve b}/x]=\alpha\cdot\ve t[{\ve
b}/x]+\beta\cdot\ve r[{\ve b}/x]$.

Note that we need to choose a reduction strategy. For example, the term
$(\lambda x.(x)\ x)$ $(y +z)$ cannot reduce to both $(\lambda
x.(x)\,x)\,y+(\lambda x.(x)\,x)\,z$ and $(y+z)\,(y+z)$. Indeed, the former
normalizes to $(y)\,y+(z)\,z$ whereas the latter normalizes to
$(y)\,z+(y)\,y+(z)\,y+(z)\,z$; which would break confluence. As in
\cite{ArrighiDowekLMCS17,ArrighiDiazcaroLMCS12,DiazcaroPetitWoLLIC12}, we
consider a call-by-value reduction strategy: The argument of the application is
required to be a base term, cf.\ Group B.

\begin{figure}
  \centering
  {\fbox{\begin{tabular}{@{}c@{}}
    $$\begin{array}[t]{@{}l@{\hspace{1.5cm}}r@{\ ::=\quad}l@{}}
      \text{\em Terms:} & \ve{r},\ve{s},\ve{t},\ve{u} & \ve{b}~|~(\ve{t})~\ve{r}~|~\ve{0}~|~\alpha\cdot\ve{t}~|~\ve{t}+\ve{r}\\
      \text{\em Basis terms:} & \ve{b} & x~|~\lambda x.\ve{t}
    \end{array}$$
    \\
    \begin{tabular}{p{3.8cm}p{3.9cm}p{4cm}}
      \emph{Group E:}

      $0\cdot\ve{t}\to\ve{0}$

      $1\cdot\ve{t}\to\ve{t}$

      $\alpha\cdot\ve{0}\to\ve{0}$

      $\alpha\cdot (\beta\cdot\ve{t})\to(\alpha\times\beta)\cdot\ve{t}$

      $\alpha\cdot (\ve{t}+\ve{r})\to\alpha\cdot\ve{t}+\alpha\cdot\ve{r}$

      &
      \emph{Group F:}

      $\alpha\cdot\ve{t}+\beta\cdot\ve{t}\to(\alpha+\beta)\cdot\ve{t}$

      $\alpha\cdot\ve{t}+\ve{t}\to(\alpha+1)\cdot\ve{t}$

      $\ve{t}+\ve{t}\to(1+1)\cdot\ve{t}$

      $\ve{t}+\ve{0}\to\ve{t}$
      \medskip

      \emph{Group B:}

      $(\lambda x.\ve{t})~\ve{b}\to\ve{t}[\ve{b}/x]$
      &
      \emph{Group A:}

      $(\ve{t}+\ve{r})~\ve{u}\to(\ve{t})~\ve{u}+(\ve{r})~\ve{u}$

      $(\ve{t})~(\ve{r}+\ve{u})\to(\ve{t})~\ve{r}+(\ve{t})~\ve{u}$

      $(\alpha\cdot\ve{t})~\ve{r}\to\alpha\cdot (\ve{t})~\ve{r}$

      $(\ve{t})~(\alpha\cdot\ve{r})\to\alpha\cdot (\ve{t})~\ve{r}$

      $(\ve{0})~\ve{t}\to \ve{0}$

      $(\ve{t})~\ve{0}\to \ve{0}$
    \end{tabular}
    \\[12ex]
    \prooftree\ve t\to\ve r
    \justifies\alpha\cdot\ve t\to\alpha\cdot\ve r
  \endprooftree
  \hfill
  \prooftree\ve t\to\ve r
  \justifies\ve u+\ve t\to \ve u+\ve r
\endprooftree
\hfill
\prooftree\ve t\to \ve r
\justifies(\ve u)~\ve t\to(\ve u)~\ve r
\endprooftree
\hfill
\prooftree\ve t\to \ve r
\justifies(\ve t)~\ve u\to(\ve r)~\ve u
\endprooftree
\hfill
\prooftree\ve t\to \ve r
\justifies\lambda x.\ve t\to\lambda x.\ve r
\endprooftree
\end{tabular}
}}
\caption{Syntax, reduction rules and context rules of \lvec.}
\label{fig:Vec}
\end{figure}

\subsection{Relation to {\it Lineal}}
Although strongly inspired from {\it Lineal}, the language $\lvec$ is closer
to~\cite{AssafDiazcaroPerdrixTassonValironLMCS14,ArrighiDiazcaroLMCS12,DiazcaroPetitWoLLIC12}.
Indeed, {\it Lineal} considers some restrictions on the reduction rules, for
example $\alpha\cdot\ve{t}+\beta\cdot\ve{t}\to(\alpha+\beta)\cdot\ve{t}$ is
only allowed when $\ve t$ is a closed normal term.  These restrictions are
enforced to ensure confluence in the untyped setting. Consider the following
example. Let $\ve Y_\ve b=(\lambda x.(\ve b+(x)~x))~\lambda x.(\ve b+(x)~x)$.
Then $\ve Y_\ve b$ reduces to $\ve b+\ve Y_\ve b$. So the term $\ve Y_\ve b-\ve
Y_\ve b$ reduces to $\ve 0$, but also reduces to $\ve b+\ve Y_\ve b-\ve Y_\ve
b$ and hence to $\ve b$, breaking confluence. The above restriction forbids the
first reduction, bringing back confluence. In our setting we do not need it
because $\ve Y_\ve b$ is not well-typed.  If one considers a typed language
enforcing strong normalisation, one can waive many of the restrictions and
consider a more canonical set of rewrite rules
\cite{AssafDiazcaroPerdrixTassonValironLMCS14,ArrighiDiazcaroLMCS12,DiazcaroPetitWoLLIC12}.
Working with a type system enforcing strong normalisation (as shown in
Section~\ref{sec:SN}), we follow this approach.

\subsection{Booleans in the vectorial $\lambda$-calculus}
We claimed in the introduction that 
the design of {\it Lineal} was motivated by quantum computing; in this
section we develop this analogy.

Both in $\lvec$ and in quantum computation one can interpret
the notion of booleans. In the former we can consider the usual
booleans $\lambda x.\lambda y.x$ and $\lambda x.\lambda y.y$ whereas
in the latter we consider the regular quantum bits $\true=\ket0$ and
$\false=\ket1$.

In $\lvec$, a representation of ${\it if}~{\ve r}~{\it then}~{\ve s}~{\it
else}~{\ve t}$ needs to take into account the special relation between sums and
applications. We cannot directly encode this test as the usual $(({\ve
r})\,{\ve s})\,{\ve t}$. Indeed, if ${\ve r}$, ${\ve s}$ and ${\ve t}$ were
respectively the terms $\true$, ${\ve s}_1+{\ve s}_2$ and ${\ve t}_1+{\ve
t}_2$, the term $(({\ve r})\,{\ve s})\,{\ve t}$ would reduce to $((\true)\,{\ve
s}_1)\,{\ve t}_1 + ((\true)\,{\ve s}_1)\,{\ve t}_2 + ((\true)\,{\ve s}_2)\,{\ve
t}_1 +((\true)\,{\ve s}_2)\,{\ve t}_2$, then to $2\cdot{\ve s}_1 + 2\cdot{\ve
s}_2$ instead of ${\ve s}_1 + {\ve s}_2$.  We need to ``freeze'' the
computations in each branch of the test so that the sum does not distribute
over the application. For that purpose we use the well-known notion of {\em
thunks} \cite{ArrighiDowekLMCS17}: we encode the test as $\cocanon{(({\ve
r})\,\canon{{\ve s}})\,\canon{{\ve t}}}$, where $\canon{-}$ is the term
$\lambda f.-$ with $f$ a fresh, unused term variable and where $\cocanon{-}$ is
the term $(-)\lambda x.x$. The former ``freezes'' the linearity while the
latter ``releases'' it. Then the term ${\it if}~\true~{\it then}~({\ve
s}_1+{\ve s}_2)~{\it else}~({\ve t}_1+{\ve t}_2)$ reduces to the term ${\ve
s}_1+{\ve s}_2$ as one could expect. Note that this test is linear, in the
sense that the term ${\it if}~(\alpha\cdot\true+\beta\cdot\false)~{\it
then}~{\ve s}~{\it else}~{\ve t}$ reduces to $\alpha\cdot{\ve s} +
\beta\cdot{\ve t}$.

This is similar to the {\it quantum test} that can be found \eg
in~\cite{tonder04lambda,AltenkirchGrattageLICS05}.  Quantum computation deals
with complex, linear combinations of terms, and a typical computation is run by
applying linear unitary operations on the terms, called {\em gates}.  For
example, the Hadamard gate {\bf H} acts on the space of booleans spanned by
$\true$ and $\false$. It sends ${\true}$ to $\frac1{\sqrt2}({\true}+{\false})$
and ${\false}$ to $\frac1{\sqrt2}({\true}-{\false})$.  If $x$ is a quantum bit,
the value $({\bf H})\,x$ can be represented as the quantum test
\[ 
  ({\bf H})\,x\quad{:}{=}\quad{ {\it
  if}~x~{\it then}~\frac1{\sqrt2}(\true+\false)~{\it
  else}~\frac1{\sqrt2}(\true-\false)}.
\]
As developed in \cite{ArrighiDowekLMCS17}, one can simulate this operation in
$\lvec$ using the test construction we just described: 
\[
  \cocanon{({\bf H})\,x}\quad{:}{=}\quad\left\{
    \left((x)\,\left[\frac1{\sqrt2}\cdot\true+\frac1{\sqrt2}\cdot\false\right]\right)\,
    \left[\frac1{\sqrt2}\cdot\true-\frac1{\sqrt2}\cdot\false\right]
  \right\}.
\]
Note that the thunks are necessary: without thunks the term 
\[
  \left((x)\,
  \left(\frac1{\sqrt2}\cdot\true+\frac1{\sqrt2}\cdot\false\right)\right)\,
  \left(\frac1{\sqrt2}\cdot\true-\frac1{\sqrt2}\cdot\false\right)
\]
would reduce to the term 
\[
  \frac12(((x)\,\true)\,\true + ((x)\,\true)\,\false +
  ((x)\,\false)\,\true + ((x)\,\false)\,\false),
\]
which is fundamentally different from the term ${\bf H}$ we are trying to
emulate.

With this procedure we can ``encode'' any matrix. If the space is of some
general dimension $n$, instead of the basis elements $\true$ and $\false$ we
can choose for $i=1$ to $n$ the terms $\lambda x_1.\cdots.\lambda x_n.x_i$'s to
encode the basis of the space. We can also take tensor products of qubits. We
come back to these encodings in Section~\ref{sec:examples}.

\section{The type system}\label{sec:vectorial}

This section presents the core definition of the paper: the vectorial type system.

\subsection{Intuitions}

Before diving into the technicalities of the definition, we discuss the
rationale behind the construction of the type-system.

\subsubsection{Superposition of types}
We want to incorporate the notion of scalars in the type system.  If $A$ is a
valid type, the construction $\alpha\cdot A$ is also a valid type and if the
terms $\ve{s}$ and $\ve{t}$ are of type $A$, the term
$\alpha\cdot\ve{s}+\beta\cdot\ve{t}$ is of type $(\alpha+\beta)\cdot A$. This
was achieved in~\cite{ArrighiDiazcaroLMCS12} and it allows us to distinguish
between the functions $\lambda x.(1\cdot x)$ and $\lambda x.(2\cdot x)$: the
former is of type $A\to A$ whereas the latter is of type $A\to (2\cdot A)$.

The terms $\true$ and $\false$ can be typed in the usual way with $\Bool =
X\to(X\to X)$, for a fixed type $X$.  So let us consider the term
$\frac1{\sqrt2}\cdot(\true-\false)$. Using the above addition to the type
system, this term should be of type $0\cdot\Bool$, a type which fails to
exhibit the fact that we have a superposition of terms.  For instance, applying
the Hadamard gate to this term produces the term $\false$ of type $\Bool$: the
norm would then jump from~$0$~to~$1$.  This time, the problem comes from the
fact that the type system does not keep track of the ``direction'' of a term.  

To address this problem we must allow sums of types.  For instance, provided
that $\True=X\to(Y\to X)$ and $\False=X\to(Y\to Y)$, we can type the term
$\frac1{\sqrt2}\cdot(\true-\false)$ with $\frac{\sqrt2}2\cdot(\True-\False)$,
which has $L_2$-norm $1$, just like the type of $\false$ has norm one.

At this stage the type system is able to type the term ${\bf H}=\lambda
x.\cocanon{((x)\,\canon{\frac1{\sqrt2}\cdot\true+\frac1{\sqrt2}\cdot\false})\,
\canon{\frac1{\sqrt2}\cdot\true-\frac1{\sqrt2}\cdot\false}}$.  Indeed, remember
that the thunk construction $\canon{-}$ is simply $\lambda f.(-)$ where $f$ is
a fresh variable and that $\cocanon{-}$ is $(-)\lambda x.x$. So whenever $\ve
t$ has type $A$, $\canon{\ve t}$ has type ${\bf I}\to A$ with ${\bf I}$ an
identity type of the form $Z\to Z$, and $\cocanon{\ve t}$ has type $A$ whenever
$\ve t$ has type ${\bf I}\to A$. The term $\bf H$ can then be typed with
$(({\bf I}\to\frac1{\sqrt 2}.(\True+\False)) \to ({\bf I}\to\frac1{\sqrt
2}.(\True-\False)) \to {\bf I}\to T)\to T$, where $T$ any fixed type.

Let us now try to type the term $({\bf H})\,\true$.  This is possible by taking
$T$ to be $\frac1{\sqrt 2}\cdot(\True+\False)$. But then, if we want to type
the term $({\bf H})\,\false$, $T$ needs to be equal to $\frac1{\sqrt
2}\cdot(\True-\False)$. It follows that we cannot type the term $({\bf
H})\,(\frac{1}{\sqrt2}\cdot\true + \frac1{\sqrt2}\cdot\false)$ since there is
no possibility to conciliate the two constraints on $T$.

To address this problem, we need a forall construction in the type system,
making it {\em \`a la System~F}.  The term ${\bf H}$ can now be typed with
$\forall T.(({\bf I}\to\frac1{\sqrt 2}\cdot(\True+\False)) \to ({\bf
I}\to\frac1{\sqrt 2}\cdot(\True-\False)) \to {\bf I} \to T)\to T$ and the types
$\True$ and $\False$ are updated to be respectively $\forall XY.X\to(Y\to X)$
and $\forall XY.X\to(Y\to Y)$.  The terms $({\bf H})\,\true$ and $({\bf
H})\,\false$ can both be well-typed with respective types $\frac1{\sqrt
2}\cdot(\True+\False)$ and $\frac1{\sqrt 2}\cdot(\True-\False)$, as expected.

\subsubsection{Type variables, units and general types}
\label{sec:introH}
Because of the call-by-value strategy, variables must range over types that are
not linear combination of other types, i.e.\ {\it unit types}.  To illustrate
this necessity, consider the following example. Suppose we allow variables to
have scaled types, such as $\alpha\cdot U$. Then the term $\lambda x.x+y$ could
have type $(\alpha\cdot U)\to \alpha\cdot U+V$ (with $y$ of type $V$). Let $\ve
b$ be of type $U$, then $(\lambda x.x+y)~(\alpha\cdot\ve b)$ has type
$\alpha\cdot U+V$, but then
\[
  (\lambda x.x+y)~(\alpha\cdot\ve b)\to
  \alpha\cdot(\lambda x.x+y)~\ve b\to
  \alpha\cdot(\ve b+y)\to
  \alpha\cdot\ve b+\alpha\cdot y\,,
\]
which is problematic since the type $\alpha\cdot U+V$ does not reflect such a
superposition.  Hence, the left side of an arrow will be required to be a unit
type. This is achieved by the grammar defined in Figure~\ref{fig:types}.

Type variables, however, do not always have to be unit type.  Indeed, a forall
of a general type was needed in the previous section in order to type the term
$\ve H$.  But we need to distinguish a general type variable from a unit type
variable, in order to make sure that only unit types appear at the left of
arrows.  Therefore, we define two sorts of type variables: the variables
$\varu{X}$ to be replaced with unit types, and $\vara{X}$ to be replaced with
any type (we use just $X$ when we mean either one). The type $\varu{X}$ is a
unit type whereas the type $\vara{X}$ is not.

In particular, the type $\True$ is now
$\forall\varu{XY}.\varu{X}\to\varu{Y}\to\varu{X}$, the type $\False$ is
$\forall\varu{XY}.\varu{X}\to\varu{Y}\to\varu{Y}$ and the type of ${\ve H}$ is
\[
  \forall\vara{X}.\left(\left({\bf I}\to\frac1{\sqrt 2}\cdot(\True+\False)\right)
  \to \left({\bf I}\to\frac1{\sqrt 2}\cdot(\True-\False)\right) \to {\bf I} \to \vara{X}\right)\to \vara{X}.
\]
Notice how the left sides of all arrows remain unit types.

\subsubsection{The term $\ve0$}\label{sec:term0}
The term $\ve 0$ will naturally have the type $0\cdot T$, for any inhabited
type $T$ (enforcing the intuition that the term $\ve 0$ is essentialy a normal
form of programs of the form $\ve t-\ve t$).

We could also consider to add the equivalence $R+0\cdot T\equiv R$ as in
\cite{ArrighiDiazcaroLMCS12}.  However, consider the following example. Let
$\lambda x.x$ be of type $U\to U$ and let $\ve t$ be of type $T$. The term
$\lambda x.x + \ve t - \ve t$ is of type $(U\to U) + 0\cdot T$, that is, $(U\to
U)$. Now choose $\ve b$ of type $U$: we are allowed to say that $(\lambda x.x +
\ve t - \ve t)\,\ve b$ is of type $U$. This term reduces to $\ve b + (\ve
t)\,\ve b - (\ve t)\,\ve b$. But if the type system is reasonable enough, we
should at least be able to type $(\ve t)\,\ve b$. However, since there is no
constraints on the type $T$, this is difficult to enforce.

The problem comes from the fact that along the typing of $\ve t - \ve t$, the
type of $\ve t$ is lost in the equivalence $(U\to U)+0\cdot T\equiv U\to U$.
The only solution is to not discard $0\cdot T$, that is, to not equate
$R+0\cdot T$ and $R$.

\subsection{Formalisation}

We now give a formal account of the type system: we first describe the language
of types, then present the typing rules.

\subsubsection{Definition of types}
Types are defined in Figure~\ref{fig:types} (top). They come in two flavours:
{\em unit types} and general types, that is, linear combinations of types.
\begin{figure}
  \centering
  {\fbox{\begin{tabular}{@{}c@{}}
    $$\begin{array}[t]{l@{\hspace{1.5cm}}r@{\ ::=\quad}l}
      \text{\em Types:} & T,R,S & U~|~\alpha\cdot T~|~T+R~|~\vara{X}\\
      \text{\em Unit types:} & U,V,W & \varu{X}~|~U\to T~|~\forall \varu{X}.U~|~\forall \vara{X}.U
    \end{array}$$
    \\
    \\
    $$\begin{array}[t]{r@{~\equiv~}l@{\hspace{1.5cm}}r@{~\equiv~}l}
      1\cdot T & T							 & 	\alpha\cdot T+\beta\cdot T & (\alpha+\beta)\cdot T\\
      \alpha\cdot(\beta\cdot T) & (\alpha\times\beta)\cdot T & 		T+R & R+T\\
      \alpha\cdot T+\alpha\cdot R	&\alpha\cdot (T+R) 			 & 	T+(R+S) & (T+R)+S
    \end{array}$$
    \\
    \\
    $$
    \quad\prooftree
    \justifies\Gamma, x:{U}\vdash x:{U}
    \using ax
  \endprooftree
  \hspace{1.8cm}
  \prooftree\Gamma\vdash\ve{t}: T
  \justifies\Gamma\vdash\ve{0}: 0\cdot T
  \using 0_I
\endprooftree
\hspace{1.8cm}
\prooftree\Gamma, x:{U} \vdash\ve{t}: T
\justifies\Gamma \vdash \lambda x.\ve{t}:{U}\to T
\using\to_I
\endprooftree\quad$$
\\
\\
$$
\prooftree\Gamma \vdash\ve{t}:\sui{n}\alpha_i\cdot\forall\vec{X}.(U\to T_i) \qquad \Gamma\vdash\ve{r}:\suj{m}\beta_j\cdot U[\vec{A}_j/\vec{X}]
\justifies\Gamma \vdash(\ve{t})~\ve{r}:\sui{n}\suj{m} \alpha_i\times\beta_j\cdot {T_i[\vec{A}_j/\vec{X}]}
\using\to_E
\endprooftree$$
\\
\\
$$
\prooftree\Gamma\vdash\ve{t}: \sui{n}\alpha_i\cdot U_i\quad{X\notin\FV{\Gamma}}
\justifies\Gamma\vdash\ve{t}:\sui{n}\alpha_i\cdot\forall X.U_i
\using \forall_{I}
\endprooftree
\hspace{2cm}
\prooftree\Gamma\vdash\ve{t}: \sui{n}\alpha_i\cdot\forall X.U_i
\justifies\Gamma\vdash\ve{t}: \sui{n}\alpha_i\cdot U_i[A/X]
\using \forall_{E}
\endprooftree$$
\\
\\
$$
\quad\prooftree\Gamma\vdash\ve{t}: T
\justifies\Gamma\vdash\alpha\cdot\ve{t}:\alpha\cdot T
\using\alpha_I
\endprooftree
\hspace{1cm}
\prooftree\Gamma\vdash\ve{t}: T\qquad\Gamma\vdash\ve{r}: R
\justifies\Gamma\vdash\ve{t}+\ve{r}: T+R
\using +_I
\endprooftree
\hspace{1cm}
\prooftree\Gamma\vdash\ve{t}: T\qquad T\equiv R
\justifies\Gamma\vdash\ve{t}: R
\using\equiv
\endprooftree\quad
$$
\end{tabular}}}
\caption{Types and typing rules of \lvec. We use $X$ when we do not want to
  specify if it is $\varu{X}$ or $\vara{X}$, that is, unit variables or general
  variables respectively. In $T[A/X]$, if $X=\varu{X}$, then $A$ is a unit
  type, and if $X=\vara{X}$, then $A$ can be any type. We also may write
  $\forall_\varu{I}$ and $\forall_\vara{I}$ (resp. $\forall_\varu{E}$ and
  $\forall_\vara{E}$) when we need to specify which kind of variable is being
used.}
\label{fig:types}
\end{figure}
Unit types include all types of
\emph{System~F}~\cite[Ch.~11]{GirardLafontTaylor89} and intuitively they are
used to type basis terms.  The arrow type admits only a unit type in its
domain. This is due to the fact that the argument of a $\lambda$-abstraction
can only be substituted by a basis term, as discussed in
Section~\ref{sec:introH}.  As discussed before, the type system features two
sorts of variables: unit variables $\varu{X}$ and general variables $\vara{X}$.
The former can only be substituted by a unit type whereas the latter can be
substituted by any type. We use the notation $X$ when the type variable is
unrestricted.  The substitution of~$\varu{X}$ by~$U$ (resp. $\vara{X}$ by~$S$)
in~$T$ is defined as usual and is written~$T[U/\varu{X}]$ (resp.
$T[S/\vara{X}]$). We use the notation $T[A/X]$ to say: ``if $X$ is a unit
variable, then $A$ is a unit type and otherwise $A$ is a general type''.  In
particular, for a linear combination, the substitution is defined as follows:
$(\alpha\cdot T+\beta\cdot R)[A/X]=\alpha\cdot T[A/X]+\beta\cdot R[A/X]$.  We
also use the vectorial notation~$T[\vec{A}/\vec{X}]$
for~$T[A_1/X_1]\cdots[A_n/X_n]$ if~$\vec{X}=X_1,\dots,X_n$
and~$\vec{A}=A_1,\dots,A_n$, and also $\forall \vec X$ for $\forall X_1\dots
X_n=\forall X_1.\dots.\forall X_n$.

The equivalence relation $\equiv$ on types is defined as a congruence.  Notice
that this equivalence makes the types into a weak module over the scalars: they
almost form a module save from the fact that there is no neutral element for
the addition.  The type $0\cdot T$ is not the neutral element of the addition.

We may use the summation ($\sum$) notation without ambiguity, due to the
associativity and commutativity equivalences of $+$.

\subsubsection{Typing rules}
The typing rules are given also in Figure~\ref{fig:types} (bottom).  Contexts
are denoted by $\Gamma$, $\Delta$, etc. and are defined as sets
$\{x:U,\dots\}$, where $x$ is a term variable appearing only once in the set,
and $U$ is a unit type.  The axiom ($ax$) and the arrow introduction rule
($\to_I$) are the usual ones. The rule ($0_I$) to type the term $\ve 0$ takes
into account the discussion in Section~\ref{sec:term0}. This rule also ensures
that the type of $\ve 0$ is inhabited, discarding problematic types like
$0\cdot \forall X.X$. Any sum of typed terms can be typed using Rule $(+_I)$.
Similarly, any scaled typed term can be typed with $(\alpha_I)$. Rule
$(\equiv)$ ensures that equivalent types can be used to type the same terms.
Finally, the particular form of the arrow-elimination rule ($\to_E$) is due to
the rewrite rules in group~A that distribute sums and scalars over application.
The need and use of this complicated arrow elimination can be illustrated by
the following three examples.

\begin{example}\rm
  Rule $(\to_E)$ is easier to read for trivial linear combinations. It states
  that provided that $\Gamma\vdash \ve s:\forall X.U\to S$ and $\Gamma\vdash
  \ve t:V$, if there exists some type $W$ such that $V=U[W/X]$, then since the
  sequent $\Gamma\vdash \ve s:V\to S[W/X]$ is valid, we also have $\Gamma\vdash
  (\ve s)\,\ve t:S[W/X]$.  Hence, the arrow elimination here performs an arrow
  and a forall elimination at the same time.
\end{example}

\begin{example}\rm
  Consider the terms $\ve b_1$ and $\ve b_2$, of respective types $U_1$ and
  $U_2$. The term $\ve b_1 + \ve b_2$ is of type $U_1+U_2$. We would reasonably
  expect the term $(\lambda x.x)\,(\ve b_1 + \ve b_2)$ to also be of type $U_1
  + U_2$. This is the case thanks to Rule $(\to_E)$. Indeed, type the term
  $\lambda x.x$ with the type $\forall X.X\to X$ and we can now apply the rule.
  Notice that we could not type such a term unless we eliminate the forall
  together with the arrow.
\end{example}

\begin{example}\label{ex:3}\rm
  A slightly more involved example is the projection of a pair of elements. It
  is possible to encode in {\em System F} the notion of pairs and projections:
  $\langle \ve b, \ve c\rangle = \lambda x.((x)~\ve b)~\ve c$, $\langle \ve b',
  \ve c'\rangle = \lambda x.((x)~\ve b')~\ve c'$, $\pi_1 = \lambda
  x.(x)~(\lambda y.\lambda z.y)$ and $\pi_2 = \lambda x.(x)~(\lambda y.\lambda
  z.z)$. Provided that $\ve b$, $\ve b'$, $\ve c$ and $\ve c'$ have respective
  types $U$, $U'$, $V$ and $V'$, the type of $\langle \ve b, \ve c\rangle$ is
  $\forall X.(U\to V\to X)\to X$ and the type of $\langle \ve b', \ve
  c'\rangle$ is $\forall X.(U'\to V'\to X)\to X$. The term $\pi_1$ and $\pi_2$
  can be typed respectively with $\forall XYZ.((X\to Y\to X)\to Z)\to Z$ and
  $\forall XYZ.((X\to Y\to Y)\to Z)\to Z$.  The term $(\pi_1 + \pi_2)\,(\langle
  \ve b, \ve c\rangle + \langle \ve b', \ve c'\rangle)$ is then typable of type
  $U+U'+V+V'$, thanks to Rule $(\to_E)$. Note that this is consistent with the
  rewrite system, since it reduces to $\ve b + \ve c + \ve b' + \ve c'$.
\end{example}

\subsection{Example: Typing Hadamard}
In this Section, we formally show how to retrieve the type that was discussed
in Section~\ref{sec:introH}, for the term ${\bf H}$ encoding the Hadamard gate.

Let $\true=\lambda x.\lambda y.x$ and $\false=\lambda x.\lambda y.y$.  It is
easy to check that
\begin{align*}
  &\vdash\true:\forall\varu{X}\varu{Y}.\varu{X}\to\varu{Y}\to\varu{X},\\
  &\vdash\false:\forall\varu{X}\varu{Y}.\varu{X}\to\varu{Y}\to\varu{Y}.
\end{align*}
We also define the following superpositions:
\[
  \ket{+}=\frac{1}{\sqrt{2}}\cdot(\true+\false)
  \qquad\textrm{and}\qquad
  \ket{-}=\frac{1}{\sqrt{2}}\cdot(\true-\false).
\]
In the same way, we define
\begin{align*}
  \boxplus&=\frac{1}{\sqrt{2}}\cdot((\forall\varu{XY}.\varu{X}\to\varu{Y}\to\varu{X})+(\forall\varu{XY}.\varu{X}\to\varu{Y}\to\varu{Y})),
  \\
  \boxminus&=\frac{1}{\sqrt{2}}\cdot((\forall\varu{XY}.\varu{X}\to\varu{Y}\to\varu{X})-(\forall\varu{XY}.\varu{X}\to\varu{Y}\to\varu{Y})).
\end{align*}
Finally, we recall $[\ve t]=\lambda x.\ve t$, where $x\notin\FV{\ve t}$ and
$\{\ve t\}=(\ve t)~I$. So $\cocanon{\canon{\ve t}}\to\ve t$. Then it is easy to
check that $\vdash\canon{\ket{+}}:I\to\boxplus$ and
$\vdash\canon{\ket{-}}:I\to\boxminus$.  In order to simplify the notation, let
$F=(I\to\boxplus)\to(I\to\boxminus)\to (I\to \vara{X})$. Then

\[\prooftree
  \prooftree
  \prooftree
  \prooftree
  \prooftree
  \prooftree
  \justifies x:F\vdash x:F
  \using ax
\endprooftree
\qquad
x:F\vdash\canon{\ket{+}}:I\to\boxplus
\justifies x:F\vdash (x)~\canon{\ket{+}}:(I\to\boxminus)\to(I\to \vara{X})
\using\to_E
        \endprooftree
        \qquad
        x:F\vdash\canon{\ket{-}}:I\to\boxminus
        \justifies x:F\vdash (x)~\canon{\ket{+}}\canon{\ket{-}}:I\to \vara{X}
        \using\to_E
      \endprooftree
      \justifies x:F\vdash\cocanon{(x)~\canon{\ket{+}}\canon{\ket{-}}}:\vara{X}
      \using\to_E
    \endprooftree
    \justifies\vdash\lambda x.\cocanon{(x)~\canon{\ket{+}}\canon{\ket{-}}}:F\to \vara{X}
    \using\to_I
  \endprooftree
  \justifies\vdash\lambda x.\cocanon{(x)~\canon{\ket{+}}\canon{\ket{-}}}:\forall \vara{X}.((I\to\boxplus)\to(I\to\boxminus)\to (I\to \vara{X}))\to \vara{X}
  \using\forall_{\vara{I}}
\endprooftree
\]

\noindent Now we can apply Hadamard to a qubit and get the right type. Let $H$
be the term $\lambda x.\cocanon{(x)~\canon{\ket{+}}\canon{\ket{-}}}$
\medskip

\noindent\scalebox{0.88}{$$\prooftree
  \prooftree\vdash H:\forall \vara{X}.((I\to\boxplus)\to(I\to\boxminus)\to (I\to \vara{X}))\to \vara{X}
  \justifies\vdash H:((I\to\boxplus)\to(I\to\boxminus)\to (I\to\boxplus))\to\boxplus
  \using\forall_{\vara{E}}
\endprooftree
\prooftree
\prooftree\vdash\true:\forall \varu{X}.\forall \varu{Y}.\varu{X}\to \varu{Y}\to \varu{X}
\justifies\vdash\true:\forall \varu{Y}.(I\to\boxplus)\to \varu{Y}\to (I\to\boxplus)
\using\forall_{\varu{E}}
    \endprooftree
    \justifies\vdash\true:(I\to\boxplus)\to(I\to\boxminus)\to (I\to\boxplus)
    \using\forall_{\varu{E}}
  \endprooftree
  \justifies \vdash (H)~\true:\boxplus
  \using\to_E
\endprooftree$$}
\medskip

\noindent Yet a more interesting example is the following.
Let 
\[
  \boxplus_I = \frac 1{\sqrt 2}\cdot(((I\to\boxplus)\to(I\to\boxminus)\to(I\to\boxplus))+((I\to\boxplus)\to(I\to\boxminus)\to(I\to\boxminus)))
\]
That is, $\boxplus$ where the forall have been instantiated. It is easy to
check that $\vdash\ket +:\boxplus_I$. Hence,
\begin{center}
  $$
  \prooftree\vdash H:\forall \vara{X}.((I\to\boxplus)\to(I\to\boxminus)\to (I\to \vara{X}))\to \vara{X}
  \quad
  \vdash\ket +:\boxplus_I
  \justifies \vdash 
  (H)~\ket +:\frac 1{\sqrt 2}\cdot\boxplus+\frac 1{\sqrt 2}\cdot\boxminus
  \using\to_E
\endprooftree
$$
\end{center}
And since $\frac 1{\sqrt 2}\cdot\boxplus+\frac 1{\sqrt
2}\cdot\boxminus\equiv\forall\varu X\varu Y.\varu X\to\varu Y\to\varu X$, we
conclude that
\[
  \vdash(H)~\ket +:\forall\varu X\varu Y.\varu X\to\varu Y\to\varu X.
\]
Notice that $(H)~\ket +\to^*\true$.

\section{Subject reduction}\label{sec:sr}

As we will now explain, the usual formulation of subject reduction is not
directly satisfied. We discuss the alternatives and opt for a weakened version
of subject reduction.
\subsection{Principal types and subtyping alternatives}
Since the terms of $\lvec$ are not explicitly typed, we are bound to have
sequents such as $\Gamma\vdash\ve{t}:T_1$ and $\Gamma\vdash\ve{t}:T_2$ with
distinct types $T_1$ and $T_2$ for the same term $\ve t$.  Using Rules $(+_I)$
and $(\alpha_I)$ we get the valid typing judgement $\Gamma\vdash\alpha\cdot \ve
t+\beta\cdot \ve t:\alpha\cdot T_1+\beta\cdot T_2$. Given that $\alpha\cdot \ve
t+\beta\cdot \ve t$ reduces to $(\alpha+\beta)\cdot \ve t$, a regular subject
reduction would ask for the valid sequent $\Gamma\vdash(\alpha+\beta)\cdot \ve
t:\alpha\cdot T_1+\beta\cdot T_2$.  But since in general we do not have
$\alpha\cdot T_1+\beta\cdot T_2\equiv(\alpha+\beta)\cdot
T_1\equiv(\alpha+\beta)\cdot T_2$, we need to find a way around this.

A first approach would be to use the notion of principal types. However, since
our type system includes {\em System~F}, the usual examples for the absence of
principal types apply to our settings: we cannot rely upon this method.

A second approach would be to ask for the sequent
$\Gamma\vdash(\alpha+\beta)\cdot \ve t:\alpha\cdot T_1+\beta\cdot T_2$ to be
valid. If we force this typing rule into the system, it seems to solve the
issue but then the type of a term becomes pretty much arbitrary: with typing
context $\Gamma$, the term $(\alpha+\beta)\cdot\ve t$ would then be typed with
any combination $\gamma\cdot T_1 + \delta\cdot T_2$, where
$\alpha+\beta=\gamma+\delta$.

The approach we favour in this paper is via a notion of order on types. The
order, denoted with $\sqsupseteq$, will be chosen so that the factorisation
rules make the types of terms smaller. We will ask in particular that
$(\alpha+\beta)\cdot T_1\sqsupseteq\alpha\cdot T_1+\beta\cdot T_2$ and
$(\alpha+\beta)\cdot T_2\sqsupseteq\alpha\cdot T_1+\beta\cdot T_2$ whenever
$T_1$ and $T_2$ are types for the same term.  This approach can also be
extended to solve a second pitfall coming from the rule ${\ve t} + \ve0 \to \ve
t$. Indeed, although $x:\varu{X}\vdash x + \ve0 : \varu{X}+0\cdot T$ is
well-typed for any inhabited $T$, the sequent $x:\varu{X}\vdash
x:\varu{X}+0\cdot T$ is not valid in general. We therefore extend the ordering
to also have $\varu{X}\sqsupseteq \varu{X}+0\cdot T$.

Notice that we are not introducing a subtyping relation with this ordering. For
example, although $\vdash (\alpha+\beta)\cdot\lambda x.\lambda
y.x:(\alpha+\beta)\cdot\forall\varu{X}.\varu{X}\to (\varu{X}\to\varu{X})$ is
valid and $(\alpha+\beta)\cdot\forall\varu{X}.\varu{X}\to
(\varu{X}\to\varu{X})\sqsupseteq \alpha\cdot \forall\varu{X}.\varu{X}\to
(\varu{X}\to\varu{X})+\beta\cdot\forall\varu{X}\varu{Y}.\varu{X}\to(\varu{Y}\to\varu{Y})$,
the sequent $\vdash (\alpha+\beta)\cdot\lambda x.\lambda y.x:\alpha\cdot\forall
\varu{X}.\varu{X}\to
(\varu{X}\to\varu{X})+\beta\cdot\forall\varu{X}\varu{Y}.\varu{X}\to(\varu{Y}\to\varu{Y})$
is not valid.

\subsection{Weak subject reduction}

We define the (antisymmetric) ordering relation $\sqsupseteq$ on types
discussed above as the smallest reflexive transitive and congruent relation
satisfying the rules:
\begin{enumerate}
  \item $(\alpha+\beta)\cdot T\sqsupseteq\alpha\cdot T+\beta\cdot T'$
    if there are $\Gamma,\ve{t}$ such that $\Gamma\vdash\alpha\cdot\ve{t}: \alpha\cdot T$ and $\Gamma\vdash\beta\cdot\ve{t}: \beta\cdot T'$.
  \item $T\sqsupseteq T+0.R$ for any type $R$.
  \item If $T\sqsupseteq R$ and $U\sqsupseteq V$, then $T+S\sqsupseteq
    R+S$,  $\alpha\cdot T\sqsupseteq\alpha\cdot R$, $U\to T\sqsupseteq
    U\to R$ and $\forall X.U\sqsupseteq\forall X.V$.
\end{enumerate}
Note the fact that $\Gamma\vdash \ve t: T$ and $\Gamma\vdash \ve t: T'$ does
not imply that $\beta\cdot T\sqsupseteq \beta\cdot T'$. For instance, although
$\beta\cdot T\sqsupseteq 0\cdot T+\beta\cdot T'$, we do not have $0\cdot
T+\beta\cdot T'\equiv\beta\cdot T'$.
\medskip

Let $R$ be any reduction rule from Figure~\ref{fig:Vec}, and $\to_R$ a one-step
reduction by rule $R$. A weak version of the subject reduction theorem can be
stated as follows.
\begin{theorem}[Weak subject reduction]\label{thm:subjectreduction}
  For any terms $\ve t$, $\ve t'$, any context $\Gamma$ and any type
  $T$, if $\ve{t}\to_R\ve{t}'$ and $\Gamma\vdash \ve t: T$, then:
  \begin{enumerate}
    \item if $R\notin$ Group F, then $\Gamma\vdash\ve t': T$;
    \item if $R\in$ Group F, then
      $\exists S\sqsupseteq T$ such that $\Gamma\vdash\ve t': S$ and
      $\Gamma\vdash\ve{t}: S$.
  \end{enumerate}
\end{theorem}

\subsection{Prerequisites to the proof}\label{sec:prereq}

The proof of Theorem~\ref{thm:subjectreduction} requires some machinery that we
develop in this section. Omitted proofs can be found in~\ref{app:srpre}.

The following lemma gives a characterisation of types as linear combinations of
unit types and general variables. 

\begin{lemma}[Characterisation of types]\label{lem:typecharact}
  For any type $T$ in $\mathcal{G}$, there exist $n,m\in\mathbb{N}$,
  $\alpha_1,\dots,\alpha_n$, $\beta_1,\dots,\beta_m\in\Sc$, distinct unit types
  $U_1,\dots,U_n$ and distinct general variables $\vara{X}_1,\dots,\vara{X}_m$
  such that \pushQED{\qed} 
  \[
    T\equiv\sui{n}\alpha_i\cdot U_i+\suj{m}\beta_j\cdot\vara{X}_j\ .\qedhere
  \]
  \popQED
\end{lemma}

\noindent Our system admits weakening, as stated by the following lemma.

\begin{lemma}[Weakening]\label{lem:weakening}
  Let $\ve t$ be such that $x\not\in\FV{\ve t}$.  Then $\Gamma\vdash\ve t:T$ is
  derivable if and only if $\Gamma,x:U\vdash\ve t:T$ is derivable.
\end{lemma}
\begin{proof}
  By a straightforward induction on the type derivation.
\end{proof}

\begin{lemma}[Equivalence between sums of distinct elements (up to $\equiv$)]\label{lem:equivdistinctscalars}
  Let $U_1,\dots,U_n$ be a set of distinct (not equivalent) unit types, and let
  $V_1,\dots,V_m$ be also a set distinct unit types. If $\sui{n}\alpha_i\cdot
  U_i\equiv\suj{m}\beta_j\cdot V_j$, then $m=n$ and there exists a permutation
  $p$ of $m$ such that $\forall i$, $\alpha_i=\beta_{p(i)}$ and $U_i\equiv
  V_{p(i)}$.
  \qed
\end{lemma}

\begin{lemma}[Equivalences $\forall_I$]\label{lem:equivforall}~
  \begin{enumerate}
    \item\label{it:equivforall1}
      $\sui{n}\alpha_i\cdot U_i\equiv\suj{m}\beta_j\cdot V_j\Leftrightarrow\sui{n}\alpha_i\cdot\forall X.U_i\equiv\suj{m}\beta_j\cdot\forall X.V_j$.
    \item\label{it:equivforall2} $\sui{n}\alpha_i\cdot\forall X.U_i\equiv\suj{m}\beta_j\cdot V_j\Rightarrow\forall V_j,\exists W_j~/~V_j\equiv\forall X.W_j$.
    \item\label{it:equivforall3} $T\equiv R\Rightarrow T[A/X]\equiv R[A/X]$.
      \qed
  \end{enumerate}
\end{lemma}

For the proof of subject reduction, we use the standard strategy developed by
Barendregth~\cite{Barendregt92}\footnote{Note that Barendregth's original proof
  contains a mistake~\cite{stackexchange}. We use the corrected proof proposed
  in~\cite{ArrighiDiazcaroLMCS12}}. It consists in defining a relation betwen
  types of the form $\forall X.T$ and $T$. For our vectorial type system, we
  take into account linear combinations of types

\begin{definition}\label{def:order} For any types $T, R$, any context $\Gamma$ and any term $\ve t$ such that
  $$\prooftree\Gamma\vdash\ve t:T
  \justifies
  \prooftree\vdots
  \justifies\Gamma\vdash\ve t:R
\endprooftree
\endprooftree$$
\begin{enumerate}
  \item if $X\notin\FV{\Gamma}$, write $T\succ^{\ve t}_{X,\Gamma} R$ if either
    \begin{itemize}
      \item $T\equiv\sui{n}\alpha_i\cdot U_i$ and  $R\equiv\sui{n}\alpha_i\cdot\forall X.U_i$,\quad or
      \item $T\equiv\sui{n}\alpha_i\cdot\forall X.U_i$ and $R\equiv \sui{n}\alpha_i\cdot U_i[A/X]$.
    \end{itemize}
  \item if $\V$ is a set of type variables such that $\V\cap\FV{\Gamma}=\emptyset$, we define $\succeq^{\ve t}_{\V,\Gamma}$ inductively by
    \begin{itemize}
      \item If $X\in \V$ and $T\succ^{\ve t}_{X,\Gamma} R$, then $T\succeq^{\ve t}_{\{X\},\Gamma} R$.
      \item If $\V_1,\V_2\subseteq\V$, $T\succeq^{\ve t}_{\V_1,\Gamma} R$ and $R\succeq^{\ve t}_{\V_2,\Gamma} S$, then $T\succeq^{\ve t}_{\V_1\cup\V_2,\Gamma} S$.
      \item If $T\equiv R$, then $T\succeq^{\ve t}_{\V,\Gamma} R$.
    \end{itemize}
\end{enumerate}
\end{definition}

\begin{example}
  Let the following be a valid derivation.
  $$\prooftree
  \prooftree
  \prooftree
  \prooftree
  \prooftree\Gamma\vdash\ve t:T\qquad T\equiv \sui{n}\alpha_i\cdot U_i
  \justifies\Gamma\vdash\ve t:\sui{n}\alpha_i\cdot U_i\qquad \varu{X}\notin\FV{\Gamma}
  \using\equiv
\endprooftree
\justifies\Gamma\vdash\ve t:\sui{n}\alpha_i\cdot\forall\varu{X}.U_i
\using\forall_\varu{I}
      \endprooftree
      \justifies\Gamma\vdash\ve t:\sui{n}\alpha_i\cdot U_i[V/\varu{X}]
      \using\forall_\varu{E}
    \endprooftree
    \vara{Y}\notin\FV{\Gamma}
    \justifies\Gamma\vdash\ve t:\sui{n}\alpha_i\cdot\forall\vara{Y}.U_i[V/\varu{X}]
    \using\forall_\varu{I}
  \endprooftree
  \qquad
  \sui{n}\alpha_i\cdot\forall\vara{Y}.U_i[V/\varu{X}]\equiv R
  \justifies\Gamma\vdash\ve t:R
  \using\equiv
\endprooftree$$
Then $T\succeq^{\ve t}_{\{\varu{X},\vara{Y}\},\Gamma} R$.
\end{example}
\medskip

\noindent Note that this relation is stable under reduction in the following way:

\begin{lemma}[$\succeq$-stability]\label{lem:subjectreductionofrelation}
  If $T\succeq^{\ve t}_{\V,\Gamma} R$, $\ve t\to\ve r$ and $\Gamma\vdash\ve r: T$, then $T\succeq^{\ve r}_{\V,\Gamma} R$.
  \qed
\end{lemma}

\noindent The following lemma states that if two arrow types are ordered, then they are equivalent up to some substitutions.

\begin{lemma}[Arrows comparison]\label{lem:arrowscomp} If
  $V\to R\succeq^{\ve t}_{\V,\Gamma} \forall\vec X.(U\to T)$, then $U\to T\equiv(V\to R)[\vec{A}/\vec{Y}]$, with $\vec Y\notin \FV{\Gamma}$.
  \qedhere
\end{lemma}

Before proving Theorem~\ref{thm:subjectreduction}, we need to prove some basic properties of the system.

\begin{lemma}[Scalars]\label{lem:scalars}
  For any context $\Gamma$, term $\ve t$, type $T$ and scalar $\alpha$, if
  $\Gamma\vdash\alpha\cdot\ve{t}: T$, then there exists a type $R$ such that
  $T\equiv\alpha\cdot R$ and $\Gamma\vdash\ve{t}: R$.  Moreover, if the minimum
  size of the derivation of $\Gamma\vdash\alpha\cdot\ve t:T$ is $s$, then if
  $T=\alpha\cdot R$, the minimum size of the derivation of $\Gamma\vdash\ve
  t:R$ is at most $s-1$, in other case, its minimum size is at most $s-2$.
  \qed
\end{lemma}

\noindent
The following lemma shows that the type for $\ve 0$ is always $0\cdot T$.

\begin{lemma}[Type for zero]\label{lem:termzero}
  Let $\ve t=\ve 0$ or $\ve t=\alpha\cdot\ve 0$, then $\Gamma\vdash\ve t:T$ implies $T\equiv 0\cdot R$.
  \qed
\end{lemma}

\begin{lemma}[Sums]\label{lem:sums}
  If $\Gamma\vdash\ve t+\ve r:S$, then $S\equiv T+R$ with $\Gamma\vdash\ve t:T$ 
  and $\Gamma\vdash\ve r:R$.
  Moreover, if the size of the derivation of $\Gamma\vdash\ve t+\ve r:S$ is $s$, 
  then if $S=T+R$, the minimum sizes of the derivations of $\Gamma\vdash\ve t:T$ 
  and $\Gamma\vdash\ve r:R$ are at most $s-1$, and if $S\neq T+R$, the 
  minimum sizes of these derivations are at most $s-2$.
  \qed
\end{lemma}

\begin{lemma}[Applications]\label{lem:app}
  If $\Gamma\vdash(\ve t)~\ve r:T$, then $\Gamma\vdash\ve t:\sui{n}\alpha_i\cdot\forall\vec{X}.(U\to T_i)$ and $\Gamma\vdash\ve r:\suj{m}\beta_j\cdot U[\vec{A}_j/\vec{X}]$
  where $\sui{n}\suj{m}\alpha_i\times\beta_j\cdot T_i[\vec{A}_j/\vec{X}]\succeq^{(\ve t)\ve r}_{\V,\Gamma} T$ for some $\V$.
  \qed
\end{lemma}

\begin{lemma}[Abstractions]\label{lem:abs} If $\Gamma\vdash\lambda x.\ve t:T$, then $\Gamma,x:U\vdash\ve t:R$ where $U\to R\succeq^{\lambda x.\ve t}_{\V,\Gamma} T$ for some $\V$.
  \qed
\end{lemma}

\noindent
A basis term can always be given a unit type.
\begin{lemma}[Basis terms]\label{lem:basevectors}
  For any context $\Gamma$, type $T$ and basis term $\ve{b}$, if
  $\Gamma\vdash\ve{b}: T$ then there exists a unit type $U$ such
  that $T\equiv U$.
  \qed
\end{lemma}

The final stone for the proof of Theorem~\ref{thm:subjectreduction} is a lemma
relating well-typed terms and substitution.

\begin{lemma}[Substitution lemma]\label{lem:substitution}
  For any term ${\ve t}$, basis term $\ve b$, term variable $x$, context $\Gamma$, types $T$, $U$, type variable $X$ and type $A$, where $A$ is a unit type if $X$ is a unit variables, otherwise $A$ is a general type, we have,
  \begin{enumerate}
    \item\label{it:substitutionTypes} if $\Gamma\vdash\ve{t}: T$, then $\Gamma[A/X]\vdash\ve{t}: T[A/X]$;
    \item\label{it:substitutionTerms} if $\Gamma,x:U\vdash\ve t:T$, $\Gamma\vdash\ve b:U$ then $\Gamma\vdash\ve t[\ve b/x]: T$.
      \qed
  \end{enumerate}
\end{lemma}

\noindent
The proof of subject reduction (Theorem~\ref{thm:subjectreduction}), follows by induction using the previous defined lemmas. It can be foun in full details in~\ref{app:srproof}.

\section{Strong normalisation}\label{sec:SN}

For proving strong normalisation of well-typed terms, we use
reducibility candidates, a well-known method described for example
in~\cite[Ch.~14]{GirardLafontTaylor89}. The technique is adapted to
linear combinations of terms.
Omitted proofs in this section can be found \mbox{in~\ref{app:SN}}.

A {\em neutral term} is a term that is not a $\lambda$-abstraction and
that does reduce to something. The set of {\em closed neutral terms} is denoted with $\Neutral$. We write $\CT$ for the set of closed terms and $\SN$ for the set of closed, strongly normalising terms. If $\ve t$ is any term, $\Red(\ve t)$ is the set of all terms $\ve t'$ such that $\ve t\to \ve t'$.
It is naturally extended to sets of terms. We say that a set $S$ of closed terms is a reducibility candidate, denoted with $S\in\RC$ if the following conditions are verified:
\begin{description}
  \item[$\RCn_1$] Strong normalisation:
    $S\subseteq\SN$.
  \item[$\RCn_2$] Stability under reduction: $\ve t\in S$ implies
    $\Red(\ve t)\subseteq S$.
  \item[$\RCn_3$] Stability under neutral expansion: If $\ve
    t\in\Neutral$ and $\Red(\ve t)\subseteq S$ then $\ve t\in S$.
  \item[$\RCn_4$] The common inhabitant: $\ve 0\in S$.
\end{description}

\noindent
We define the notion of {\em algebraic context} over a list of terms $\vec {\ve
t}$, with the following grammar:
\[
  F(\vec{\ve t}),G(\vec{\ve t})\quad::=\quad \ve t_i~|~ F(\vec{\ve t}) +
  G(\vec{\ve t})~|~\alpha\cdot F(\vec{\ve t}) ~|~ {\ve 0},
\]
where $\ve t_i$ is the $i$-th element of the list $\ve t$.  Given a set of
terms $S=\{\ve s_i\}_i$, we write $\mathcal{F}(S)$ for the set of terms of the
form $F(\vec{\ve s})$ when $F$ spans over algebraic contexts.

We introduce two conditions on contexts, which will be handy to define some of the operations on candidates:

\begin{description}
  \item[$\CC_1$] If for some $F$, $F(\vec{\ve s})\in S$ then $\forall
    i, \ve s_i\in S$.
  \item[$\CC_2$] If for all $i$, $\ve s_i\in S$ and $F$ is
    an algebraic context, then $F(\vec{\ve s})\in S$.
\end{description}

\medskip\noindent
We then define the following operations on reducibility
candidates.
\begin{enumerate}
  \item Let $\bcal{A}$ and $\bcal{B}$ be in $\RC$.
    $\bcal{A}\to\bcal{B}$ is the closure under $\RCn_3$ and $\RCn_4$ of
    the set of $\ve t\in\CT$ such that $(\ve t)\,\ve 0\in\bcal{B}$ and
    such that for all base terms $\ve b\in\bcal{A}$, $(\ve t)\,\ve
    b\in\bcal{B}$.
  \item If $\{\bcal A_i\}_i$ is a family of reducibility candidates,
    $\sum_i\bcal{A_i}$ is the closure under $\CC_1$, $\CC_2$, $\RCn_2$
    and $\RCn_3$ of the set $\cup_i\bcal A_i$.
\end{enumerate}

\begin{remark}\rm 
  Notice that $\sui{1}\bcal A\neq\bcal A$. Indeed, $\sui 1\bcal A$ is
  in particular the closure over $\CC_2$, meaning that all
  linear combinations of terms of $\bcal A$ belongs to $\sui{1}\bcal
  A$, whereas they might not be in $\bcal A$.
\end{remark}

\begin{remark}\label{rem:linalgcontext}\rm
  In the definition of algebraic contexts, a term $t_i$ might appear
  at several positions in the context. However, for any given
  algebraic context $F(\vec{\ve t})$ there is always a {\em linear}
  algebraic context $F_l(\vec{\ve t}')$ with a suitably modified list
  of terms $\vec{\ve t}'$ such that $F(\vec{\ve t})$ and
  $F_l(\vec{\ve t}')$ are the same terms. For example, choose the
  following (arguably very verbose) construction: if $F$ contains $m$
  placeholders, and if $\vec{\ve t}$ is of size $n$, let
  $\vec{\ve t}'$ be the list ${\ve t}_1\ldots{\ve t}_1,{\ve
  t}_2\ldots{\ve t}_2,\ldots,{\ve t}_n\ldots{\ve t}_n$ with each time $m$
  repetitions of each ${\ve t}_i$. Then construct $F_l(\vec{\ve t'})$
  exactly as $F$ except that for each $i$th placeholder we pick the
  $i$th copy of the corresponding term in $F$. By construction, each
  element in the list $\vec{\ve t}'$ is used at most once, and the
  term $F_l(\vec{\ve t'})$ is the same as the term $F(\vec{\ve t})$.
\end{remark}

\begin{lemma}\label{lem:RCop}
  If $\bcal{A}$, $\bcal{B}$ and all the $\bcal{A}_i$'s are in $\RC$,
  then so are $\bcal{A}\to\bcal{B}$, $\sum_i\bcal{A}_i$ and
  $\cap_i\bcal{A}_i$.
  \qed
\end{lemma}

A \emph{single type valuation} is a partial function 
from type
variables to reducibility candidates, that we define as a sequence of
comma-separated mappings, with $\emptyset$ denoting the empty
valuation: $\rho:=\,\emptyset~|~\rho,X\mapsto\bcal{A}$.
Type variables are interpreted using pairs of single type valuations,
that we simply call {\em valuations}, with common domain: $\rho =
(\rho_+,\rho_-)$ with $|\rho_+|=|\rho_-|$.
Given a valuation $\rho=(\rho_+,\rho_-)$, the {\em complementary
valuation} $\bar\rho$ is the pair $(\rho_-,\rho_+)$. We write
$(X_+,X_-)\mapsto(A_+,A_-)$ for the valuation $(X_+\mapsto A_+,
X_-\mapsto A_-)$. A valuation is called \emph{valid} if for all $X$,
$\rho_-(X)\subseteq\rho_+(X)$.

From now on, we will consider the following grammar
$$\vara{U,V,W} ::= U~|~\vara{X}.$$

That is, we will use $\vara{U,V,W}$ for unit and $\vara{X}$-kind of  variables.

To define the interpretation of a type $T$, we use the following
result.

\begin{lemma}\label{lem:typedecomp}
  Any type $T$, has a unique (up to $\equiv$) canonical decomposition $T\equiv\sui{n}\alpha_i\cdot\vara U_i$ such that for all $l,k$, $\vara U_l\not\equiv\vara U_k$.
  \qed
\end{lemma}

The interpretation $\denot{T}_\rho$ of a type $T$ in a
valuation~$\rho=(\rho_+,\rho_-)$ defined for each free type variable
of~$T$ is given by: 
$$\begin{array}{r@{~=~}l}
  \denot{X}_\rho & \rho_+(X),\\
  \denot{U\to T}_\rho & \denot{U}_{\bar\rho}\to\denot{T}_\rho,\\
  \denot{\forall X.U}_\rho & 
  \cap_{\bcal{A}\in\RC}\denot{U}_{\rho,(X_+,X_-)\mapsto(\bcal{A},
  \bcal{A})},\\
  \multicolumn{2}{c}{\mbox{If }T\equiv\sum_i\alpha_i\cdot\vara U_i\mbox{ is the 
  canonical decomposition of }T\mbox{ and }T\not\equiv\vara U}\\
  \denot{T}_\rho & \sum_i\denot{\vara U_i}_{\rho}
\end{array}$$
From Lemma~\ref{lem:RCop}, the interpretation of any type is a
reducibility candidate.

Reducibility candidates deal with closed terms, whereas proving the
adequacy lemma by induction requires the use of open terms with some
assumptions on their free variables, that will be guaranteed by a
context. Therefore we use \emph{substitutions} $\sigma$ to close
terms: 
$$\sigma := \emptyset \;|\; (x \mapsto\ve b;\sigma)\enspace,$$
then
$\ve{t}_{\emptyset} = \ve{t}$ and $\ve{t}_{x \mapsto \ve b;\sigma} =
\ve{t}[\ve b/x]_{\sigma}$. All the substitutions ends by $\emptyset$, hence we 
omit it when not necessary.

Given a context $\Gamma$, we say that a substitution~$\sigma$
\emph{satisfies}~$\Gamma$ for the valuation~$\rho$
(notation:~$\sigma\in\denot{\Gamma}_{\rho}$) when~$(x:U) \in \Gamma$
implies $x_\sigma\in\denot{U}_{\bar\rho}$ (Note the change in
polarity). 
A typing judgement $\Gamma\vdash\ve t: T$, is said to be \emph{valid}
(notation $\Gamma\models\ve t: T$) if
\begin{itemize}
  \item in case $T\equiv\vara U$, then for every valuation~$\rho$, and for every substitution~$\sigma\in\denot{\Gamma}_\rho$, we have
    $\ve{t}_{\sigma}\in\denot{\vara U}_{\rho}$.
  \item in other case, that is, $T\equiv\sui{n}\alpha_i\cdot\vara U_i$ with $n>1$, such that for all $i,j$, $\vara U_i\not\equiv\vara U_j$ (notice that by Lemma~\ref{lem:typedecomp} such a decomposition always exists), then for every valuation~$\rho$, and set of valuations $\{\rho_i\}_n$, where $\rho_i$ acts on $FV(U_i)\setminus FV(\Gamma)$, and for every substitution~$\sigma\in\denot{\Gamma}_\rho$, we have $\ve{t}_{\sigma}\in\sum_{i=1}^n\denot{\vara U_i}_{\rho,\rho_i}$. 
\end{itemize}

\begin{lemma}\label{lem:substRed}
  For any types $T$ and $A$, variable $X$ and valuation $\rho$, we have
  $\denot{T[A/X]}_\rho
  =
  \denot{T}_{\rho,(X_+,X_-)\mapsto(\denot{A}_{\bar\rho},\denot{A}_{\rho})}$
  and
  $\denot{T[A/X]}_{\bar\rho}
  =
  \denot{T}_{\bar\rho,(X_-,X_+)\mapsto(\denot{A}_{\rho},\denot{A}_{\bar\rho})}$.
  \qed
\end{lemma}

The proof of the Adequacy Lemma as well as the machinery of needed auxiliary lemmas can be found in~\ref{app:adequacy}.
\begin{lemma}[Adequacy Lemma]\label{lem:SNadeq}
  Every derivable typing judgement is valid: For every valid sequent
  $\Gamma\vdash\ve t:T$, we have $\Gamma\models\ve t:T$.
  \qed
\end{lemma}

\begin{theorem}[Strong normalisation]\label{th:SN}
  If $\Gamma\vdash\ve t:T$ is a valid sequent, then $\ve t$ is
  strongly normalising.
\end{theorem}
\begin{proof}
  If $\Gamma$ is the list $(x_i:U_i)_i$, the sequent $\vdash\lambda x_1\ldots
  x_n.\ve t:U_1\to(\cdots\to(U_n\to T)\cdots)$ is derivable. Using
  Lemma~\ref{lem:SNadeq}, we deduce that for any valuation $\rho$ and any
  substitution $\sigma\in\denot{\emptyset}_\rho$, we have $\lambda x_1\ldots
  x_n.\ve t_\sigma\in\denot{T}_\rho$. By construction, $\sigma$ does nothing on
  $\ve t$: $\ve t_\sigma = \ve t$. Since $\denot{T}_\rho$ is a reducibility
  candidate, $\lambda x_1\ldots x_n.\ve t$ is strongly normalising and hence
  $\ve t$ is strongly normalising.
\end{proof}

\section{Interpretation of typing judgements}\label{sec:examples}

\subsection{The general case}
In the general case the calculus can represent infinite-dimensional linear operators such as $\lambda x.x$, $\lambda x.\lambda y.y$, $\lambda x.\lambda f.(f)\,x$,\dots and their applications. 
Even for such general terms $\ve t$, the vectorial type system provides much information about the superposition of basis terms $\sum_i\alpha_i\cdot\ve b_i$ to which $\ve t$ reduces, as explained in Theorem~\ref{thm:termcharact}.
How much information is brought by the type system in the finitary case is the topic of Section~\ref{sec:finitary}. 

\begin{theorem}[Characterisation of terms]
  \label{thm:termcharact}
  Let $T$ be a generic type with canonical decomposition $\sui{n}\alpha_i.\vara{U}_i$, in the sense of Lemma~\ref{lem:typedecomp}. If ${}\vdash\ve t:T$,
  then $\ve t\to^*\sui{n}\suj{m_i}\beta_{ij}\cdot\ve b_{ij}$, where
  for all $i$, $\vdash\ve b_{ij}:\vara{U}_i$ and
  $\suj{m_i}\beta_{ij}=\alpha_i$, and with the convention that $\suj{0}\beta_{ij}=0$ and $\suj{0}\beta_{ij}\cdot\ve b_{ij}=\ve 0$.
  \qed
\end{theorem}
The detailed proof of the previous theorem can be found in~\ref{app:examples}

\subsection{The finitary case: Expressing matrices and vectors}\label{sec:finitary}
In what we call the ``finitary case'', we show how to encode finite-dimensional linear operators, i.e. matrices, together with their applications to vectors, as well as matrix and tensor products.
Theorem~\ref{thm:matrixsound} shows that we can encode matrices, vectors and operations upon them, and the type system will provide the result of such operations.

\subsubsection{In 2 dimensions}
In this section we come back to the motivating example introducing the
type system and we show how {\lvec} handles the Hadamard gate, and how
to encode matrices and vectors.

With an empty typing context, the booleans
$\true=\lambda x.\lambda y.x\,$ and $\,\false=\lambda x.\lambda y.y$
can be respectively typed with the types
$\True=\forall \varu{XY}.\varu X\to (\varu Y\to\varu X)\,$ and $\,\False=\forall\varu{XY}.\varu X\to (\varu Y\to\varu Y)$.
The superposition has the following type
$\vdash\alpha\cdot\true+\beta\cdot\false:\alpha\cdot\True +
\beta\cdot\False$. (Note that it can also be typed with
$(\alpha+\beta)\cdot \forall\varu X.\varu X\to\varu X\to\varu X$).

The linear map $\ve{U}$ sending $\true$ to $a\cdot\true+b\cdot\false$ and $\false$ to $c\cdot\true+d\cdot\false$, that is
\begin{align*}
  \true
  &\mapsto
  a\cdot\true+b\cdot\false,
  \\
  \false
  &\mapsto
  c\cdot\true+d\cdot\false
\end{align*}
is written as
\[
  \ve U={\lambda
    x.\cocanon{((x)\canon{a\cdot\true+b\cdot\false})\canon{c\cdot\true+d\cdot\false}}
  }.
\]
The following sequent is valid:
\[
  \vdash\ve{U}:\forall \vara{X}.((I\to
  (a\cdot\True+b\cdot\False))\to(I\to (c\cdot\True+d\cdot\False))\to I\to
  \vara{X})\to \vara{X}.
\]
This is consistent with the discussion in the introduction: the Hadamard gate
is the case $a=b=c=\frac1{\sqrt2}$ and $d=-\frac1{\sqrt2}$.  One can check that
with an empty typing context, $({\bf U})~\true$ is well typed of type
$a\cdot\True+b\cdot\False$, as expected since it reduces to
$a\cdot\true+b\cdot\false$.

The term $({\bf H})~\frac1{\sqrt2}\cdot (\true+\false)$ is well-typed of type
$\True+0\cdot\False$. Since the term reduces to $\true$, this is consistent
with the subject reduction: we indeed have
$\True\sqsupseteq\True+0\cdot\False$.

But we can do more than typing $2$-dimensional vectors $2\times2$-matrices:
using the same technique we can encode vectors and matrices of any size.

\subsubsection{Vectors in $n$ dimensions}\label{sec:vec}
The $2$-dimensional space is represented by the span of $\lambda
x_1x_2.x_1$ and $\lambda x_1x_2.x_2$: the $n$-dimensional space is
simply represented by the span of all the $\lambda
x_1\cdots{}x_n.x_i$, for $i=1\cdots{}n$. As for the two dimensional
case where
\[
  \vdash
  \alpha_1\cdot\lambda x_1x_2.x_1 + 
  \alpha_2\cdot\lambda x_1x_2.x_2
  :
  \alpha_1\cdot\forall \varu{X}_1\varu{X}_2.\varu{X}_1
  +
  \alpha_2\cdot\forall \varu{X}_1\varu{X}_2.\varu{X}_2,
\]
an $n$-dimensional vector is typed with
\[
  \vdash
  \sui{n}\alpha_i\cdot\lambda x_1\cdots{}x_n.x_i
  :
  \sui{n}\alpha_i\cdot\forall \varu{X}_1\cdots{}\varu{X}_n.\varu{X}_i.
\]
We use the notations
\[
  {\ve e}_i^n = \lambda x_1\cdots{}x_n.x_i,
  \qquad
  {\ve E}_i^n = \forall \varu{X}_1\cdots{}\varu{X}_n.\varu{X}_i
\]
and we write
\[
  \begin{array}{r@{~=~}l@{~=~}l}
    \left\llbracket
    \begin{array}{c}
      \alpha_1
      \\
      \vdots
      \\
      \alpha_n
    \end{array}
    \right\rrbracket^{\rm term}_{n}
    &
    \left(\begin{array}{c}
      \alpha_{1}\cdot{\ve e}_1^n
      \\+\\
      \cdots
      \\+\\
      \alpha_{n}\cdot{\ve e}_n^n
    \end{array}\right)
    &
    \sum\limits_{i=1}^{n}\alpha_i\cdot {\ve e}_i^n,
    \\\multicolumn{3}{c}{ }\\
    \left\llbracket
    \begin{array}{c}
      \alpha_1
      \\
      \vdots
      \\
      \alpha_n
    \end{array}
    \right\rrbracket^{\rm type}_{n}
    &
    \left(\begin{array}{c}
      \alpha_{1}\cdot{\ve E}_1^n
      \\+\\
      \cdots
      \\+\\
      \alpha_{n}\cdot{\ve E}_n^n
    \end{array}\right)
    &
    \sum\limits_{i=1}^{n}\alpha_i\cdot {\ve E}_i^n.
  \end{array}
\]

\subsubsection{$n\times m$ matrices}\label{sec:mat}
Once the representation of vectors is chosen, it is easy to generalize
the representation of $2\times 2$ matrices to the $n\times m$
case. Suppose that the matrix $U$ is of the form
\[
  U = 
  \left(
  \begin{array}{ccc}
    \alpha_{11} & \cdots & \alpha_{1m}
    \\
    \vdots && \vdots
    \\
    \alpha_{n1} & \cdots & \alpha_{nm}
  \end{array}
  \right),
\]
then its representation is
\[
  \left\llbracket
  U
  \right\rrbracket^{\rm term}_{n\times m}
  ={~~~~}
  \lambda x.
  \left\{
    \left(
    \cdots
    \left(
    (x)
    \left[
      \begin{array}{c}
        \alpha_{11}\cdot{\ve e}_1^n 
        \\+\\
        \cdots
        \\+\\
        \alpha_{n1}\cdot{\ve e}_n^n
      \end{array}
    \right]
    \right)
    \cdots
    \left[
      \begin{array}{c}
        \alpha_{1m}\cdot{\ve e}_1^n
        \\+\\
        \cdots
        \\+\\
        \alpha_{nm}\cdot{\ve e}_n^n
      \end{array}
    \right]
    \right)
  \right\}\qquad
\]
and its type is
\[
  \left\llbracket
  U
  \right\rrbracket^{\rm type}_{n\times m}
  ={~~~~}
  \forall\vara{X}.
  \left(
  \left[
    \begin{array}{c}
      \alpha_{11}\cdot{\ve E}_1^n 
      \\+\\
      \cdots
      \\+\\
      \alpha_{n1}\cdot{\ve E}_n^n
    \end{array}
  \right]\to
  \cdots
  \to
  \left[
    \begin{array}{c}
      \alpha_{1m}\cdot{\ve E}_1^n
      \\+\\
      \cdots
      \\+\\
      \alpha_{nm}\cdot{\ve E}_n^n
    \end{array}
  \right]\to
  [~\vara{X}~]
  \right)
  \to
  \vara{X},
\]
that is, an almost direct encoding of the matrix $U$.

\medskip
We also use the shortcut notation
\[
  {\bf mat}(\ve t_1,\ldots,\ve t_n)
  =
  \lambda x.(\ldots((x)\,\canon{\ve t_1})\ldots)\,\canon{\ve t_n}
\]

\subsubsection{Useful constructions}

In this section, we describe a few terms representing constructions
that will be used later on.

\paragraph{Projections}

The first useful family of terms are the projections, sending a vector
to its $i^{\rm th}$ coordinate:
\[
  \left(
  \begin{array}{c}
    \alpha_1
    \\
    \vdots
    \\
    \alpha_i
    \\
    \vdots
    \\
    \alpha_n
  \end{array}
  \right)
  \longmapsto
  \left(
  \begin{array}{c}
    0
    \\
    \vdots
    \\
    \alpha_i
    \\
    \vdots
    \\
    0
  \end{array}
  \right).
\]
Using the matrix representation, the term projecting the $i^{\rm th}$
coordinate of a vector of size $n$ is
\[
  \xymatrix@R=0em@C=0em{
    \textrm{$i^{\rm th}$ position}\ar@/^1em/[rd]&&
    \\
    {\ve p}^n_i = 
    {\bf mat}({\ve 0},\cdots,{\ve 0},&{\ve e}^n_i,&{\ve 0},\cdots,{\ve
    0}).
  }
\]
We can easily verify that
\[
  \vdash
  {\ve p}^n_i
  :
  \left\llbracket
  \begin{array}{ccccc}
    0 & \cdots & 0 & \cdots & 0
    \\
    \vdots & \ddots &  &  & \vdots
    \\
    0 & & 1 & & 0
    \\
    \vdots & &  & \ddots & \vdots
    \\
    0 & \cdots & 0 & \cdots & 0
  \end{array}
  \right\rrbracket^{\rm type}_{n\times n}
\]
and that
\[
  ({\ve p}^n_{i_0})\,
  \left(
  \sui{n}\alpha_i\cdot{\ve e}^n_i
  \right)
  \longrightarrow^*
  \alpha_{i_0}\cdot{\ve e}^n_{i_0}.
\]

\paragraph{Vectors and diagonal matrices}
Using the projections defined in the previous section, it is possible
to encode the map sending a vector of size $n$ to the corresponding
$n\times n$ matrix:
\[
  \left(
  \begin{array}{c}
    \alpha_1
    \\
    \vdots
    \\
    \alpha_n
  \end{array}
  \right)
  \longmapsto
  \left(
  \begin{array}{c@{~}c@{~}c}
    \alpha_1&&0
    \\
    &\ddots&
    \\
    0&&\alpha_n
  \end{array}
  \right)
\]
with the term
\[
  {\bf diag}^n = 
  \lambda b.{\bf mat}(({\ve p}^n_1)\,\{b\},\ldots,({\ve p}^n_n)\,\{b\})
\]
of type
\[
  \vdash
  {\bf diag}^n :
  \left[
    \left\llbracket
    \begin{array}{c}
      \alpha_1
      \\
      \vdots
      \\
      \alpha_n
    \end{array}
    \right\rrbracket^{\rm type}_{n}
  \right]
  \to
  \left\llbracket
  \begin{array}{ccc}
    \alpha_1 & & 0
    \\
    & \ddots & 
    \\
    0 & & \alpha_n
  \end{array}
  \right\rrbracket^{\rm type}_{n\times n}.
\]
It is easy to check that
\[
  ({\bf diag}^n)\,
  \left[
    \sui{n}\alpha_i\cdot{\bf e}^n_i
  \right]
  \longmapsto^*
  {\bf mat}(\alpha_1\cdot{\bf e}^n_1,\ldots,\alpha_n\cdot{\bf e}^n_n)
\]

\paragraph{Extracting a column vector out of a matrix}
Another construction that is worth exhibiting is the operation
\[
  \left(
  \begin{array}{c@{~}c@{~}c}
    \alpha_{11}&\cdots&\alpha_{1n}
    \\
    \vdots&&\vdots
    \\
    \alpha_{m1}&\cdots&\alpha_{mn}
  \end{array}
  \right)
  \longmapsto
  \left(
  \begin{array}{c}
    \alpha_{1i}
    \\
    \vdots
    \\
    \alpha_{mi}
  \end{array}
  \right).
\]
It is simply defined by multiplying the input matrix with the $i^{\rm
th}$ base column vector:
\[
  {\bf col}^n_i 
  =
  \lambda x.(x)\,{\ve e}^n_i
\]
and one can easily check that this term has type
\[
  \vdash
  {\bf col}^n_i 
  :
  \left\llbracket
  \begin{array}{ccc}
    \alpha_{11}&\cdots&\alpha_{1n}
    \\
    \vdots&&\vdots
    \\
    \alpha_{m1}&\cdots&\alpha_{mn}
  \end{array}
  \right\rrbracket^{\rm type}_{m\times n}
  \to
  \left\llbracket
  \begin{array}{c}
    \alpha_{1i}
    \\
    \vdots
    \\
    \alpha_{mi}
  \end{array}
  \right\rrbracket^{\rm type}_{m}.
\]
Note that the same term ${\bf col}^n_i$ can be typed with several
values of $m$.

\subsubsection{A language of matrices and vectors}

In this section we formalize what was informally presented in the
previous sections: the fact that one can encode simple matrix and
vector operations in \lvec{}, and the fact that the type system
serves as a witness for the result of the encoded operation.

We define the language $\lmat$ of matrices and vectors with the
grammar
\[
  \begin{array}{ll}
    M,N &{}::={} \zeta ~|~ M\otimes N ~|~ (M)\,N
    \\
    u,v &{}::={} \nu ~|~ u\otimes v ~|~ (M)\,u,
  \end{array}
\]
where $\zeta$ ranges over the set matrices and $\nu$ over the set of
(column) vectors.
Terms are implicitly typed: types of matrices are $(m,n)$ where $m$
and $n$ ranges over positive integers, while types of vectors are
simply integers. Typing rules are the following.
\[
  \infer{\zeta:(m,n)}{\zeta\in\mathbb{C}^{m\times n}}
  \qquad
  \infer{M\otimes N:(mm',nn')}{M:(m,n) & N:(m',n')}
  \qquad
  \infer{(M)\,N:(m,n)}{M:(m,n') & N:(n',n)}
\]
\[
  \infer{\nu:m}{\nu\in\mathbb{C}^{m}}
  \qquad
  \infer{u\otimes v:mn}{u:m & v:n}
  \qquad
  \infer{(M)\,u:n}{M:(m,n) & u:m}
\]
The operational semantics of this language is the natural
interpretation of the terms as matrices and vectors. If $M$ computes
the matrix $\zeta$, we write $M\downarrow\zeta$. Similarly, if $u$
computes the vector $\nu$, we write $u\downarrow\nu$.

Following what we already said, matrices and vectors can be interpreted
as types and terms in \lvec{}. The map $\denot{-}^{\rm term}$ sends
terms of ${\it Mat}$ to terms of $\lvec$ and the map $\denot{-}^{\rm
type}$ sends matrices and vectors to types of $\lvec$.
\begin{itemize}
  \item Vectors and matrices are defined as in Sections~\ref{sec:vec}
    and~\ref{sec:mat}.
  \item As we already discussed, the matrix-vector multiplication
    is simply the application of terms in \lvec{}:
    \[
      \termdenot{(M)\,u} = (\termdenot{M})\,\termdenot{u}
    \]
  \item The matrix multiplication is performed by first extracting the
    column vectors, then performing the matrix-vector multiplication:
    this gives a column of the final matrix. We conclude by recomposing
    the final matrix column-wise.

    That is done with the term
    \[
      {\bf app} = 
      \lambda xy.{\bf mat}((x)\,(({\bf col}^m_1)\,y),\ldots,(x)\,(({\bf col}^m_n)\,y))
    \]
    and its type is 
    \begin{center}
      \scalebox{.95}{$
        \left\llbracket
        \begin{array}{ccc}
          \alpha_{11}&\cdots&\alpha_{1n}
          \\
          \vdots&&\vdots
          \\
          \alpha_{m1}&\cdots&\alpha_{mn}
        \end{array}
        \right\rrbracket^{\rm type}_{m\times n}
        \to
        \left\llbracket
        \begin{array}{ccc}
          \beta_{11}&\cdots&\beta_{1k}
          \\
          \vdots&&\vdots
          \\
          \beta_{n1}&\cdots&\beta_{nk}
        \end{array}
        \right\rrbracket^{\rm type}_{n\times k}
        \to
        \left\llbracket
        \left(
        \sui{n}\alpha_{ji}\beta_{il}
        \right)_{\substack{j=1\ldots m\\l=1\ldots k}}
        \right\rrbracket^{\rm type}_{m\times k}
      $}
    \end{center}

    \noindent Hence,
    \[
      \termdenot{(M)~N}=
      (({\bf app})~\termdenot{M})~\termdenot{N}
    \]
  \item
    For defining the the tensor of vectors, we need to multiply the
    coefficients of the vectors:
    \[
      \left(
      \begin{array}{c}
        \alpha_1
        \\
        \vdots
        \\
        \alpha_n
      \end{array}
      \right)
      \otimes
      \left(
      \begin{array}{c}
        \beta_1
        \\
        \vdots
        \\
        \beta_m
      \end{array}
      \right)
      =
      \left(
      \begin{array}{c}
        \alpha_1
        \cdot
        \left(
        \begin{array}{c}
          \beta_1
          \\
          \vdots
          \\
          \beta_m
        \end{array}
        \right)
        \\
        \vdots
        \\
        \alpha_n
        \cdot
        \left(
        \begin{array}{c}
          \beta_1
          \\
          \vdots
          \\
          \beta_m
        \end{array}
        \right)
      \end{array}
      \right)
      =
      \left(
      \begin{array}{c}
        \alpha_1\beta_1
        \\
        \vdots
        \\
        \alpha_1\beta_m
        \\
        \vdots
        \\
        \alpha_n\beta_1
        \\
        \vdots
        \\
        \alpha_n\beta_m
      \end{array}
      \right).
    \]
    We perform this operation in several steps: First, we map the two
    vectors $(\alpha_i)_i$ and $(\beta_j)_j$ into matrices of size
    $mn\times mn$:
    \begin{center}
      \scalebox{.83}{
        $\left(
        \begin{array}{c}
          \alpha_1
          \\
          \vdots
          \\
          \alpha_n
        \end{array}
        \right)
        \mapsto
        \left(
        \begin{array}{c@{}c@{}c@{}c@{}c@{}c@{}c}
          \alpha_1&&&&&&
          \\
          &\vddots&&&&&
          \\
          &&\alpha_1&&&&
          \\
          &&&\vddots&&&
          \\
          &&&&\alpha_n&&
          \\
          &&&&&\vddots&
          \\
          &&&&&&\alpha_n
        \end{array}
        \right)\!\!
        \begin{array}{c}
          \left.\rule{0em}{2.5em}\right\}
          \times{}m
          \\
          \rule{0em}{1em}
          \\
          \left.\rule{0em}{2.5em}\right\}
          \times{}m
        \end{array}
        \textrm{and}\quad
        \left(
        \begin{array}{c}
          \beta_1
          \\
          \vdots
          \\
          \beta_m
        \end{array}
        \right)
        \mapsto
        \left(
        \begin{array}{c@{}c@{}c@{}c@{}c@{}c@{}c}
          \beta_1&&&&&
          \\
          &\vddots&&&&
          \\
          &&\beta_m&&&&
          \\
          &&&\vddots&&&
          \\
          &&&&\beta_1&&
          \\
          &&&&&\vddots&
          \\
          &&&&&&\beta_m
        \end{array}
        \right)\!\!
        \begin{array}{c}
          \left.\rule{0em}{2.5em}\right\}
          \\
          \rule{0em}{1em}
          \\
          \left.\rule{0em}{2.5em}\right\}
        \end{array}
        \left.\rule{0em}{3em}\right\}
        \times{}n.
      $}
    \end{center}
    These two operations 
    can be represented as terms of \lvec{} respectively as follows:
    \begin{center}
      \scalebox{.9}{
        ${\ve m}^{n,m}_1 = 
        \lambda b.
        \,{\bf mat}\left(
        \begin{array}{c}
          ({\ve p}^{n}_1)\cocanon{b},
          \\
          \vdots
          \\
          ({\ve p}^{n}_1)\cocanon{b},
          \\
          \vdots
          \\
          ({\ve p}^{n}_n)\cocanon{b},
          \\
          \vdots
          \\
          ({\ve p}^{n}_n)\cocanon{b}
        \end{array}
        \right)\!\!
        \begin{array}{c}
          \left.\rule{0em}{2.5em}\right\}
          \times{}m
          \\
          \rule{0em}{1em}
          \\
          \left.\rule{0em}{2.5em}\right\}
          \times{}m
        \end{array}
        \textrm{and}\quad
        {\ve m}^{m,n}_2 = 
        \lambda b.\,
        {\bf mat}
        \left(
        \begin{array}{c}
          ({\ve p}^{m}_1)\cocanon{b},
          \\
          \vdots
          \\
          ({\ve p}^{m}_m)\cocanon{b},
          \\
          \vdots
          \\
          ({\ve p}^{m}_1)\cocanon{b},
          \\
          \vdots
          \\
          ({\ve p}^{m}_m)\cocanon{b}
        \end{array}
        \right)\!\!
        \begin{array}{c}
          \left.\rule{0em}{2.5em}\right\}
          \\
          \rule{0em}{1em}
          \\
          \left.\rule{0em}{2.5em}\right\}
        \end{array}
        \left.\rule{0em}{3em}\right\}
        \times{}n.$}
      \end{center}
      It is now enough to multiply these two matrices together to retrieve
      the diagonal:
      \[
        \left(
        \begin{array}{c@{}c@{}c@{}c@{}c@{}c@{}c}
          \alpha_1&&&&&&
          \\
          &\vddots&&&&&
          \\
          &&\alpha_1&&&&
          \\
          &&&\vddots&&&
          \\
          &&&&\alpha_n&&
          \\
          &&&&&\vddots&
          \\
          &&&&&&\alpha_n
        \end{array}
        \right)
        \left(
        \begin{array}{c@{}c@{}c@{}c@{}c@{}c@{}c}
          \beta_1&&&&&
          \\
          &\vddots&&&&
          \\
          &&\beta_m&&&&
          \\
          &&&\vddots&&&
          \\
          &&&&\beta_1&&
          \\
          &&&&&\vddots&
          \\
          &&&&&&\beta_m
        \end{array}
        \right)
        \left(
        \begin{array}{c}
          1\\
          \vdots
          \\
          1\\
          \vdots
          \\
          1\\
          \vdots
          \\
          1\\
        \end{array}
        \right)
        \quad=\quad
        \left(
        \begin{array}{c}
          \alpha_1\beta_1
          \\
          \vdots
          \\
          \alpha_1\beta_m
          \\
          \vdots
          \\
          \alpha_n\beta_1
          \\
          \vdots
          \\
          \alpha_n\beta_m
        \end{array}
        \right)
      \]
      and this can be implemented through matrix-vector multiplication:
      \[
        {\bf tens}^{n,m} = 
        \lambda bc.(({\ve m}^{n,m}_1)\,b)\,\left((({\ve m}^{m,n}_2)\,c)\,\left(\sui{mn}{\ve
        e}^n_i\right)\right).
      \]

      \noindent Hence, if $u:n$ and $v:m$, we have
      \[
        \termdenot{u\otimes v} = 
        (({\bf tens}^{n,m})~\termdenot{u})~\termdenot{v}
      \]

    \item
      The tensor of matrices is done column by column:
      \begin{multline*}
        \qquad\left(
        \begin{array}{ccc}
          \alpha_{11} & \dots & \alpha_{1n}\\
          \vdots & & \vdots\\
          \alpha_{n'1} & \dots & \alpha_{n'n}
        \end{array}
        \right)
        \otimes
        \left(
        \begin{array}{ccc}
          \beta_{11} & \dots & \beta_{1m}\\
          \vdots & & \vdots\\
          \beta_{m'1} & \dots & \beta_{m'm}
        \end{array}
        \right)
        = \\
        \left( 
        \begin{array}{ccc}
          \left( 
          \begin{array}{c}
            \alpha_{11}\\
            \vdots\\
            \alpha_{n'1}
          \end{array}
          \right)\otimes
          \left( 
          \begin{array}{c}
            \beta_{11}\\
            \vdots\\
            \beta_{m'1}
          \end{array}
          \right)
          & \dots &
          \left( 
          \begin{array}{c}
            \alpha_{1n}\\
            \vdots\\
            \alpha_{n'n}
          \end{array}
          \right)\otimes
          \left( 
          \begin{array}{c}
            \beta_{1m}\\
            \vdots\\
            \beta_{m'm}
          \end{array}
          \right)
        \end{array}
        \right)
      \end{multline*}

      \noindent If $M$ be a matrix of size $m\times m'$ and $N$ a matrix of size $n\times n'$. Then $M\otimes N$ has size $m\times n$, and it can be implemented as
      \begin{multline*}
        {\bf Tens}^{m,n} ={}\\
        \lambda bc.{\bf mat}(
        (({\bf tens}^{m,n})~({\bf col}_1^m)~b)~({\bf col}_1^n)~c,
        \cdots
        (({\bf tens}^{m,n})~({\bf col}_n^m)~b)~({\bf col}_m^n)~c)
      \end{multline*}

      \noindent Hence, if $M:(m,m')$ and $N:(n,n')$, we have
      \[
        \termdenot{M\otimes N} =
        (({\bf Tens}^{m,n})~\termdenot{M})~\termdenot{N}
      \]
\end{itemize}

\begin{theorem}\label{thm:matrixsound}
  The denotation of {\it Mat} as terms and types of {\lvec}
  are sound in the following sense.
  \[
    M\downarrow\zeta
    \qquad
    \textrm{implies}
    \qquad
    \vdash\termdenot{M} : \typedenot{\zeta},
  \]
  \[
    u\downarrow\nu
    \qquad
    \textrm{implies}
    \qquad
    \vdash\termdenot{u} : \typedenot{\nu}.
  \]
\end{theorem}

\begin{proof}
  The proof is a straightfoward structural induction on $M$ and $u$. 
\end{proof}

\subsection{{\lvec} and quantum computation}

In quantum computation, data is encoded on normalised vectors in
Hilbert spaces. For our purpose, their interesting property is to be modules over the ring of complex numbers. The
smallest non-trivial such space is the space of {\em qubits}. The space of qubits is the two-dimensional vector space $\mathbb{C}^2$, together with a chosen orthonormal basis $\{\ket0, \ket1\}$. A quantum bit (or qubit) is a normalised vector $\alpha\ket0 +
\beta\ket1$, where $|\alpha|^2 + |\beta|^2=1$. In quantum computation, the operations on
qubits that are usually considered are the {\em quantum gates}, \ie a chosen set of unitary
operations. For our purpose, their interesting property is to be {\em
linear}.

The fact that one can encode quantum circuits in {\lvec} is a corollary
of Theorem~\ref{thm:matrixsound}. Indeed, a quantum circuit can be
regarded as a sequence of multiplications and tensors of matrices. The
language of term can faithfully represent those, where as the type system can
serve as an abstract interpretation of the actual unitary map computed
by the circuit.

We believe that this tool is a first step towards lifting the ``quantumness'' of algebraic $\lambda$-calculi to the level of a type based analysis. It could also be a step
towards a ``quantum theoretical logic'' coming readily with a Curry-Howard
isomorphism. The logic we are sketching merges intuitionistic logic
and vectorial structure, which makes it intriguing.

The next step in the study of the quantumness of the linear algebraic
$\lambda$-calculus is the exploration of the notion of orthogonality
between terms, and the validation of this notion by means of a compilation
into quantum circuits. The work of~\cite{ValironQPL10} shows
that it is worthwhile pursuing in this direction.

\subsection{{\lvec} and other calculi}
No direct connection seems to exists between {\lvec} and intersection and union types~\cite{PimentelRonchiRoversiFI12,BarbaneraDezaniDeLiguoroIC95}. However, there is an ongoing project of a new type system based on intersections, which may take some of the ideas from the Vectorial lambda calculus.
Indeed, the sum resembles as a non-idempotent intersection, with some extra quantitative information. In~\cite{DebenadettieRonchiITRS12} a type system with non-idempotent intersection has been used to compute a bound on the normalisation time, and in~\cite{BernadetLengrandLMCS13,KesnerVenturaLNCS14} to provide new characterisations on strongly normalising terms. In any case, the Scalar type system~\cite{ArrighiDiazcaroLMCS12} seems more close to these results: only scalars are considered and so $\ve t+\ve t$ has type $2\cdot T$ if $\ve t$ has type $T$, hence the difference with the non-idempotent intersection type is that the scalars are not just natural numbers, but members of a given ring. The Vectorial lambda calculus goes beyond, as it not only have the quantitative information, thought the scalars, but also allows to \textit{intersect} different types.
On the other hand, the interpretation of linear combinations we give in this paper is more in line with a union than an intersection, which may be a interesting starting question to follow.

A different path to follow has been started at \cite{DiazcaroDowek15}, where a non-idempotent intersection (simply) type has been considered, with no scalars, and where all the isomorphisms on such a system are made explicit with an equivalence relation. One derivative of such a work has been a characterisation of superpositions and projective measurement~\cite{DiazcaroDowek16}. A natural objective to pursue is to add scalars to it, by somehow merging it with $\lvec$.

\paragraph{Acknowledgements}
We would like to thank
Gilles Dowek
and
Barbara Petit
for enlightening discussions.

\section{Bibliography}

\bibliographystyle{elsarticle-harv}

\appendix
\section{Detailed proofs of lemmas and theorems in Section~\ref{sec:sr}}\label{app:sr}
\subsection{Lemmas from Section~\ref{sec:prereq}}\label{app:srpre}
\recap{Lemma}{Characterisation of types}{lem:typecharact}
For any type $T$ in $\mathcal{G}$, there exist $n,m\in\mathbb{N}$, $\alpha_1,\dots,\alpha_n$, $\beta_1,\dots,\beta_m\in\Sc$, distinct unit types $U_1,\dots,U_n$ and distinct general variables $\vara{X}_1,\dots,\vara{X}_m$ such that $$T\equiv\sui{n}\alpha_i\cdot U_i+\suj{m}\beta_j\cdot\vara{X}_j\ .$$

\begin{proof}
  Structural induction on $T$.
  \begin{itemize}
    \item Let $T$ be a unit type, then take $\alpha=\beta=1$, $n=1$ and $m=0$, and so $T\equiv\sui{1}1\cdot U=1\cdot U$.
    \item Let $T=\alpha\cdot T'$, then by the induction hypothesis $T'\equiv\sui{n}\alpha_i\cdot U_i+\suj{m}\beta_j\cdot\vara{X}_j$, so $T=\alpha\cdot T'\equiv\alpha\cdot (\sui{n}\alpha_i\cdot U_i+\suj{m}\beta_j\cdot\vara{X}_j)\equiv\sui{n}(\alpha\times\alpha_i)\cdot U_i+\suj{m}(\alpha\times\beta_j)\cdot\vara{X}_j$.
    \item Let $T=R+S$, then by the induction hypothesis $R\equiv\sui{n}\alpha_i\cdot U_i+\suj{m}\beta_j\cdot\vara{X}_j$ and $S\equiv\sui{n'}\alpha'_i\cdot U'_i+\suj{m'}\beta'_j\cdot\vara{X'}_j$, so $T=R+S\equiv\sui{n}\alpha_i\cdot U_i+\sui{n'}\alpha'_i\cdot U'_i+\suj{m}\beta_j\cdot\vara{X}_j+\suj{m'}\beta'_j\cdot\vara{X'}_j$. If the $U_i$ and the $U'_i$ are all different each other, we have finished, in other case, if $U_k=U'_h$, notice that $\alpha_k\cdot U_k+\alpha'_h\cdot U'_h=(\alpha_k+\alpha'_h)\cdot U_k$.
    \item Let $T=\vara X$, then take $\alpha=\beta=1$, $m=1$ and $n=0$, and so $T\equiv\suj{1} 1\cdot\vara{X}=1\cdot\vara X$. \qedhere
  \end{itemize}
\end{proof}

\begin{definition}\rm
  Let $F$ be an algebraic context with $n$ holes. Let
  $\vec U = U_1,\ldots,U_n$ be a list of $n$ unit types. 
  If $U$ is a unit type, we write $\bar U$ for the set of unit types
  equivalent to $U$:
  \[
    \bar U := \{ V ~|~ V \textrm{ is unit and } V \equiv U \}.
  \]
  The {\em context vector} $v_F(\vec U)$ associated with the context
  $F$ and the unit types $\vec U$ is partial map from the set
  $\cal S = \{ \bar U \}$ to scalars. It is inductively defined as
  follows: $v_{\alpha\cdot F}(\vec U) := \alpha v_F(\vec U)$,
  $v_{F + G}(\vec U) := v_F(\vec U) + v_G(\vec U)$, and finally
  $v_{[-_i]}(\vec U) := \{\bar U_i \mapsto 1\}$.
  The sum is defined on these partial map as follows:
  \[
    (f + g)(\vec U) = 
    \left\{
      \begin{array}{ll}
        f(\vec U) + g(\vec U)& \textrm{if both are defined}
        \\
        f(\vec U)& \textrm{if $f(\vec U)$ is defined but not $g(\vec U)$}
        \\
        g(\vec U)& \textrm{if $g(\vec U)$ is defined but not $f(\vec U)$}
        \\
        \textrm{is not defined} & \textrm{if neither $f(\vec U)$ nor
        $g(\vec U)$ is defined.}
      \end{array}
      \right.
    \]
    Scalar multiplication is defined as follows:
    \[
      (\alpha f)(\vec U) = 
      \left\{
        \begin{array}{ll}
          \alpha (f(\vec U))& \textrm{if $f(\vec U)$ is defined}
          \\
          \textrm{is not defined} & \textrm{if $f(\vec U)$ is not defined.}
        \end{array}
        \right.
      \]
    \end{definition}

    \begin{lemma}\label{lem:equivdistinctscalarsaux}
      Let $F$ and $G$ be two algebraic contexts with respectively $n$ and
      $m$ holes. Let $\vec U$ be a list of $n$ unit types, and $\vec V$ be
      a list of $m$ unit types. Then $F(\vec U) \equiv G(\vec V)$ implies
      $v_F(\vec U) = v_G(\vec V)$.
    \end{lemma}

    \begin{proof}
      The derivation of $F(\vec U) \equiv F(\vec V)$ essentially consists
      in a sequence of the elementary rules (or congruence thereof) in
      Figure~\ref{fig:types} composed with transitivity:
      \[
        F(\vec U) = W_1\equiv W_2 \equiv \cdots \equiv W_k = G(\vec V).
      \]
      We prove the result by induction on $k$.
      \begin{itemize}
        \item Case $k=1$. Then $F(\vec U)$ is syntactically equal to
          $G(\vec V)$: we are done.
        \item Suppose that the result is true for sequences of size $k$, and
          let 
          \[
            F(\vec U) = W_1\equiv W_2 \equiv \cdots \equiv W_k
            \equiv W_{k+1} = G(\vec V).
          \]
          Let us concentrate on the first step $F(\vec U) \equiv W_2$: it is
          an elementary step from Figure~\ref{fig:types}. By structural
          induction on the proof of $F(\vec U) \equiv W_2$ (which only uses
          congruence and elementary steps, and not transitivity), we can show
          that $W_2$ is of the form $F'(\vec U')$ where
          $v_F(\vec U) = v_{F'}(\vec U')$. We are now in power of applying
          the induction hypothesis, because the sequence of elementary
          rewrites from $F'(\vec U')$ to $G(\vec V)$ is of size $k$.
          Therefore $v_{F'}(\vec U') = v_G(\vec V)$. We can then conclude
          that $v_F(\vec U) = v_G(\vec V)$.
      \end{itemize}
      This conclude the proof of the lemma.
    \end{proof}

    \recap{Lemma}{Equivalence between sums of distinct elements (up to $\equiv$)}{lem:equivdistinctscalars}
    Let $U_1,\dots,U_n$ be a set of distinct (not equivalent) unit types, and let $V_1,\dots,V_m$ be also a set distinct unit types. If $\sui{n}\alpha_i\cdot U_i\equiv\suj{m}\beta_j\cdot V_j$, then $m=n$ and there exists a permutation $p$ of $n$ such that $\forall i$, $\alpha_i=\beta_{p(i)}$ and $U_i\equiv V_{p(i)}$.

    \begin{proof}
      Let $S = \sui{n}\alpha_i\cdot U_i$ and $T = \suj{m}\beta_j\cdot V_j$.
      Both $S$ and $T$ can be
      respectively written as $F(\vec U)$ and $G(\vec V)$. Using 
      Lemma~\ref{lem:equivdistinctscalarsaux}, we conclude that $v_F(\vec
      U) = v_G(\vec V)$.
      Since all $U_i$'s are pairwise non-equivalent, the $\bar U_i$'s are
      pairwise distinct.
      \[
        v_F(\vec U) = \{ \bar U_i \mapsto \alpha_i~|~i=1\ldots n\}.
      \]
      Similarly, the $\bar V_j$'s  are pairwise disjoint, and
      \[
        v_G(\vec G) = \{ \bar V_j \mapsto \beta_j~|~i=1\ldots m\}.
      \]
      We obtain the desired result because these two partial maps are
      supposed to be equal. Indeed, this immplies:
      \begin{itemize}
        \item $m=n$ because the domains are equal (so they should have the
          same size)
        \item Again using the fact that the domains are equal, the sets
          $\{\bar U_i\}$ and $\{\bar V_j\}$ are equal: this means there
          exists a permutation $p$ of $n$ such that $\forall i$,
          $\bar U_i= \bar V_{p(i)}$, meaning $U_i\equiv V_{p(i)}$.
        \item Because the partial maps are equal, the images of a given
          element $\bar U_i = \bar V_{p(i)}$ under $v_F$ and $v_G$ are in
          fact the same: we therefore have $\alpha_i=\beta_{p(i)}$.
      \end{itemize}
      And this closes the proof of the lemma.
    \end{proof}

    \recap{Lemma}{Equivalences $\forall_I$}{lem:equivforall}
    Let $U_1,\dots,U_n$ be a set of distinct (not equivalent) unit types and let $V_1,\dots,V_n$ be also a set of distinct unit types.
    \begin{enumerate}
      \item\label{ap:it:equivforall1}
        $\sui{n}\alpha_i\cdot U_i\equiv\suj{m}\beta_j\cdot V_j\Leftrightarrow\sui{n}\alpha_i\cdot\forall X.U_i\equiv\suj{m}\beta_j\cdot\forall X.V_j$.
      \item\label{ap:it:equivforall2} $\sui{n}\alpha_i\cdot\forall X.U_i\equiv\suj{m}\beta_j\cdot V_j\Rightarrow\forall V_j,\exists W_j~/~V_j\equiv\forall X.W_j$.
      \item\label{ap:it:equivforall3} $T\equiv R\Rightarrow T[A/X]\equiv R[A/X]$.
    \end{enumerate}

    \begin{proof}
      Item (1) From Lemma~\ref{lem:equivdistinctscalars}, $m=n$, and without
      loss of generality, for all $i$, $\alpha_i=\beta_i$ and $U_i=V_i$ in
      the left-to-right direction, $\forall X.U_i=\forall X.V_i$ in the
      right-to-left direction. In both cases we easily conclude.

      \medskip
      \noindent
      Item (2) is similar.

      \medskip
      \noindent
      Item (3) is a straightforward induction on the equivalence $T\equiv R$.
    \end{proof}

    \recap{Lemma}{$\succeq$-stability}{lem:subjectreductionofrelation}
    If $T\succeq^{\ve t}_{\V,\Gamma} R$, $\ve t\to\ve r$ and $\Gamma\vdash\ve r: T$, then $T\succeq^{\ve r}_{\V,\Gamma} R$.

    \begin{proof}
      It suffices to show this for $\succ^{\ve t}_{X,\Gamma}$, with $X\in\V$. Observe that since $T\succ^{\ve t}_{X,\Gamma}R$, then $X\notin\FV{\Gamma}$. We only have to prove that $\Gamma\vdash\ve r: R$ is derivable from $\Gamma\vdash\ve r: T$. We proceed now by cases:
      \begin{itemize}
        \item $T\equiv\sui n\alpha_i\cdot U_i$ and $R\equiv\sui n\alpha_i\cdot\forall X.U_i$, then using rules $\forall_I$ and $\equiv$, we can deduce $\Gamma\vdash\ve r: R$.
        \item $T\equiv\sui n\alpha_i\cdot\forall X.U$ and $R\equiv\sui n\alpha_i\cdot U_i[A/X]$, then using rules $\forall_E$ and $\equiv$, we can deduce $\Gamma\vdash\ve r: R$.\qedhere
      \end{itemize}
    \end{proof}

    \recap{Lemma}{Arrows comparison}{lem:arrowscomp} If
    $V\to R\succeq^{\ve t}_{\V,\Gamma} \forall\vec X.(U\to T)$, then $U\to T\equiv(V\to R)[\vec{A}/\vec{Y}]$, with $\vec Y\notin \FV{\Gamma}$.

    \begin{proof}
      Let $(~\cdot~)^\circ$ be a map from types to types defined as follows,
      $$\begin{array}{r@{~}c@{~}l@{\hspace{0.4cm}}r@{~}c@{~}l@{\hspace{0.4cm}}r@{~}c@{~}l}
        X^\circ &=& X &
        (U\to T)^\circ &=& U\to T &
        (\forall X.T)^\circ &=& T^\circ \\
        (\alpha\cdot T)^\circ &=&\alpha\cdot T^\circ & && &
        (T+R)^\circ &=&T^\circ+R^\circ
      \end{array}$$

      \noindent
      We need three intermediate results:
      \begin{enumerate}
        \item If $T\equiv R$, then $T^\circ\equiv R^\circ$.
        \item For any types $U, A$, there exists $B$ such that $(U[A/X])^\circ=U^\circ[B/X]$.
        \item For any types $V, U$, there exists $\vec A$ such that if $V\succeq^{\ve t}_{\V,\Gamma} \forall\vec X.U$, then $U^\circ\equiv V^\circ[\vec A/\vec X]$.
      \end{enumerate}
      {\it Proofs.}
      \begin{enumerate}
        \item Induction on the equivalence rules. We only give the basic cases since the inductive step, given by the context where the equivalence is applied, is trivial.
          \begin{itemize}
            \item $(1\cdot T)^\circ=1\cdot T^\circ\equiv T^\circ$.
            \item $(\alpha\cdot(\beta\cdot T))^\circ=\alpha\cdot(\beta\cdot T^\circ)\equiv(\alpha\times\beta)\cdot T^\circ=((\alpha\times\beta)\cdot T)^\circ$.
            \item $(\alpha\cdot T+\alpha\cdot R)^\circ=\alpha\cdot T^\circ+\alpha\cdot R^\circ\equiv\alpha\cdot(T^\circ+R^\circ)=(\alpha\cdot(T+R))^\circ$.
            \item $(\alpha\cdot T+\beta\cdot T)^\circ=\alpha\cdot T^\circ+\beta\cdot T^\circ\equiv(\alpha+\beta)\cdot T^\circ=((\alpha+\beta)\cdot T)^\circ$.
            \item $(T+R)^\circ=T^\circ+R^\circ\equiv R^\circ+T^\circ=(R+T)^\circ$.
            \item $(T+(R+S))^\circ=T^\circ+(R^\circ+S^\circ)\equiv (T^\circ+R^\circ)+S^\circ=((T+R)+S)^\circ$.
          \end{itemize}
        \item Structural induction on $U$.
          \begin{itemize}
            \item $U=\varu X$. Then $(\varu X[V/\varu X])^\circ=V^\circ=\varu X[V^\circ/\varu X]=\varu X^\circ[V^\circ/\varu X]$.
            \item $U=\varu Y$. Then $(\varu Y[A/X])^\circ=\varu Y=\varu Y^\circ[A/X]$.
            \item $U=V\to T$. Then $((V\to T)[A/X])^\circ=(V[A/X]\to T[A/X])^\circ=V[A/X]\to T[A/X]=(V\to T)[A/X]=(V\to T)^\circ[A/X]$.
            \item $U=\forall Y.V$. Then $((\forall Y.V)[A/X])^\circ=(\forall Y.V[A/X])^\circ=(V[A/X])^\circ$, which by the induction hypothesis is equivalent to $V^\circ[B/X]=(\forall Y.V)^\circ[B/X]$.
          \end{itemize}
        \item It suffices to show this for $V\succ^{\ve t}_{X,\Gamma} \forall\vec X.U$. Cases:
          \begin{itemize}
            \item $\forall\vec X.U\equiv\forall Y.V$, then notice that $(\forall\vec X.U)^\circ \equiv_{(1)}(\forall Y.V)^\circ=V^\circ$.
            \item $V\equiv\forall Y.W$ and $\forall\vec X.U\equiv W[A/X]$, then

              $(\forall\vec X.U)^\circ\equiv_{(1)}(W[A/X])^\circ\equiv_{(2)} W^\circ[B/X]=(\forall Y.W)^\circ[B/X]\equiv_{(1)}V^\circ[B/X]$.
          \end{itemize}
      \end{enumerate}
      Proof of the lemma. $U\to T\equiv(U\to T)^\circ$, by the intermediate result 3, this is equivalent to $(V\to R)^\circ[\vec A/\vec X]=(V\to R)[\vec A/\vec X]$.
    \end{proof}

    \recap{Lemma}{Scalars}{lem:scalars}
    For any context $\Gamma$, term $\ve t$, type $T$ and scalar
    $\alpha$, if
    $\Gamma\vdash\alpha\cdot\ve{t}: T$,
    then there exists a type $R$ such that
    $T\equiv\alpha\cdot R$ and
    $\Gamma\vdash\ve{t}: R$.
    Moreover, if the minimum size of the derivation of
    $\Gamma\vdash\alpha\cdot\ve t:T$ is $s$,
    then if $T=\alpha\cdot R$,
    the minimum size of the derivation of
    $\Gamma\vdash\ve t:R$
    is at most $s-1$, in other case, its minimum size is at most $s-2$.

    \begin{proof}
      We proceed by induction on the typing derivation.
    \item\parbox{4.5cm}{\prooftree\Gamma\vdash\alpha\cdot\ve t:\sui{n}\alpha_i\cdot U_i
      \justifies\Gamma\vdash\alpha\cdot\ve t:\sui{n}\alpha_i\cdot\forall{X}.U_i
      \using\forall_{I}
  \endprooftree}
  \parbox{8.9cm}{By the induction hypothesis $\sui{n}\alpha_i\cdot 
  U_i\equiv\alpha\cdot R$, and by Lemma~\ref{lem:typecharact}, 
  $R\equiv\suj{m}\beta_j\cdot V_j+\suk{h}\gamma_k\cdot\vara{X}_k$. So it is easy to see that $h=0$ and so $R\equiv\suj{m}\beta_j\cdot V_j$.
  Hence $\sui{n}\alpha_i\cdot 
  U_i\equiv\suj{m}\alpha\times\beta_j\cdot V_j$. Then by 
  Lemma~\ref{lem:equivforall}, 
  $\sui{n}\alpha_i\cdot\forall{X}.U_i\equiv\suj{m}
  \alpha\times\beta_j\cdot\forall{X}.V_j\equiv\alpha\cdot\suj{m}
  \beta_j\cdot\forall{X}.V_j$. In addition, by the induction hypothesis, 
  $\Gamma\vdash\ve t:R$ with a derivation of size $s-3$ (or $s-2$ if $n=1$), so 
  by rules $\forall_{{I}}$ and $\equiv$ (not needed if $n=1$), 
  $\Gamma\vdash\ve t:\suj{m}\beta_j\cdot\forall{X}.V_j$ in size $s-2$ (or 
$s-1$ in the case $n=1$).}
\medskip

\item\parbox{5cm}{\prooftree\Gamma\vdash\alpha\cdot\ve t:\sui{n}\alpha_i\cdot\forall{X}.U_i
  \justifies\Gamma\vdash\alpha\cdot\ve t:\sui{n}\alpha_i\cdot U_i[A/{X}]
  \using\forall_{{E}}
\endprooftree}
\parbox{8.4cm}{By the induction hypothesis 
  $\sui{n}\alpha_i\cdot\forall{X}.U_i\equiv\alpha\cdot R$, and by 
  Lemma~\ref{lem:typecharact}, $R\equiv\suj{m}\beta_j\cdot V_j+\suk{h}\gamma_k\cdot\vara{X}_k$. So it is easy to see that $h=0$ and so 
  $R\equiv\suj{m}\beta_j\cdot V_j$. Hence 
  $\sui{n}\alpha_i\cdot\forall{X}.U_i\equiv\suj{m}\alpha\times\beta_j\cdot 
  V_j$. Then by Lemma~\ref{lem:equivforall}, for each $V_j$, there exists $W_j$ 
  such that $V_j\equiv\forall{X}.W_j$, so 
  $\sui{n}\alpha_i\cdot\forall{X}.U_i\equiv\suj{m}
  \alpha\times\beta_j\cdot\forall{X}.W_j$. Then by the same lemma, 
  $\sui{n}\alpha_i\cdot U_i[A/{X}]\equiv\suj{m}\alpha\times\beta_j\cdot 
  W_j[A/{X}]\equiv\alpha\cdot\suj{m}\beta_j\cdot W_j[A/{X}]$. In 
  addition, by the induction hypothesis, $\Gamma\vdash\ve t:R$ with a derivation 
  of size $s-3$ (or $s-2$ if $n=1$), so by rules $\forall_{{E}}$ and 
  $\equiv$ (not needed if $n=1$), $\Gamma\vdash\ve t:\suj{m}\beta_j\cdot 
  W_j[A/{X}]$ in size $s-2$ (or 
$s-1$ in the case $n=1$).}
\medskip

  \item\parbox{3.1cm}{\prooftree\Gamma\vdash\ve t:T
      \justifies\Gamma\vdash\alpha\cdot\ve t:\alpha\cdot T
      \using\alpha_I
  \endprooftree}
  \parbox{10.3cm}{Trivial case.}
  \medskip

\item\parbox{4.3cm}{\prooftree\Gamma\vdash\alpha\cdot\ve t:T\qquad T\equiv R
    \justifies\Gamma\vdash\alpha\cdot\ve t:R
    \using\equiv
\endprooftree}
\parbox{9.1cm}{By the induction hypothesis $T\equiv\alpha\cdot S$, and 
  $\Gamma\vdash\ve t:S$. Notice that $R\equiv T\equiv\alpha\cdot S$.
  If 
  $T=\alpha\cdot S$, then it is derived with a minimum size of at most $s-2$. If 
  $T=R$, then the minimum size remains because the last $\equiv$ 
  rule is redundant. In other case, the sequent can be derived with minimum size 
  at most $s-1$.
\qedhere}
\end{proof}

\recap{Lemma}{Type for zero}{lem:termzero}
Let $\ve t=\ve 0$ or $\ve t=\alpha\cdot\ve 0$, then $\Gamma\vdash\ve t:T$ implies $T\equiv 0\cdot R$.

\begin{proof}
  We proceed by induction on the typing derivation.
\item\parbox{6.9cm}{
    \prooftree\Gamma\vdash\ve 0:T
    \justifies\Gamma\vdash \alpha\cdot\ve 0:0\cdot T
    \using\alpha_I
  \endprooftree
  \quad and\quad
  \prooftree\Gamma\vdash\ve t:T
  \justifies\Gamma\vdash\ve 0:0\cdot T
  \using 0_I
\endprooftree  
  }
  \parbox{6.5cm}{Trivial cases}
  \medskip

\item \parbox{3.7cm}{
    \prooftree\Gamma\vdash \ve t:\sui{n}\alpha_i\cdot U_i
    \justifies\Gamma\vdash \ve t:\sui{n}\alpha_i\cdot V_i
    \using\forall
\endprooftree}
\parbox{9.7cm}{$\forall$-rules ($\forall_I$ and
  $\forall_E$) have both the same structure as shown on the
  left. In both cases, by the induction hypothesis $\sui{n}\alpha_i\cdot 
  U_i\equiv 0\cdot R$, and by Lemma~\ref{lem:typecharact}, $R\equiv\suj{m}\beta_j\cdot W_j+\suk{h}\gamma_k\cdot\vara{X}_k$. It is easy to check that $h=0$, so $\sui{n}\alpha_i\cdot U_i\equiv 0\cdot\suj{m}\beta_j\cdot W_j\equiv\suj{m}0\cdot W_j$.
  Hence, using the same $\forall$-rule, we can derive
  $\Gamma\vdash\ve t:\suj{m} 0\cdot W'_j$,
  and by Lemma~\ref{lem:equivforall} we can ensure that 
  $\sui{n}\alpha_i\cdot V_i\equiv 0\cdot\suj{m}W_j'$.
}
\medskip

 \item \parbox{4.2cm}{
     \prooftree\Gamma\vdash \ve t:T\qquad T\equiv R
     \justifies\Gamma\vdash \ve t:R
     \using\equiv
 \endprooftree}
 \parbox{9.2cm}{By the induction hypothesis $R\equiv T\equiv 0\cdot S$.
 \qedhere}
\end{proof}

\recap{Lemma}{Sums}{lem:sums}
If $\Gamma\vdash\ve t+\ve r:S$, then $S\equiv T+R$ with $\Gamma\vdash\ve t:T$ 
and $\Gamma\vdash\ve r:R$.
Moreover, if the size of the derivation of $\Gamma\vdash\ve t+\ve r:S$ is $s$, 
then if $S=T+R$, the minimum sizes of the derivations of $\Gamma\vdash\ve t:T$ 
and $\Gamma\vdash\ve r:R$ are at most $s-1$, and if $S\neq T+R$, the 
minimum sizes of these derivations are at most $s-2$.

\begin{proof}
  We proceed by induction on the typing derivation.
\item \parbox{4cm}{
    \prooftree\Gamma\vdash\ve t+\ve r:\sui{n}\alpha_i\cdot U_i
    \justifies\Gamma\vdash\ve t+\ve r:\sui{n}\alpha_i\cdot V_i
    \using\forall
\endprooftree}
\parbox{9.4cm}{Rules $\forall_I$ and $\forall_E$ have both the same
  structure as shown on the left. In any case, by the induction hypothesis 
  $\Gamma\vdash\ve t:T$ and $\Gamma\vdash\ve r:R$ with 
  $T+R\equiv\sui{n}\alpha_i\cdot U_i$, and derivations of minimum size at most 
$s-2$ if the equality is true, or $s-3$ if these types are not equal.}

\noindent In the second case (when the types are not equal), there exists
$N,M\subseteq\{1,\dots,n\}$ with $N\cup M=\{1,\dots,n\}$ such that
\begin{align*}T\equiv\sum_{i\in N\setminus M}\alpha_i\cdot U_i+\sum_{i\in N\cap M}\alpha_i'\cdot U_i&\qquad\textrm{and}\\
  R\equiv\sum_{i\in M\setminus N}\alpha_i\cdot U_i+\sum_{i\in N\cap M}\alpha_i''\cdot U_i&
\end{align*}
where $\forall i\in N\cap M$, $\alpha_i'+\alpha_i''=\alpha_i$.
Therefore, using $\equiv$ (if needed) and the same $\forall$-rule, we get 
$\Gamma\vdash\ve t:\sum_{i\in N\setminus M}\alpha_i\cdot V_i+\sum_{i\in N\cap 
M}\alpha_i'\cdot V_i$ and $\Gamma\vdash\ve r:\sum_{i\in M\setminus 
N}\alpha_i\cdot V_i+\sum_{i\in N\cap M}\alpha_i''\cdot V_i$, with derivations of 
minimum size at most $s-1$.
\medskip

  \item \parbox{4.5cm}{
      \prooftree\Gamma\vdash\ve t+\ve r:S'\qquad S'\equiv S
      \justifies\Gamma\vdash\ve t+\ve r:S
      \using\equiv
  \endprooftree}
  \parbox{8.9cm}{By the induction hypothesis, $S\equiv S'\equiv T+R$ and we 
    can derive $\Gamma\vdash\ve t:T$ and $\Gamma\vdash\ve r:R$ with a minimum size 
  of at most $s-2$.}
  \medskip

\item \parbox{4.3cm}{
    \prooftree\Gamma\vdash\ve t:T\qquad\Gamma\vdash\ve r:R
    \justifies\Gamma\vdash\ve t+\ve r:T+R
    \using+_I
\endprooftree}
\parbox{9.1cm}{This is the trivial case.\qedhere}
\end{proof}

\recap{Lemma}{Applications}{lem:app}
If $\Gamma\vdash(\ve t)~\ve r:T$, then $\Gamma\vdash\ve t:\sui{n}\alpha_i\cdot\forall\vec{X}.(U\to T_i)$ and $\Gamma\vdash\ve r:\suj{m}\beta_j\cdot U[\vec{A}_j/\vec{X}]$
where $\sui{n}\suj{m}\alpha_i\times\beta_j\cdot T_i[\vec{A}_j/\vec{X}]\succeq^{(\ve t)\ve r}_{\V,\Gamma} T$ for some $\V$.

\begin{proof}
  We proceed by induction on the typing derivation.
\item \parbox{4cm}{
    \prooftree\Gamma\vdash(\ve t)~\ve r:\suk{o}\gamma_k\cdot V_k
    \justifies\Gamma\vdash(\ve t)~\ve r:\suk{o}\gamma_k\cdot W_k
    \using\forall
\endprooftree}
\parbox{9.4cm}{Rules $\forall_I$ and $\forall_E$ have both the same
  structure as shown on the left. In any case, by the induction hypothesis
  $\Gamma\vdash\ve t:\sui{n}\alpha_i\cdot\forall\vec X.(U\to T_i)$,
  $\Gamma\vdash\ve r:\suj{m}\beta_j\cdot U[\vec{A}_j/\vec{X}]$ and
  $\sui{n}\suj{m}\alpha_i\times\beta_j\cdot T_i[\vec{A}_j/\vec{X}]\succeq^{(\ve t)\ve r}_{\V,\Gamma}\suk{o}\gamma_k\cdot V_k\succeq^{(\ve t)\ve r}_{\V,\Gamma}\suk{o}\gamma_k\cdot W_k$.}
  \medskip

\item \parbox{4.3cm}{
    \prooftree\Gamma\vdash(\ve t)~\ve r:S\qquad S\equiv R
    \justifies\Gamma\vdash(\ve t)~\ve r:R
    \using\equiv
\endprooftree}
\parbox{9.1cm}{By the induction hypothesis
  $\Gamma\vdash\ve t:\sui{n}\alpha_i\cdot\forall\vec X.(U\to T_i)$,
  $\Gamma\vdash\ve r:\suj{m}\beta_j\cdot U[\vec{A}_j/\vec{X}]$ and
  $\sui{n}\suj{m}\alpha_i\times\beta_j\cdot T_i[\vec{A}_j/\vec{X}]\succeq^{(\ve t)\ve r}_{\V,\Gamma}S\equiv R$.}

  \medskip
\item
  \prooftree\Gamma\vdash\ve t:\sui{n}\alpha_i\cdot\forall\vec{X}.(U\to T_i)
  \quad
  \Gamma\vdash\ve r:\suj{m}\beta_j\cdot U[\vec{A}_j/\vec{X}]
  \justifies\Gamma\vdash(\ve t)~\ve r:\sui{n}\suj{m}\alpha_i\times\beta_j\cdot T_i[\vec{A}_j/\vec{X}]
  \using\to_E
\endprooftree
\quad This is the trivial case.\qedhere
\end{proof}

\recap{Lemma}{Abstractions}{lem:abs} If $\Gamma\vdash\lambda x.\ve t:T$, then $\Gamma,x:U\vdash\ve t:R$ where $U\to R\succeq^{\lambda x.\ve t}_{\V,\Gamma} T$ for some $\V$.

\begin{proof}
  We proceed by induction on the typing derivation.

\item \parbox{3.7cm}{
    \prooftree\Gamma\vdash\lambda x.\ve t:\sui{n}\alpha_i\cdot U_i
    \justifies\Gamma\vdash\lambda x.\ve t:\sui{n}\alpha_i\cdot V_i
    \using\forall
\endprooftree}
\parbox{9.7cm}{Rules $\forall_I$ and $\forall_E$ have both the same
  structure as shown on the left. In any case, by the induction hypothesis
  $\Gamma,x:U\vdash\ve t:R$, where $U\to R\succeq^{\lambda x.\ve t}_{\V,\Gamma}\sui{n}\alpha_i\cdot U_i\succeq^{\lambda x.\ve t}_{\V,\Gamma}\sui{n}\alpha_i\cdot V_i$.}
  \medskip

\item \parbox{4.3cm}{
    \prooftree\Gamma\vdash\lambda x.\ve t:R\qquad R\equiv T
    \justifies\Gamma\vdash\lambda x.\ve t:T
    \using\equiv
\endprooftree}
\parbox{9.1cm}{By the induction hypothesis
  $\Gamma,x:U\vdash\ve t:S$ where $U\to S\succeq^{\lambda x.\ve t}_{\V,\Gamma} R\equiv T$.}
  \medskip

\item\parbox{3.6cm}{
    \prooftree\Gamma,x:U\vdash\ve t:T
    \justifies\Gamma\vdash\lambda x.\ve t:U\to T
    \using\to_I
\endprooftree}
\parbox{9.8cm}{This is the trivial case.\qedhere}
\end{proof}

\recap{Lemma}{Basis terms}{lem:basevectors}
For any context $\Gamma$, type $T$ and basis term $\ve{b}$, if
$\Gamma\vdash\ve{b}: T$ then there exists a unit type $U$ such
that $T\equiv U$.

\begin{proof}
  By induction on the typing derivation.
\item \parbox{3.2cm}{
    \prooftree\Gamma\vdash\ve b:\sui{n}\alpha_i\cdot U_i
    \justifies\Gamma\vdash\ve b:\sui{n}\alpha_i\cdot V_i
    \using\forall
\endprooftree}
\parbox{10.2cm}{Rules $\forall_I$ and $\forall_E$ have both the same
  structure as shown on the left. In any case, by the induction hypothesis $U\equiv\sui{n}\alpha_i\cdot U_i\succeq^{\ve b}_{\V,\Gamma} \sui{n}\alpha_i\cdot V_i$, then by a straightforward case analysis, we can check that $\sui{n}\alpha_i\cdot V_i\equiv U'$.}
  \medskip

\item \parbox{3.8cm}{
    \prooftree\Gamma\vdash\ve b:R\qquad R\equiv T
    \justifies\Gamma\vdash\ve b:T
    \using\equiv
\endprooftree}
\parbox{9.6cm}{By the induction hypothesis $U\equiv R\equiv T$.}
\medskip

  \item\parbox{7.8cm}{
      \prooftree
      \justifies\Gamma,x:U\vdash x:U
      \using ax
    \endprooftree\quad or\quad
    \prooftree\Gamma,x:U\vdash\ve t:T
    \justifies\Gamma\vdash\lambda x.\ve t:U\to T
    \using\to_I
\endprooftree}
\parbox{5.6cm}{These two are the trivial cases.\qedhere}
\end{proof}

\recap{Lemma}{Substitution lemma}{lem:substitution}
For any term ${\ve t}$, basis term $\ve b$, term variable $x$, context $\Gamma$, types $T$, $U$, type variable $X$ and type $A$, where $A$ is a unit type if $X$ is a unit variables, otherwise $A$ is a general type, we have,
\begin{enumerate}
  \item\label{ap:it:substitutionTypes} if $\Gamma\vdash\ve{t}: T$, then $\Gamma[A/X]\vdash\ve{t}: T[A/X]$;
  \item\label{ap:it:substitutionTerms} if $\Gamma,x:U\vdash\ve t:T$, $\Gamma\vdash\ve b:U$ then $\Gamma\vdash\ve t[\ve b/x]: T$.
\end{enumerate}

\begin{proof}~
  \begin{enumerate}
    \item Induction on the typing derivation.
      \bigskip

      \parbox{3.1cm}{
        \prooftree
        \justifies\Gamma,x:U\vdash x:U
        \using ax
    \endprooftree}
    \parbox{9.4cm}{Notice that $\Gamma[A/X],x:U[A/X]\vdash x:U[A/X]$ can also be derived with the same rule.}
    \bigskip

    \parbox{2.5cm}{
      \prooftree\Gamma\vdash\ve t:T
      \justifies\Gamma\vdash\ve 0:0\cdot T
      \using 0_I
  \endprooftree}
  \parbox{10cm}{By the induction hypothesis $\Gamma[A/X]\vdash\ve t:T[A/X]$, so by rule $0_I$, $\Gamma[A/X]\vdash\ve 0:0\cdot T[A/X]=(0\cdot T)[A/X]$.}
  \bigskip

  \parbox{3.6cm}{
    \prooftree\Gamma,x:U\vdash\ve t:T
    \justifies\Gamma\vdash\lambda x.\ve t:U\to T
    \using\to_I
\endprooftree}
\parbox{8.9cm}{By the induction hypothesis $\Gamma[A/X],x:U[A/X]\vdash\ve t:T[A/X]$, so by rule $\to_I$, $\Gamma[A/X]\vdash\lambda x.\ve t:U[A/X]\to T[A/X]=(U\to T)[A/X]$.}
\bigskip

\prooftree\Gamma\vdash\ve t:\sui{n}\alpha_i\cdot\forall\vec Y.(U\to T_i)\qquad\Gamma\vdash\ve r:\suj{m}\beta_j\cdot U[\vec B_j/\vec Y]
\justifies\Gamma\vdash(\ve t)~\ve r:\sui{n}\suj{m}\alpha_i\times\beta_j\cdot T_i[\vec B_j/\vec Y]
\using\to_E
    \endprooftree

    By the induction hypothesis
    $\Gamma[A/X]\vdash\ve t:(\sui{n}\alpha_i\cdot\forall\vec Y.(U\to T_i))[A/X]$ and this type is equal to $\sui{n}\alpha_i\cdot\forall\vec Y.(U[A/X]\to T_i[A/X])$. Also $\Gamma[A/X]\vdash\ve r:(\suj{m}\beta_j\cdot U[\vec B_j/\vec Y])[A/X]=
    \suj{m}\beta_j\cdot U[\vec B_j/\vec Y][A/X]$. Since $\vec Y$ is bound, we can consider it is not in $A$. Hence $U[\vec B_j/\vec Y][A/X]=U[A/X][\vec B_j[A/X]/\vec Y]$, and so, by rule $\to_E$,
    \begin{align*}
      \Gamma[A/X]\vdash(\ve t)~\ve r&:\sui{n}\suj{m}\alpha_i\times\beta_j\cdot T_i[A/X][\vec B_j[A/X]/\vec Y]\\
      &=(\sui{n}\suj{m}\alpha_i\times\beta_j\cdot T_i[\vec B_j/\vec Y])[A/X]\ .
    \end{align*}
    \bigskip

    \parbox{6.1cm}{
      \prooftree\Gamma\vdash\ve t:\sui{n}\alpha_i\cdot U_i\qquad Y\notin\FV{\Gamma}
      \justifies\Gamma\vdash\ve t:\sui{n}\alpha_i\cdot\forall Y.U_i
      \using\forall_I
  \endprooftree}
  \parbox{6.4cm}{By the induction hypothesis,
    $\Gamma[A/X]\vdash\ve t:(\sui{n}\alpha_i\cdot
    U_i)[A/X]=\sui{n}\alpha_i\cdot U_i[A/X]$. Then, by
    rule $\forall_I$, $\Gamma[A/X]\vdash\ve
    t:\sui{n}\alpha_i\cdot \forall
    Y.U_i[A/X]=(\sui{n}\alpha_i\cdot\forall Y.U_i)[A/X]$
    (in the case $Y\in\FV{A}$, we can rename the free variable).}
    \bigskip

    \parbox{4.5cm}{
      \prooftree\Gamma\vdash\ve t:\sui{n}\alpha_i\cdot\forall Y.U_i
      \justifies\Gamma\vdash\ve t:\sui{n}\alpha_i\cdot U_i[B/Y]
      \using\forall_E
  \endprooftree}
  \parbox{8cm}{Since $Y$ is bounded, we can
    consider $Y\notin\FV{A}$. By the induction
    hypothesis $\Gamma[A/X]\vdash\ve
    t:(\sui{n}\alpha_i\cdot\forall
    Y.U_i)[A/X]=\sui{n}\alpha_i\cdot\forall
    Y.U_i[A/X]$. Then by rule $\forall_E$,
    $\Gamma[A/X]\vdash\ve t:\sui{n}\alpha_i\cdot
    U_i[A/X][B/Y]$. We can consider $X\notin\FV{B}$ (in
    other case, just take $B[A/X]$ in the $\forall$-elimination), hence $\sui{n}\alpha_i\cdot U_i[A/X][B/Y]=\sui{n}\alpha_i\cdot U_i[B/Y][A/X]$.}
    \bigskip

    \parbox{3.1cm}{
      \prooftree\Gamma\vdash\ve t:T
      \justifies\Gamma\vdash\alpha\cdot\ve t:\alpha\cdot\ve T
      \using\alpha_I
  \endprooftree}
  \parbox{9.4cm}{By the induction hypothesis $\Gamma[A/X]\vdash\ve t:T[A/X]$, so by rule $\alpha_I$, $\Gamma[A/X]\vdash\alpha\cdot\ve t:\alpha\cdot T[A/X]=(\alpha\cdot T)[A/X]$.}
  \bigskip

  \parbox{4.3cm}{
    \prooftree\Gamma\vdash\ve t:T\qquad\Gamma\vdash\ve r:R
    \justifies\Gamma\vdash\ve t+\ve r:T+R
    \using +_I
\endprooftree}
\parbox{8.2cm}{By the induction hypothesis $\Gamma[A/X]\vdash\ve t:T[A/X]$ and $\Gamma[A/X]\vdash\ve r:R[A/X]$, so by rule $+_I$, $\Gamma[A/X]\vdash\ve t+\ve r:T[A/X]+R[A/X]=(T+R)[A/X]$.}
\bigskip

\parbox{3.7cm}{
  \prooftree\Gamma\vdash\ve t:T\qquad T\equiv R
  \justifies\Gamma\vdash\ve t:R
  \using\equiv
\endprooftree
      }
      \parbox{8.8cm}{By the induction hypothesis $\Gamma[A/X]\vdash\ve t:T[A/X]$, and since $T\equiv R$, then $T[A/X]\equiv R[A/X]$, so by rule $\equiv$, $\Gamma[A/X]\vdash\ve t:R[A/X]$.}
      \bigskip

    \item We proceed by induction on the typing derivation of $\Gamma,x:U\vdash\ve t:T$.
      \begin{enumerate}
        \item Let $\Gamma,x:U\vdash\ve t:T$ as a consequence of rule $ax$. Cases:
          \begin{itemize}
            \item $\ve t=x$, then $T=U$, and so $\Gamma\vdash\ve t[\ve b/x]:T$
              and $\Gamma\vdash\ve b:U$ are the same sequent.
            \item $\ve t=y$. Notice that $y[\ve b/x]=y$. By Lemma~\ref{lem:weakening} $\Gamma,x:U\vdash y:T$ implies
              $\Gamma\vdash y:T$.
          \end{itemize}
        \item Let $\Gamma,x:U\vdash\ve t:T$ as a consequence of rule $0_I$,
          then $\ve t=\ve 0$ and $T=0\cdot R$, with $\Gamma,x:U\vdash\ve r:R$
          for some $\ve r$. By the induction hypothesis, $\Gamma\vdash\ve
          r[\ve b/x]:R$. Hence, by rule $0_I$,
          $\Gamma\vdash\ve 0:0\cdot R$.

        \item Let $\Gamma,x:U\vdash\ve t:T$ as a consequence of rule
          $\to_I$, then $\ve t=\lambda y.\ve r$ and $T=V\to R$, with
          $\Gamma,x:U,y:V\vdash\ve r:R$. 
          Since our system admits weakening (Lemma~\ref{lem:weakening}), the sequent $\Gamma,y:V\vdash\ve b:U$
          is derivable. 
          Then by the induction hypothesis,
          $\Gamma,y:V\vdash\ve r[\ve b/x]:R$, from where, by
          rule $\to_I$, we obtain $\Gamma\vdash\lambda y.\ve r[\ve
          b/x]:V\to R$. We are done since 
          $\lambda y.\ve r[\ve b/x]=(\lambda y.\ve r)[\ve b/x]$.

        \item Let $\Gamma,x:U\vdash\ve t:T$ as a consequence of rule
          $\to_E$, then 
          $\ve t=(\ve r)~\ve u$ and
          $T=\sui{n}\suj{m}\alpha_i\times\beta_j\cdot R_i[\vec B/\vec Y]$,
          with 
          $\Gamma,x:U
          \vdash
          \ve r:\sui{n}\alpha_i\cdot\forall\vec Y.(V\to T_i)$ 
          and 
          $\Gamma,x:U
          \vdash
          \ve u:\suj{m}\beta_j\cdot V[\vec B/\vec Y]$.
          By the induction hypothesis, 
          $\Gamma
          \vdash
          \ve r[\ve b/x]:\sui{n}\alpha_i\cdot\forall\vec Y.(V\to R_i)$ 
          and
          $\Gamma\vdash\ve u[\ve b/x]:\suj{m}\beta_j\cdot V[\vec B/\vec Y]$. 
          Then, by rule $\to_E$, 
          $\Gamma
          \vdash \ve r[\ve b/x])~\ve u[\ve b/x]:
          \sui{n}\suj{m}\alpha_i\times\beta_j\cdot R_i[\vec B/\vec Y]$.

        \item Let $\Gamma,x:U\vdash\ve t:T$ as a consequence of rule
          $\forall_I$. Then $T=\sui{n}\alpha_i\cdot\forall Y.V_i$, with
          $\Gamma,x:U\vdash\ve t:\sui{n}\alpha_i\cdot V_i$ and
          $Y\notin\FV{\Gamma}\cup\FV{U}$.  By the induction hypothesis,
          $\Gamma\vdash\ve t[\ve b/x]:\sui{n}\alpha_i\cdot V_i$. Then by rule
          $\forall_I$, $\Gamma\vdash\ve t[\ve
          b/x]:\sui{n}\alpha_i\cdot\forall Y.V_i$.

        \item Let $\Gamma,x:U\vdash\ve t:T$ as a consequence of rule
          $\forall_E$, then $T=\sui{n}\alpha_i\cdot V_i[B/Y]$, with
          $\Gamma,x:U\vdash\ve t:\sui{n}\alpha_i\cdot\forall Y.V_i$. By the
          induction hypothesis, $\Gamma\vdash\ve t[\ve
          b/x]:\sui{n}\alpha_i\cdot\forall Y.V_i$. By
          rule $\forall_E$, $\Gamma\vdash\ve t[\ve
          b/x]:\sui{n}\alpha_i\cdot V_i[B/Y]$.

        \item Let $\Gamma,x:U\vdash\ve t:T$ as a consequence of rule
          $\alpha_I$. Then $T=\alpha\cdot R$ and $\ve t=\alpha\cdot\ve r$,
          with $\Gamma,x:U\vdash\ve r:R$. By the induction hypothesis
          $\Gamma\vdash\ve r[\ve b/x]:R$. Hence by rule
          $\alpha_I$, $\Gamma\vdash\alpha\cdot\ve r[\ve b/x]:\alpha\cdot
          R$. Notice that $\alpha\cdot\ve r[\ve
          b/x]=(\alpha\cdot\ve r)[\ve b/x]$.

        \item Let $\Gamma,x:U\vdash\ve t:T$ as a consequence of rule
          $+_I$. Then $\ve t=\ve r+\ve u$ and $T=R+S$, with
          $\Gamma,x:U\vdash\ve r:R$ and $\Gamma,x:U\vdash\ve u:S$.  By the
          induction hypothesis, $\Gamma\vdash\ve r[\ve b/x]:R$ 
          and $\Gamma\vdash\ve u[\ve b/x]:S$. Then by
          rule $+_I$, $\Gamma\vdash\ve r[\ve b/x]+\ve u[\ve b/x]:R+S$. 
          Notice that $\ve r[\ve b/x]+\ve
          u[\ve b/x]=(\ve r+\ve u)[\ve b/x]$.

        \item Let $\Gamma,x:U\vdash\ve t:T$ as a consequence of rule
          $\equiv$. Then $T\equiv R$ and $\Gamma,x:U\vdash\ve t:R$.  By the
          induction hypothesis, $\Gamma\vdash\ve t[\ve b/x]:R$. 
          Hence, by rule $\equiv$, $\Gamma\vdash\ve t[\ve b/x]:T$. 
          \qedhere
      \end{enumerate}

  \end{enumerate}
\end{proof}

\subsection{Proof of Theorem~\ref{thm:subjectreduction}}\label{app:srproof}

\recap{Theorem}{Weak subject reduction}{thm:subjectreduction}
For any terms $\ve t$, $\ve t'$, any context $\Gamma$ and any type
$T$, if $\ve{t}\to_R\ve{t}'$ and $\Gamma\vdash \ve t: T$, then:
\begin{enumerate}
  \item if $R\notin$ Group F, then $\Gamma\vdash\ve t': T$;
  \item if $R\in$ Group F, then
    $\exists S\sqsupseteq T$ such that $\Gamma\vdash\ve t': S$ and
    $\Gamma\vdash\ve{t}: S$.
\end{enumerate}
\begin{proof}
  Let $\ve t\to_R\ve t'$ and $\Gamma\vdash\ve t:T$. We proceed by induction on the rewrite relation.
  \paragraph{Group E}~
  \begin{description}
    \item[$0\cdot\ve{t}\to\ve{0}$] Consider $\Gamma\vdash 0\cdot\ve t:T$. By Lemma~\ref{lem:scalars}, we have that $T\equiv 0\cdot R$ and $\Gamma\vdash\ve t:R$. Then, by rule $0_I$, $\Gamma\vdash\ve 0:0\cdot R$. We conclude using rule $\equiv$.
    \item[$1\cdot\ve{t}\to\ve{t}$] Consider $\Gamma\vdash 1\cdot\ve t:T$, then by Lemma~\ref{lem:scalars}, $T\equiv 1\cdot R$ and $\Gamma\vdash\ve t:R$. Notice that $R\equiv T$, so we conclude using rule $\equiv$.
    \item[$\alpha\cdot\ve{0}\to\ve{0}$] Consider $\Gamma\vdash\alpha\cdot\ve 0:T$, then by Lemma~\ref{lem:termzero}, $T\equiv 0\cdot R$. Hence by rules $\equiv$ and $0_I$, $\Gamma\vdash\ve 0:0\cdot 0\cdot R$ and so we conclude using rule $\equiv$.
    \item[$\alpha\cdot(\beta\cdot\ve{t})\to(\alpha\times\beta)\cdot\ve{t}$] Consider $\Gamma\vdash\alpha\cdot(\beta\cdot\ve{t}):T$. 
      By Lemma~\ref{lem:scalars}, $T\equiv\alpha\cdot R$ and $\Gamma\vdash\beta\cdot\ve t:R$. 
      By Lemma~\ref{lem:scalars} again, $R\equiv\beta\cdot S$ with $\Gamma\vdash\ve t:S$. Notice that $(\alpha\times\beta)\cdot S\equiv\alpha\cdot(\beta\cdot S)\equiv T$, hence by rules $\alpha_I$ and $\equiv$, we obtain $\Gamma\vdash(\alpha\times\beta)\cdot\ve t:T$.
    \item[$\alpha\cdot(\ve{t}+\ve{r})\to\alpha\cdot\ve{t}+\alpha\cdot\ve{r}$] Consider $\Gamma\vdash\alpha\cdot(\ve t+\ve r):T$.
      By Lemma~\ref{lem:scalars}, $T\equiv\alpha\cdot R$ and $\Gamma\vdash\ve t+\ve r:R$. By Lemma~\ref{lem:sums} $\Gamma\vdash\ve t:R_1$ and $\Gamma\vdash\ve r:R_2$, with $R_1+R_2\equiv R$. Then by rules $\alpha_I$ and $+_I$, $\Gamma\vdash\alpha\cdot\ve t+\alpha\cdot\ve r:\alpha\cdot R_1+\alpha\cdot R_2$. Notice that $\alpha\cdot R_1+\alpha\cdot R_2\equiv\alpha\cdot(R_1+R_2)\equiv\alpha\cdot R\equiv T$. We conclude by rule $\equiv$.
  \end{description}
  \paragraph{Group F}~
  \begin{description}
    \item[$\alpha\cdot\ve{t}+\beta\cdot\ve{t}\to(\alpha+\beta)\cdot\ve{t}$] Consider $\Gamma\vdash\alpha\cdot\ve{t}+\beta\cdot\ve{t}:T$, then by Lemma~\ref{lem:sums}, $\Gamma\vdash\alpha\cdot\ve t:T_1$ and $\Gamma\vdash\beta\cdot\ve t:T_2$ with $T_1+T_2\equiv T$. 
      Then by Lemma~\ref{lem:scalars}, $T_1\equiv\alpha\cdot R$ and $\Gamma\vdash\ve t:R$ and
      $T_2\equiv\beta\cdot S$. By rule $\alpha_I$, $\Gamma\vdash(\alpha+\beta)\cdot\ve t:(\alpha+\beta)\cdot R$. Notice that $(\alpha+\beta)\cdot R\sqsupseteq\alpha\cdot R+\beta\cdot S\equiv T_1+T_2\equiv T$.
    \item[$\alpha\cdot\ve{t}+\ve{t}\to(\alpha+1)\cdot\ve{t}$ and $R=\ve{t}+\ve{t}\to(1+1)\cdot\ve{t}$] The proofs of these two cases are simplified versions of the previous case.
    \item[$\ve{t}+\ve{0}\to\ve{t}$] Consider $\Gamma\vdash\ve t+\ve
      0:T$. By Lemma~\ref{lem:sums}, $\Gamma\vdash\ve t:R$ and
      $\Gamma\vdash\ve 0:S$ with $R+S\equiv T$. In addition, by
      Lemma~\ref{lem:termzero}, $S\equiv 0\cdot S'$. Notice that $R +
      0\cdot R\equiv R\sqsupseteq R+0\cdot S'\equiv R+S\equiv T$.
  \end{description}
  \paragraph{Group B}~
  \begin{description}
    \item[$(\lambda x.\ve{t})~\ve{b}\to\ve{t}\subst{\ve b}{x}$]
      Consider $\Gamma\vdash(\lambda x.\ve t)~\ve b:T$, then by
      Lemma~\ref{lem:app}, we have $\Gamma\vdash\lambda x.\ve
      t:\sui{n}\alpha_i\cdot\forall\vec X.(U\to R_i)$ and
      $\Gamma\vdash\ve b:\suj{m}\beta_j\cdot U[\vec A_j/\vec X]$ where
      $\sui{n}\suj{m}\alpha_i\times\beta_j\cdot R_i[\vec A_j/\vec
      X]\succeq^{(\lambda x.\ve t)\ve b}_{\V,\Gamma} T$.  However, we can simplify these types
      using Lemma~\ref{lem:basevectors}, and so we have
      $\Gamma\vdash\lambda x.\ve t:\forall\vec X.(U\to R)$ and
      $\Gamma\vdash\ve b:U[\vec A/\vec X]$ with $R[\vec A/\vec
      X]\succeq^{(\lambda x.\ve t)\ve b}_{\V,\Gamma} T$.  Note that
      $\vec{X}\not\in\FV{\Gamma}$ (from the arrow introduction rule). 
      Hence, by Lemma~\ref{lem:abs},
      $\Gamma,x:V\vdash\ve t:S$, with $V\to
      S\succeq^{\lambda x.\ve t}_{\V,\Gamma}\forall\vec X.(U\to R)$. Hence, by
      Lemma~\ref{lem:arrowscomp}, $U\equiv V[\vec B/\vec Y]$ and
      $R\equiv S[\vec B/\vec Y]$ with $\vec Y\notin\FV{\Gamma}$, so by
      Lemma~\ref{lem:substitution}(\ref{it:substitutionTypes}),
      $\Gamma,x:U\vdash\ve t:R$.  Applying
      Lemma~\ref{lem:substitution}(\ref{it:substitutionTypes}) once
      more, we have $\Gamma[\vec A/\vec X,x:U[\vec A/\vec X]\vdash\ve
      t[\ve b/x]:R[\vec A/\vec X]$.  Since $\vec X\not\in\FV{\Gamma}$,
      $\Gamma[\vec A/\vec X] = \Gamma$ and we can apply
      Lemma~\ref{lem:substitution}(\ref{it:substitutionTerms}) to get
      $\Gamma\vdash\ve t[\ve b/x]:R[\vec A/\vec X]\succeq^{(\lambda x.\ve t)\ve b}_{\V,\Gamma}
      T$. So, by Lemma~\ref{lem:subjectreductionofrelation}, $R[\vec A/\vec X]\succeq^{\ve t[\ve b/x]}_{\V,\Gamma} T$, which implies $\Gamma\vdash\ve t[\ve b/x]:T$.
  \end{description}

  \paragraph{Group A}~
  \begin{description}
    \item[$(\ve{t}+\ve{r})~\ve{u}\to(\ve{t})~\ve{u}+(\ve{r})~\ve{u}$] Consider $\Gamma\vdash(\ve t+\ve r)~\ve{u}:T$. Then by Lemma~\ref{lem:app},
      $\Gamma\vdash\ve t+\ve r:\sui{n}\alpha_i\cdot\forall\vec{X}.(U\to T_i)$ and
      $\Gamma\vdash\ve u:\suj{m}\beta_j.U[\vec A_j/\vec X]$ where
      $\sui{n}\suj{m}\alpha_i\times\beta_j\cdot T_i[\vec A_j/\vec X]\succeq^{(\ve t+\ve r)\ve u}_{\V,\Gamma}T$. Then by Lemma~\ref{lem:sums},
      $\Gamma\vdash\ve t:R_1$
      and
      $\Gamma\vdash\ve r:R_2$, with $R_1+R_2\equiv\sui{n}\alpha_i\cdot\forall\vec{X}.(U\to T_i)$. Hence, there exists $N_1, N_2\subseteq\{1,\dots,n\}$ with $N_1\cup N_2=\{1,\dots,n\}$ such that
      \begin{align*}
        R_1\equiv\sum\limits_{i\in N_1\setminus N_2}\alpha_i\cdot\forall\vec X.(U\to T_i)+
        \sum\limits_{i\in N_1\cap N_2}\alpha'_i\cdot\forall\vec X.(U\to T_i) & \mbox{\quad and}\\
        R_2\equiv\sum\limits_{i\in N_2\setminus N_1}\alpha_i\cdot\forall\vec X.(U\to T_i)+
        \sum\limits_{i\in N_1\cap N_2}\alpha''_i\cdot\forall\vec X.(U\to T_i) &
      \end{align*}
      where $\forall i\in N_1\cap N_2$, $\alpha'_i+\alpha''_i=\alpha_i$. Therefore, using $\equiv$ we get
      \begin{align*}
        \Gamma\vdash\ve t:\sum\limits_{i\in N_1\setminus N_2}\alpha_i\cdot\forall\vec X.(U\to T_i)+
        \sum\limits_{i\in N_1\cap N_2}\alpha'_i\cdot\forall\vec X.(U\to T_i) & \mbox{\quad and}\\
        \Gamma\vdash\ve r:\sum\limits_{i\in N_2\setminus N_1}\alpha_i\cdot\forall\vec X.(U\to T_i)+
        \sum\limits_{i\in N_1\cap N_2}\alpha''_i\cdot\forall\vec X.(U\to T_i) &
      \end{align*}
      So, using rule $\to_E$, we get
      \begin{align*}
        \Gamma\vdash(\ve t)~\ve u:\sum\limits_{i\in N_1\setminus N_2}\suj{m}\alpha_i\times\beta_j\cdot T_i[\vec A_j/\vec X]+
        \sum\limits_{i\in N_1\cap N_2}\suj{m}\alpha'_i\times\beta_j\cdot T_i[\vec A_j/\vec X] & \mbox{\quad and}\\
        \Gamma\vdash(\ve r)~\ve u:\sum\limits_{i\in N_2\setminus N_1}\suj{m}\alpha_i\times\beta_j\cdot T_i[\vec A_j/\vec X]+
        \sum\limits_{i\in N_1\cap N_2}\suj{m}\alpha''_i\times\beta_j\cdot T_i[\vec A_j/\vec X] &
      \end{align*}
      Finally, by rule $+_I$ we can conclude
      $\Gamma\vdash(\ve t)~\ve u+(\ve r)~\ve u:\sui{n}\suj{m}\alpha_i\times\beta_j\cdot T_i[\vec A_j/\vec X]\succeq^{(\ve t+\ve r)\ve u}_{\V,\Gamma} T$. Then by Lemma~\ref{lem:subjectreductionofrelation}, $:\sui{n}\suj{m}\alpha_i\times\beta_j\cdot T_i[\vec A_j/\vec X]\succeq^{(\ve t)\ve u+(\ve r)\ve u}_{\V,\Gamma} T$, so $\Gamma\vdash(\ve t)~\ve u+(\ve r)~\ve u:T$.

    \item[$(\ve{t})~(\ve{r}+\ve{u})\to(\ve{t})~\ve{r}+(\ve{t})~\ve{u}$]
      Consider $\Gamma\vdash(\ve{t})~(\ve{r}+\ve{u}):T$. By Lemma~\ref{lem:app},
      $\Gamma\vdash\ve t:\sui{n}\alpha_i\cdot\forall\vec{X}.(U\to T_i)$ and
      $\Gamma\vdash\ve r+\ve u:\suj{m}\beta_j.U[\vec A_j/\vec X]$ where
      $\sui{n}\suj{m}\alpha_i\times\beta_j\cdot T_i[\vec A_j/\vec X]\succeq^{(\ve t)(\ve r+\ve u)}_{\V,\Gamma} T$. Then by Lemma~\ref{lem:sums},
      $\Gamma\vdash\ve r:R_1$
      and
      $\Gamma\vdash\ve u:R_2$, with $R_1+R_2\equiv\suj{m}\beta_j.U[\vec A_j/\vec X]$. Hence, there exists $M_1, M_2\subseteq\{1,\dots,m\}$ with $M_1\cup M_2=\{1,\dots,m\}$ such that
      \begin{align*}
        R_1\equiv\sum\limits_{j\in M_1\setminus M_2}\beta_j.U[\vec A_j/\vec X]+
        \sum\limits_{j\in M_1\cap M_2}\beta'_j.U[\vec A_j/\vec X] & \mbox{\quad and}\\
        R_2\equiv\sum\limits_{j\in M_2\setminus M_1}\beta_j.U[\vec A_j/\vec X]+
        \sum\limits_{j\in M_1\cap M_2}\beta''_j.U[\vec A_j/\vec X] &
      \end{align*}
      where $\forall j\in M_1\cap M_2$, $\beta'_j+\beta''_j=\beta_j$. Therefore, using $\equiv$ we get
      \begin{align*}
        \Gamma\vdash\ve r:\sum\limits_{j\in M_1\setminus M_2}\beta_j.U[\vec A_j/\vec X]+
        \sum\limits_{j\in M_1\cap M_2}\beta'_j.U[\vec A_j/\vec X] & \mbox{\quad and}\\
        \Gamma\vdash\ve u:\sum\limits_{j\in M_2\setminus M_1}\beta_j.U[\vec A_j/\vec X]+
        \sum\limits_{j\in M_1\cap M_2}\beta''_j.U[\vec A_j/\vec X] &
      \end{align*}
      So, using rule $\to_E$, we get
      \begin{align*}
        \Gamma\vdash(\ve t)~\ve r:\sui{n}\sum\limits_{j\in M_1\setminus M_2}\alpha_i\times\beta_j\cdot T_i[\vec A_j/\vec X]+
        \sui{n}\sum\limits_{j\in M_1\cap M_2}\alpha_i\times\beta'_j\cdot T_i[\vec A_j/\vec X] & \mbox{\quad and}\\
        \Gamma\vdash(\ve t)~\ve u:\sui{n}\sum\limits_{j\in M_2\setminus M_1}\alpha_i\times\beta_j\cdot T_i[\vec A_j/\vec X]+
        \sui{n}\sum\limits_{j\in M_1\cap M_2}\alpha_i\times\beta''_j\cdot T_i[\vec A_j/\vec X] &
      \end{align*}
      Finally, by rule $+_I$ we can conclude
      $\Gamma\vdash(\ve t)~\ve r+(\ve t)~\ve u:\sui{n}\suj{m}\alpha_i\times\beta_j\cdot T_i[\vec A_j/\vec X]$.
      We finish the case with Lemma~\ref{lem:subjectreductionofrelation}.

    \item[$(\alpha\cdot\ve{t})~\ve{r}\to\alpha\cdot (\ve{t})~\ve{r}$] Consider
      $\Gamma\vdash(\alpha\cdot\ve t)~\ve r:T$. Then by Lemma~\ref{lem:app}, $\Gamma\vdash\alpha\cdot\ve t:\sui{n}\alpha_i\cdot\forall\vec X.(U\to T_i)$ and
      $\Gamma\vdash\ve r:\suj{m}\beta_j\cdot U[\vec A_j/\vec X]$, where
      $\sui{n}\suj{m}\alpha_i\times\beta_j\cdot T_i[\vec A_j/\vec X]\succeq^{(\alpha\cdot\ve t)\ve r}_{\V,\Gamma} T$.
      Then by Lemma~\ref{lem:scalars},
      $\sui{n}\alpha_i\cdot\forall\vec X.(U\to T_i)\equiv\alpha\cdot R$ and
      $\Gamma\vdash\ve t:R$.
      By Lemma~\ref{lem:typecharact}, $R\equiv\sui{n'}\gamma_i\cdot
      V_i+\suk{h}\eta_k\cdot\vara X_k$, however it is easy to see that
      $h=0$ because $R$ is equivalent to a sum of terms, where none of them is $\vara X$.
      So $R\equiv\sui{n'}\gamma_i\cdot V_i$. Without lost of generality (cf.~previous case), take $T_i\neq
      T_k$ for all $i\neq k$ and $h=0$, and notice that
      $\sui{n}\alpha_i\cdot\forall\vec X.(U\to T_i)\equiv\sui{n'}\alpha\times\gamma_i\cdot V_i$. Then by Lemma~\ref{lem:equivdistinctscalars}, there exists a permutation $p$ such that $\alpha_i=\alpha\times\gamma_{p(i)}$ and $\forall\vec X.(U\to T_i)\equiv V_{p(i)}$. Without lost of generality let $p$ be the trivial permutation, and so
      $\Gamma\vdash\ve t:\sui{n}\gamma_i\cdot\forall\vec X.(U\to T_i)$.
      Hence, using rule $\to_E$, $\Gamma\vdash(\ve t)~\ve r:\sui{n}\suj{m}\gamma_i\times\beta_j\cdot T_i[\vec A_j/\vec X]$. Therefore, by rule $\alpha_I$, $\Gamma\vdash\alpha\cdot(\ve t)~\ve r:\alpha\cdot\sui{n}\suj{m}\gamma_i\times\beta_j\cdot T_i[\vec A_j/\vec X]$. Notice that
      $\alpha\cdot\sui{n}\suj{m}\gamma_i\times\beta_j\cdot T_i[\vec A_j/\vec X]\equiv
      \sui{n}\suj{m}\alpha_i\times\beta_j\cdot T_i[\vec A_j/\vec X]$. We finish the case with Lemma~\ref{lem:subjectreductionofrelation}.

    \item[$(\ve{t})~(\alpha\cdot\ve{r})\to\alpha\cdot (\ve{t})~\ve{r}$]

      Consider
      $\Gamma\vdash(\ve t)~(\alpha\cdot\ve r):T$. Then by Lemma~\ref{lem:app},
      $\Gamma\vdash\ve t:\sui{n}\alpha_i\cdot\forall\vec X.(U\to T_i)$ and
      $\Gamma\vdash\alpha\cdot\ve r:\suj{m}\beta_j\cdot U[\vec A_j/\vec X]$, where
      $\sui{n}\suj{m}\alpha_i\times\beta_j\cdot T_i[\vec A_j/\vec X]\succeq^{(\ve t)(\alpha\cdot\ve r)}_{\V,\Gamma} T$.
      Then by Lemma~\ref{lem:scalars},
      $\suj{m}\beta_j\cdot U[\vec A_j/\vec X]\equiv\alpha\cdot R$ and
      $\Gamma\vdash\ve r:R$.
      By Lemma~\ref{lem:typecharact}, $R\equiv\suj{m'}\gamma_j\cdot
      V_j+\suk{h}\eta_k\cdot\vara X_k$, however it is easy to see that
      $h=0$ because $R$ is equivalent to a sum of terms, where none of them is $\vara X$.
      So $R\equiv\suj{m'}\gamma_j\cdot V_j$. Without lost of generality (cf.~previous case), take $A_j\neq
      A_k$ for all $j\neq k$, and notice that
      $\suj{m}\beta_j\cdot U[\vec A_j/\vec X]\equiv\suj{m'}\alpha\times\gamma_j\cdot V_j$. Then by Lemma~\ref{lem:equivdistinctscalars}, there exists a permutation $p$ such that $\beta_j=\alpha\times\gamma_{p(j)}$ and $U[\vec A_j/\vec X]\equiv V_{p(j)}$. Without lost of generality let $p$ be the trivial permutation, and so
      $\Gamma\vdash\ve r:\suj{m}\gamma_i\cdot U[\vec A_j/\vec X]$.
      Hence, using rule $\to_E$, $\Gamma\vdash(\ve t)~\ve r:\sui{n}\suj{m}\alpha_i\times\gamma_j\cdot T_i[\vec A_j/\vec X]$. Therefore, by rule $\alpha_I$, $\Gamma\vdash\alpha\cdot(\ve t)~\ve r:\alpha\cdot\sui{n}\suj{m}\alpha_i\times\gamma_j\cdot T_i[\vec A_j/\vec X]$. Notice that
      $\alpha\cdot\sui{n}\suj{m}\alpha_i\times\gamma_j\cdot T_i[\vec A_j/\vec X]\equiv
      \sui{n}\suj{m}\alpha_i\times\beta_j\cdot T_i[\vec A_j/\vec X]$. We finish the case with Lemma~\ref{lem:subjectreductionofrelation}.

    \item[$(\ve{0})~\ve{t}\to \ve{0}$] Consider $\Gamma\vdash(\ve 0)~\ve t:T$. By Lemma~\ref{lem:app}, $\Gamma\vdash\ve 0:\sui{n}\alpha_i\cdot\forall\vec X.(U\to T_i)$ and $\Gamma\vdash\ve t:\suj{m}\beta_j\cdot U[\vec A_j/\vec X]$, where
      $\sui{n}\suj{m}\alpha_i\times\beta_j\cdot T_i[\vec A_j/\vec X]\succeq^{(\ve 0)\ve t}_{\V,\Gamma} T$.
      Then by Lemma~\ref{lem:termzero}, $\sui{n}\alpha_i\cdot\forall\vec X.(U\to T_i)\equiv 0\cdot R$.
      By Lemma~\ref{lem:typecharact}, $R\equiv\sui{n'}\gamma_i\cdot
      V_i+\suk{h}\eta_k\cdot\vara X_k$, however, it is easy to see that $h=0$ and so $R\equiv\sui{n'}\gamma_i\cdot V_i$. Without lost of generality, take $T_i\neq T_k$ for all $i\neq k$, and notice that
      $\sui{n}\alpha_i\cdot\forall\vec X.(U\to T_i)\equiv\sui{n'}0\cdot V_i$. By Lemma~\ref{lem:equivdistinctscalars}, $\alpha_i=0$. Notice that by rule $\to_E$,
      $\Gamma\vdash(\ve 0)~\ve t:\sui{n}\suj{m}0\cdot T_i[\vec A_j/\vec X]$, hence by rules $0_I$ and $\equiv$, $\Gamma\vdash\ve 0:\sui{n}\suj{m}0\cdot T_i[\vec A_j/\vec X]\succeq^{(\ve 0)\ve t}_{\V,\Gamma} T$. By Lemma~\ref{lem:subjectreductionofrelation}, $\sui{n}\suj{m}0\cdot T_i[\vec A_j/\vec X]\succeq^{\ve 0}_{\V,\Gamma} T$, so $\Gamma\vdash\ve 0:T$.

    \item[$(\ve{t})~\ve{0}\to \ve{0}$]

      Consider $\Gamma\vdash(\ve t)~\ve 0:T$. By Lemma~\ref{lem:app}, $\Gamma\vdash\ve t:\sui{n}\alpha_i\cdot\forall\vec X.(U\to T_i)$ and $\Gamma\vdash\ve 0:\suj{m}\beta_j\cdot U[\vec A_j/\vec X]$, where
      $\sui{n}\suj{m}\alpha_i\times\beta_j\cdot T_i[\vec A_j/\vec X]\succeq^{(\ve t)\ve 0}_{\V,\Gamma}T$.
      Then by Lemma~\ref{lem:termzero}, $\suj{m}\beta_j\cdot U[\vec A_j/\vec X]\equiv 0\cdot R$.
      By Lemma~\ref{lem:typecharact}, $R\equiv\suj{m'}\gamma_j\cdot
      V_j+\suk{h}\eta_k\cdot\vara X_k$, however, it is easy to see that $h=0$ and so $R\equiv\suj{m'}\gamma_j\cdot V_j$. Without lost of generality, take $A_j\neq A_k$ for all $j\neq
      k$, and notice that
      $\suj{m}\beta_j\cdot U[\vec A_j/\vec X]\equiv\suj{m'}0\cdot V_j$. By Lemma~\ref{lem:equivdistinctscalars}, $\beta_j=0$. Notice that by rule $\to_E$,
      $\Gamma\vdash(\ve t)~\ve 0:\sui{n}\suj{m}0\cdot T_i[\vec A_j/\vec
      X]$, hence by rules $0_I$ and $\equiv$, $\Gamma\vdash\ve
      0:\sui{n}\suj{m}0\cdot T_i[\vec A_j/\vec X]\succeq^{(\ve t)\ve 0}_{\V,\Gamma}T$. By Lemma~\ref{lem:subjectreductionofrelation}, $\sui{n}\suj{m}0\cdot T_i[\vec A_j/\vec X]\succeq^{\ve 0} T$. Hence, $\Gamma\vdash\ve 0:T$.
  \end{description}

  \paragraph{Contextual rules}
  Follows from the generation lemmas, the induction hypothesis and the
  fact that $\sqsupseteq$ is congruent.
  \qedhere
\end{proof}

\section{Detailed proofs of lemmas and theorems in Section~\ref{sec:SN}}\label{app:SN}
\subsection{First lemmas}\label{app:SNf}
\xrecap{Lemma}{lem:RCop}
If $\bcal{A}$, $\bcal{B}$ and all the $\bcal{A}_i$'s are in $\RC$,
then so are $\bcal{A}\to\bcal{B}$, $\sum_i\bcal{A}_i$ and
$\cap_i\bcal{A}_i$.

\begin{proof}
  Before proving that these operators define reducibility candidates, we need the following result which simplifies its proof: a linear combination of strongly normalising terms, is strongly normalising. That is

  {\noindent\textbf{Auxiliary Lemma} (AL)\textbf{.}}
  If $\{\ve t_i\}_i$ are strongly normalising, then so is $F(\vec{\ve
  t})$ for any algebraic context $F$.

  \noindent  {\it Proof.}
  Let $\vec{\ve t}=\ve t_1,\dots,\ve t_n$.  We define
  two notions.
  \begin{itemize}
    \item
      A measure $s$ on $\vec{\ve t}$ defined as the the sum over
      $i$ of the sum of the lengths of all the possible rewrite sequences
      starting with $\ve t_i$.
    \item
      An algebraic measure $a$ over algebraic contexts $F(.)$
      defined inductively by $a(\ve t_i)=1$, $a(F(\vec{\ve t})+G(\vec{\ve
      t'}))=2+a(F(\vec{\ve t}))+a(G(\vec{\ve t'}))$, $a(\alpha\cdot
      F(\vec{\ve t}))=1+2\cdot a(F(\vec{\ve t}))$, $a(\ve 0)=0$.
  \end{itemize}
  We claim that for all linear algebraic contexts $F(\cdot)$ (in the
  sense of Remark~\ref{rem:linalgcontext}) and all strongly
  normalising terms ${\ve t}_i$ that are not linear combinations (that
  is, of the form $x$, $\lambda x.\ve r$ or $(\ve s)~\ve r$), the term
  $F(\vec{\ve t})$ is also strongly normalising.

  The claim is proven by induction on $s(\vec{\ve t})$ (the size is
  finite because $\ve t$ is SN, and because the rewrite system is
  finitely branching).
  \begin{itemize}
    \item If $s(\vec{\ve t})=0$. Then none of the $\ve t_i$ reduces.
      We show by induction on $a(F(\vec{\ve t}))$ that $F(\vec{\ve t})$
      is SN.
      \begin{itemize}
        \item If $a(F(\vec{\ve t}))=0$, then $F(\vec{\ve t})=\ve 0$ which is SN.
        \item Suppose it is true for all $F(\vec{\ve t})$ of algebraic
          measure less or equal to $m$, and consider $F(\vec{\ve t})$
          such that $a(F(\vec{\ve t}))=m+1$. Since the $\ve t_i$ are not
          linear combinations and they are in normal form, because $s(\vec{\ve 
          t})=0$, then $F(\vec{\ve
          t})$ can only reduce with a rule from Group E or a rule from
          group F. We show that those reductions are strictly decreasing
          on the algebraic measure, by a rule by rule analysis, and so,
          we can conclude by induction hypothesis.
          \begin{itemize}
            \item $0\cdot F(\vec{\ve t})\to\ve 0$. Note that $a(0\cdot
              F(\vec{\ve t}))=1> 0=a(\ve 0)$.
            \item $1\cdot F(\vec{\ve t})\to F(\vec{\ve t})$. Note that
              $a(1\cdot F(\vec{\ve t}))=1+2\cdot a(F(\vec{\ve t}))>
              a(F(\vec{\ve t}))$.
            \item $\alpha\cdot \ve 0\to\ve 0$. Note that $a(\alpha\cdot \ve 0)=1> 
              0=a(\ve 0)$.
            \item $\alpha\cdot (\beta\cdot F(\vec{\ve
              t}))\to(\alpha\times\beta)\cdot F(\vec{\ve t})$. Note that
              $a(\alpha\cdot (\beta\cdot F(\vec{\ve t})))=1+2\cdot
              (1+2\cdot a(F(\vec{\ve t})))> 1+2\cdot a(F(\vec{\ve
              t}))=a((\alpha\times\beta)\cdot F(\vec{\ve t}))$.
            \item $\alpha\cdot (F(\vec{\ve t})+G(\vec{\ve
              t'}))\to\alpha\cdot F(\vec{\ve t})+\alpha\cdot G(\vec{\ve
              t}')$. Note that $a(\alpha\cdot (F(\vec{\ve t})+G(\vec{\ve
              t}')))=5+2\cdot a(F(\vec{\ve t}))+2\cdot a(G(\vec{\ve
              t}'))> 4+2\cdot a(F(\vec{\ve t}))+2\cdot a(G(\vec{\ve
              t}'))=a(\alpha\cdot F(\vec{\ve t})+\alpha\cdot G(\vec{\ve
              t}'))$.
            \item $\alpha\cdot F(\vec{\ve t})+\beta\cdot F(\vec{\ve
              t})\to(\alpha+\beta)\cdot F(\vec{\ve t})$. Note that
              $a(\alpha\cdot F(\vec{\ve t})+\beta\cdot F(\vec{\ve
              t}))=4+4\cdot a(F(\vec{\ve t}))> 1+2\cdot a(F(\vec{\ve
              t}))=a((\alpha+\beta)\cdot F(\vec{\ve t}))$.
            \item $\alpha\cdot F(\vec{\ve t})+F(\vec{\ve
              t})\to(\alpha+1)\cdot F(\vec{\ve t})$.  Note that
              $a(\alpha\cdot F(\vec{\ve t})+F(\vec{\ve t}))=3+3\cdot
              a(F(\vec{\ve t}))>1+2\cdot a(F(\vec{\ve t}))=a\cdot
              ((\alpha+1)\cdot F(\vec{\ve t}))$.
            \item $F(\vec{\ve t})+F(\vec{\ve t})\to (1+1)\cdot F(\vec{\ve
              t})$. Note that $a\cdot (F(\vec{\ve t})+F(\vec{\ve
              t}))=2+2\cdot a(F(\vec{\ve t}))> 1+2\cdot a(F(\vec{\ve
              t}))=a\cdot ((1+1)\cdot F(\vec{\ve t}))$.
            \item $F(\vec{\ve t})+\ve 0\to F(\vec{\ve t})$. Note that
              $a\cdot (F(\vec{\ve t})+\ve 0)=2+a(F(\vec{\ve t}))>
              a(F(\vec{\ve t}))$.
            \item Contextual rules are trivial.
          \end{itemize}
      \end{itemize}
    \item Suppose it is true for $n$, then consider $\vec{\ve t}$ such
      that $s(\vec{\ve t})=n+1$. Again, we show that $F(\vec{\ve t})$ is
      SN by induction on $a(F(\vec{\ve t}))$.
      \begin{itemize}
        \item If $a(F(\vec{\ve t}))=0$, then $F(\vec{\ve t})=\ve 0$ which is SN.
        \item Suppose it is true for all $F(\vec{\ve t})$ of algebraic
          measure less or equal to $m$, and consider $F(\vec{\ve t})$ such
          that $a(F(\vec{\ve t}))=m+1$. Since the $\ve t_i$ are not linear
          combinations, $F(\vec{\ve t})$ can reduce in two ways:
          \begin{itemize}
            \item $F(\ve t_1,\ldots \ve t_i,\ldots \ve t_k)\to F(\ve
              t_1,\ldots \ve t'_i,\ldots \ve t_k)$ with $\ve t_i\to\ve t'_i$.
              Then $\ve t'_i$ can be written as
              $
              H(\ve r_1,\ldots\ve r_l)
              $
              for some algebraic context $H$, where the $\ve r_j$'s are not
              linear combinations.
              Note that
              \[
                \sum_{j=1}^l s(\ve r_j) \leq s(\ve t'_i) < s(\ve t_i).
              \]
              Define the context
              \begin{gather*}
                G(\ve t_1,\ldots,\ve t_{i-1},\ve u_1,\ldots \ve u_l,\ve
                t_{i+1},\ldots \ve t_k) =\hspace{10ex}\\
                \hspace{10ex}
                F(\ve t_1,\ldots,\ve t_{i-1},H(\ve
                u_1,\ldots \ve u_l),\ve t_{i+1},\ldots \ve t_k).
              \end{gather*}
              The term $F(\vec{\ve t})$ then reduces to the term
              \[
                G(\ve t_1,\ldots,\ve t_{i-1},\ve r_1,\ldots\ve r_l,\ve
                t_{i+1}\ldots \ve t_k),
              \]
              where
              \[
                s(\ve t_1,\ldots,\ve t_{i-1},\ve r_1,\ldots\ve r_l,\ve
                t_{i+1}\ldots \ve t_k) < s(\vec{\ve t}).
              \]
              Using the top induction hypothesis, we conclude that $F(\ve
              t_1,\ldots \ve t'_i,\ldots \ve t_k)$ is SN.
            \item $F(\vec{\ve t})\to G(\vec{\ve t})$, with $a(G(\vec{\ve
              t}))<a(F(\vec{\ve t}))$. Using the second induction
              hypothesis, we conclude that $G(\vec{\ve t})$ is SN
          \end{itemize}
          All the possible reducts of $F(\vec{\ve t})$ are SN: so is
          $F(\vec{\ve t})$.
      \end{itemize}
  \end{itemize}
  This closes the proof of the claim.  Now, consider any SN terms
  $\{\ve t_i\}_i$ and any algebraic context $G(\vec{\ve t})$. Each
  $\ve t_i$ can be written as an algebraic sum of $x$'s,
  $\lambda x.\ve s$'s and $(\ve r)\,\ve s$'s. The context
  $G(\vec{\ve t}) $ can then be written as $F(\vec{\ve t'})$ where
  none of the ${\ve t'}_i$ is a linear combination.
  From Remark~\ref{rem:linalgcontext}, there exists a linear algebraic
  context $F'(\vec{\ve t}';)$ where $\vec{\ve t}''$ is $\vec{\ve t}'$
  with possibly some repetitions, and where
  $F'(\vec{\ve t}'') = F(\vec{\ve t}')$, when considered as terms.
  The hypotheses of the claim are satisfied: $F'(\vec{\ve t}'')$ is
  therefore SN. Since it is ultimely equal to $G(\vec{\ve t})$,
  $G(\vec{\ve t})$ is also SN: the Auxiliary Lemma (AL) is valid.

  \medskip

  Now, we can prove Lemma~\ref{lem:RCop}
  First, we consider the case $\bcal{A}\to\bcal{B}$.
  \begin{description}
    \item[$\RCn_1$] We must show that all $\ve
      t\in\bcal{A}\to\bcal{B}$ are in $\SN$. We proceed by
      induction on the definition of $\bcal{A}\to\bcal{B}$.
      \begin{itemize}
        \item Assume that $\ve t$ is such that for $\ve r=\ve 0$ and $\ve 
          r=\ve b$, with $\ve b\in A$, then $(\ve t)\,\ve r\in\bcal B$. Hence by $\RCn_1$ 
          in $\bcal B$, $\ve t\in\SN$.
        \item Assume that $\ve t$ is closed neutral and that
          $\Red(\ve t)\subseteq\bcal{A}\to\bcal{B}$. By induction
          hypothesis, all the elements of $\Red(\ve t)$ are strongly
          normalising: so is $\ve t$.
        \item The last case is immediate: if $\ve t$ is the term
          $\ve 0$, it is strongly normalising.
      \end{itemize}
    \item[$\RCn_2$] We must show that if $\ve t\to\ve t'$ and $\ve
      t\in\bcal{A}\to\bcal{B}$, then $\ve
      t'\in\bcal{A}\to\bcal{B}$. We again proceed by induction on
      the definition of $\bcal{A}\to\bcal{B}$.
      \begin{itemize}
        \item Let $\ve t$ such that $(\ve t)\,\ve 0\in\bcal{B}$ and
          such that for all $\ve b\in\bcal{A}$, $(\ve t)~\ve
          b\in\bcal{B}$. Then by $\RCn_2$ in $\bcal{B}$, $(\ve
          t')~\ve 0\in\bcal{B}$ and $(\ve
          t')~\ve b\in\bcal{B}$, and so $\ve
          t'\in\bcal{A}\to\bcal{B}$.
        \item If $\ve t$ is closed neutral and $\Red(\ve
          t)\subseteq\bcal{A}\to\bcal{B}$, then $\ve
          t'\in\bcal{A}\to\bcal{B}$ since $\ve t'\in\Red(\ve t)$.
        \item If
          $\ve t=\ve 0$, it does not reduce.
      \end{itemize}
    \item[$\RCn_3$ and $\RCn_4$] Trivially true by definition.
  \end{description}
  Then we analyze the case $\sum_i\bcal{A}_i$.
  \begin{description}
    \item[$\RCn_1$] 
      We show that ``if $\ve t \in \sum_i\bcal{A}_i$ then $\ve t$ is
      stongly normalizing'' by structural induction on the
      justification of $\ve t \in \sum_i\bcal{A}_i$.
      \begin{itemize}
        \item Base case. $\ve t$ belongs to one of the $\bcal A_i$: it
          is then SN by $\RCn_1$.
        \item $\CC_1$. $\ve t$ is part of a list $\vec{\ve t'}$
          where $F(\vec{\ve t'})\in\sum_i\bcal{A}_i$ for some
          alg. context $F$. The induction hypothesis says that
          $F(\vec{\ve t'})$ is SN. This implies that $\ve t$ is also SN.
        \item $\CC_2$. $\ve t = F(\vec{\ve t'})$ where $F$ is an
          alg. context and $\ve t'_i\in\bcal{A}_i$. The result is
          obtained using the auxiliary lemma (AL) and $\RCn_1$ on the
          $\bcal{A}_i$'s.
        \item $\RCn_2$. $\ve t$ is such that $\ve s\to \ve t$ where $\ve
          s$ is SN. This implies that $\ve t$ is also SN.
        \item $\RCn_3$. $\ve t$ is closed neutral and
          $\Red(\ve t)\subseteq\sum_i\bcal{A}_i$, then $\ve t$ is
          strongly normalising since all elements of $\Red(\ve t)$ are
          strongly normalising.
      \end{itemize}
    \item[$\RCn_2$ and $\RCn_3$] Trivially true by definition.
    \item[$\RCn_4$] Since $\ve 0$ is an algebraic context, it
      is also in the set by $\CC_2$.
  \end{description}
  Finally, we prove the case $\cap_i\bcal{A}_i$.
  \begin{description}
    \item[$\RCn_1$] Trivial since for all $i$,
      $\bcal{A}_i\subseteq\SN$.
    \item[$\RCn_2$] Let $\ve t\in\cap_i\bcal{A}_i$, then $\forall
      i,\,\ve t\in\bcal{A}_i$ and so by $\RCn_2$ in $\bcal{A}_i$,
      $\Red(\ve t)\subseteq\bcal{A}_i$. Thus $\Red(\ve
      t)\subseteq\cap_i\bcal{A}_i$.
    \item[$\RCn_3$] Let $\ve t\in\Neutral$ and $\Red(\ve
      t)\subseteq\cap_i\bcal{A}$. Then $\forall_i,\,\Red(\ve
      t)\subseteq A_i$, and thus, by $\RCn_3$ in $\bcal{A}_i$,
      $\ve t\in\bcal{A}_i$, which implies $\ve
      t\in\cap_i\bcal{A}_i$.
    \item[$\RCn_4$] By $\RCn_4$, for all $i$, $\ve
      0\in\bcal{A}_i$. Therefore, $\ve 0\in\cap_i\bcal{A}_i$.
  \end{description}
  This concludes the proof of Lemma~\ref{lem:RCop}.
\end{proof}

\xrecap{Lemma}{lem:typedecomp}
Any type $T$, has a unique (up to $\equiv$) canonical decomposition $T\equiv\sui{n}\alpha_i\cdot\vara U_i$ such that for all $l,k$, $\vara U_l\not\equiv\vara U_k$.

\begin{proof} 
  By Lemma~\ref{lem:typecharact}, $T\equiv\sui{n}\alpha_i\cdot
  U_i+\suj{m}\beta_j\cdot\vara X_j$. Suppose that there exist
  $l,k$ such that $U_l\equiv U_k$. Then notice that $T\equiv
  (\alpha_l+\alpha_k)\cdot U_l+\sum_{i\neq l,k}\alpha_i\cdot U_i$.
  Repeat the process until there is no more $l,k$ such that $U_l\not\equiv U_k$. Proceed in the analogously to obtain a linear combination of different $\vara X_j$.
\end{proof}

\xrecap{Lemma}{lem:substRed}
For any types $T$ and $A$, variable $X$ and valuation $\rho$, we have
$\denot{T[A/X]}_\rho
=
\denot{T}_{\rho,(X_+,X_-)\mapsto(\denot{A}_{\bar\rho},\denot{A}_{\rho})}$
and
$\denot{T[A/X]}_{\bar\rho}
=
\denot{T}_{\bar\rho,(X_-,X_+)\mapsto(\denot{A}_{\rho},\denot{A}_{\bar\rho})}$.

\begin{proof}
  We proceed by structural induction on $T$. On each case we only show the case of $\rho$ since the $\bar\rho$ case follows analogously.
  \begin{itemize}
    \item $T=X$. Then
      $\denot{X[A/X]}_\rho
      =
      \denot{A}_\rho
      =
      \denot{X}_{\rho,(X_+,X_-)\mapsto(\denot{A}_{\bar\rho},\denot{A}_{\rho})}$.
    \item $T=Y$. Then
      $\denot{Y[A/X]}_\rho
      =
      \denot{Y}_\rho
      =
      \rho_+(Y)
      =
      \denot{Y}_{\rho,(X_+,X_-)\mapsto(\denot{A}_{\bar\rho},\denot{A}_{\rho})}$.
    \item $Y=U\to R$. Then
      $\denot{(U\to R)[A/X]}_\rho
      =
      \denot{U[A/X]}_{\bar\rho}\to\denot{R[A/X]}_\rho$. By the induction hypothesis, we have
      $\denot{U[A/X]}_{\bar\rho}\to\denot{R[A/X]}_\rho=
      \denot{U}_{\bar\rho,(X_-,X_+)\mapsto(\denot{A}_{\rho},\denot{A}_{\bar\rho})}
      \to
      \denot{R}_{\rho,(X_+,X_-)\mapsto(\denot{A}_{\bar\rho},\denot{A}_{\rho})}
      =
      \denot{U\to R}_{\rho,(X_+,X_-)\mapsto(\denot{A}_{\bar\rho},\denot{A}_{\rho})}$.
    \item $U=\forall Y.V$. Then
      $\denot{(\forall Y.V)[A/X]}_\rho
      =
      \denot{\forall Y.V[A/X]}_\rho$
      which by definition is equal to
      $\cap_{\bcal B\in\RC}
      \denot{V[A/X]}_{\rho,(Y_+,Y_-)\mapsto(\bcal B,\bcal B)}$
      and this, by the induction hypothesis, is equal to
      $\cap_{\bcal B\in\RC}
      \denot{V}_{\rho,(Y_+,Y_-)\mapsto(\bcal B,\bcal B),(X_+,X_-)\mapsto(\denot{A}_{\bar\rho},\denot{A}_{\rho})}
      =
      \denot{\forall Y.V}_{\rho,(X_+,X_-)\mapsto(\denot{A}_{\bar\rho},\denot{A}_{\rho})}$.

    \item $T$ of canonical decomposition $\sum_i\alpha_i\cdot\vara U_i$.
      Then
      $\denot{T[A/X]}_\rho
      =
      \sum_i\denot{\vara U_i[A/X]}_\rho$,
      which by induction hypothesis is
      $\sum_i\denot{\vara U_i}_{\rho,(X_+,X_-)\mapsto(\denot{A}_{\bar\rho},\denot{A}_{\rho})}
      =\denot{T}_{\rho,(X_+,X_-)\mapsto(\denot{A}_{\bar\rho},\denot{A}_{\rho})}$.
      \qedhere
  \end{itemize}
\end{proof}

\subsection{Proof of the Adequacy Lemma (\ref{lem:SNadeq})}\label{app:adequacy}
We need the following results first.

\begin{lemma}\label{lem:polar2appendix}
  For any type $T$, 
  if
  $\rho=(\rho_+,\rho_-)$ and
  $\rho'=(\rho'_+,\rho'_-)$
  are two valid valuations over $\FV{T}$ such that
  $\forall X$, $\rho'_-(X)\subseteq\rho_-(X)$ and
  $\rho_+(X)\subseteq\rho'_+(X)$,
  then we have
  $\denot{T}_{\rho}\subseteq\denot{T}_{\rho'}$ and
  $\denot{T}_{\bar\rho'}\subseteq\denot{T}_{\bar\rho}$.
\end{lemma}
\begin{proof}
  Structural induction on $T$.
  \begin{itemize}
    \item $T=X$. Then
      $\denot{X}_\rho=\rho_+(X)\subseteq\rho'_+(X)=\denot{X}_{\rho'}$
      and
      $\denot{X}_{\bar\rho'}=\rho'_-(X)\subseteq\rho_-(X)=\denot{X}_{\bar\rho}$.
    \item $T=U\to R$. Then
      $\denot{U\to R}_\rho=\denot{U}_{\bar{\rho}}\to\denot{R}_\rho$ and
      $\denot{U\to R}_{\bar\rho'}=\denot{U}_{\rho'}\to\denot{R}_{\bar\rho'}$.
      By the induction hypothesis
      $\denot{U}_{\bar{\rho}'}\subseteq\denot{U}_{\bar{\rho}}$ ,
      $\denot{U}_\rho\subseteq\denot{U}_{\rho'}$,
      $\denot{R}_\rho\subseteq\denot{R}_{\rho'}$ and
      $\denot{R}_{\bar\rho'}\subseteq\denot{R}_{\bar\rho}$.
      We proceed by induction on the definition of $\to$ to show that
      $\forall\ve t\in\denot{U}_{\bar{\rho}}\to\denot{R}_\rho$, then
      $\ve t\in\denot{U}_{\bar{\rho}'}\to\denot{R}_{\rho'}=\denot{U\to R}_{\rho'}$
      \begin{itemize}
        \item Let
          $\ve t\in\{\ve t\,| (\ve t)~\ve 0\in\denot{R}_{\rho}$ and
          $\forall\ve b\in\denot{U}_{\bar\rho}, (\ve r)~\ve b\in\denot{R}_{\rho}\}$.
          Notice that $(\ve t)~\ve 0\in\denot{R}_\rho\subseteq\denot{R}_{\rho'}$.
          Also,
          $\forall\ve b\in\denot{U}_{\bar\rho'},\,\ve b\in\denot{U}_{\bar\rho}$
          and then $(\ve t)~\ve b\in\denot{R}_\rho\subseteq\denot{R}_{\rho'}$.
        \item Let $\Red(\ve t)\in\denot{U\to R}_{\rho}$ and $\ve
          t\in\Neutral$. By the induction hypothesis $\Red(\ve
          t)\in\denot{U\to R}_{\rho'}$ and so, by $\RCn_3$, $\ve
          t\in\denot{U\to R}_{\rho'}$.
        \item Let $\ve t=\ve 0$. By $\RCn_4$, $\ve 0$ is in any
          reducibility candidate, in particular it is in $\denot{U\to
          R}_{\rho'}$.
      \end{itemize}
      Analogously, $\forall\ve
      t\in\denot{U}_{\rho'}\to\denot{R}_{\bar\rho'}$, $\ve
      t\in\denot{U}_\rho\to\denot{R}_{\bar\rho}=\denot{U\to R}_\rho$.
    \item $T=\forall X.U$. Then $\denot{\forall X.U}_{\rho} =
      \cap_{\bcal{A}\in\RC}\denot{U}_{\rho,(X_+,X_-)\mapsto(\bcal{A},\bcal{A})}$.
      By the induction hypothesis we have
      $\denot{U}_{\rho,(X_+,X_-)\mapsto(\bcal{A},\bcal{A})}\subseteq\denot{U}_{\rho',(X_+,X_-)\mapsto{(\bcal A,\bcal A)}}$, 
      Hence we have that
      $\cap_{\bcal{A}\in\RC}\denot{U}_{\rho,(X_+,X_-)\mapsto(\bcal{A},\bcal{A})}\subseteq\cap_{\bcal{A}\in\RC}\denot{U}_{\rho',(X_+,X_-)\mapsto(\bcal{A},\bcal{A})}=\denot{\forall X.U}_{\rho'}$.
      The proof for the case
      $\denot{\forall X.U}_{\bar\rho'}\subseteq\denot{\forall X.U}_{\bar\rho}$
      is analogous.
    \item $T\equiv\sum_i\alpha_i\cdot\vara U_i$ and $T\not\equiv\vara U$. Then $\denot{T}_{\rho}=\sum_i\denot{\vara U_i}_{\rho}$. By the
      induction hypothesis
      $\denot{\vara U_i}_{\rho}\subseteq\denot{\vara U_i}_{\rho'}$. We proceed by
      induction on the justification of $\ve t \in
      \sum_i\denot{\vara  U_i}_{\rho}$ to show that if
      $\ve t\in\sum_i\denot{\vara U_i}_{\rho}$ then $\ve t\in\sum_i\denot{\vara U_i}_{\rho'}$.
      \begin{itemize}
        \item Base case. $\ve t$ belongs to one of the
          $\denot{\vara U_i}_{\rho}$. We conclude using the fact that
          $\denot{\vara U_i}_{\rho}\subseteq\denot{\vara U_i}_{\rho'}$:
          $\ve t$ then belongs to
          $\denot{\vara U_i}_{\rho'} \subseteq \sum_i\denot{\vara
          U_i}_{\rho'}$.
        \item $\CC_1$.  $\ve t$ belongs to
          $\sum_i\denot{\vara U_i}_{\rho}$ because it is in a list
          $\vec{\ve t'}$ where
          $F(\vec{\ve t'})\in\sum_i\denot{\vara U_i}_{\rho}$ for some
          alg. context $F$. Induction hypothesis says that
          $F(\vec{\ve t'})\in\sum_i\denot{\vara U_i}_{\rho'}$. We then
          get $\ve t\in\sum_i\denot{\vara U_i}_{\rho'}$ using $\CC_1$.
        \item $\CC_2$.  Let $\ve t=F(\vec{\ve r})$ where $F$ is an
          algebraic context and
          $\ve r_i\in\denot{\vara U_i}_{\bar{\rho}}$.  Note that by
          induction hypothesis
          $\forall \ve r_i\in\denot{\vara U_i}_{\rho}$,
          $\ve r_i\in\denot{\vara U_i}_{\rho'}$ and so
          $F(\vec{\ve r})\in\sum_i\denot{\vara
          U_i}_{\rho'}$.
        \item $\RCn_2$.
          $\ve t'\to\ve t$ and $\ve t'\in\sui{n}\denot{\vara U_i}_{\rho'}$.
          Invoking $\RCn_2$, $\ve t\in\sui{n}\denot{\vara U_i}_{\rho'}$.
        \item $\RCn_3$.
          $\Red(\ve t)\subseteq\denot{T}_{\rho}$ and
          $\ve t\in\Neutral$.
          By the induction hypothesis
          $\Red(\ve t)\subseteq\denot{T}_{\rho'}$
          and so, by $\RCn_3$,
          $\ve t\in\denot{T}_{\rho'}$.
      \end{itemize}
      The case $\denot{T}_{\bar\rho'}\subseteq\denot{T}_{\bar\rho}$ is
      analogous.\qedhere
  \end{itemize}
\end{proof}

\begin{lemma}\label{lem:polar1appendix}
  For any type $T$, if $\rho=(\rho_+,\rho_-)$ is a valid valuation over $\FV{T}$, then
  we have $\denot{T}_{\bar{\rho}}\subseteq\denot{T}_{\rho}$.
\end{lemma}
\begin{proof}
  Structural induction on $T$.
  \begin{itemize}
    \item $T=X$. Then
      $\denot{T}_{\bar{\rho}}=\rho_-(X)\subseteq\rho_+(X)=\denot{T}_\rho$.
    \item $T=U\to R$. Then $\denot{U\to
      R}_{\bar{\rho}}=\denot{U}_\rho\to\denot{R}_{\bar{\rho}}$. By the
      induction hypothesis
      $\denot{U}_{\bar{\rho}}\subseteq\denot{U}_\rho$ and
      $\denot{R}_{\bar{\rho}}\subseteq\denot{R}_\rho$. We must show that
      $\forall\ve t\in\denot{U\to R}_{\bar{\rho}}$, $\ve t\in\denot{U\to
      R}_\rho$. Let $\ve t\in\denot{U\to
      R}_{\bar{\rho}}=\denot{U}_\rho\to\denot{R}_{\bar{\rho}}$. We
      proceed by induction on the definition of~$\to$.
      \begin{itemize}
        \item Let $\ve t\in\{\ve t\,|(\ve t)\,\ve 0\in\denot{R}_{\bar{\rho}}$ and 
          $\forall\ve b\in\denot{U}_\rho, (\ve
          t)~\ve b\in\denot{R}_{\bar{\rho}}\}$. Notice that $(\ve t)\,\ve 
          0\in\denot{R}_{\bar{\rho}}\subseteq\denot{R}_{\rho}$ and forall $\ve
          b\in\denot{U}_{\bar{\rho}}$, $\ve b\in\denot{U}_\rho$, and so
          $(\ve t)~\ve
          b\in\denot{R}_{\bar{\rho}}\subseteq\denot{R}_\rho$. Thus $\ve
          t\in\denot{U}_{\bar{\rho}}\to\denot{R}_\rho=\denot{U\to
          R}_\rho$.
        \item Let $\Red(\ve t)\in\denot{U\to R}_{\bar{\rho}}$ and $\ve
          t\in\Neutral$. By the induction hypothesis $\Red(\ve
          t)\in\denot{U\to R}_\rho$ and so, by $\RCn_3$, $\ve
          t\in\denot{U\to R}_\rho$.
        \item Let $\ve t=\ve 0$. By $\RCn_4$, $\ve 0$ is in any
          reducibility candidate, in particular it is in $\denot{U\to
          R}_\rho$.
      \end{itemize}
    \item $T=\forall X.U$. Then 
      $\denot{\forall 
      X.U}_{\bar{\rho}}=\cap_{\bcal{A}\in\RC}\denot{U}_{\bar{\rho}, 
      (X_+,X_-)\mapsto(\bcal{A},\bcal{A})}$. By
      the induction hypothesis
      \[
        \denot{U}_{\bar{\rho},(X_+,X_-)\mapsto(\bcal A,\bcal
      A)}\subseteq\denot{U}_{\rho,(X_+,X_-)\mapsto(\bcal A,\bcal
      A)}.
    \]
    So
    $\cap_{\bcal{A}\in\RC}\denot{U}_{\bar{\rho},(X_+,
    X_-)\mapsto(\bcal{A},\bcal{A})}\subseteq\cap_{\bcal{A}\in\RC}
    \denot{U}_{\rho,(X_+,X_-)\mapsto(\bcal{A},\bcal{A})}$ which is
    $\denot{\forall X.U}_\rho$  by
    definition.
  \item $T\equiv\sum_i\alpha_i\cdot\vara U_i$ and
    $T\not\equiv\vara U$. Then
    $\denot{T}_{\bar{\rho}}=\sum_i\denot{\vara U_i}_{\bar{\rho}}$. By
    the induction hypothesis
    $\denot{\vara U_i}_{\bar{\rho}}\subseteq\denot{\vara U_i}_\rho$.
    We proceed by induction on the justification of
    $\ve t\in\sum_i\denot{\vara U_i}_{\bar\rho}$ to show that
    $\ve t\in\sum_i\denot{\vara U_i}_{\bar\rho}$ imples
    $\ve t\in\sum_i\denot{\vara U_i}_{\rho}$.
    \begin{itemize}
      \item Base case. $\ve t\in\denot{\vara U_i}_{\bar\rho}$ for some
        $i$: by induction hypothesis $\ve t\in\denot{\vara
        U_i}_{\rho}$, which is included in $\sum_i\denot{\vara U_i}_{\rho}$.
      \item $\CC_1$.
        $\ve t\in\sum_i\denot{\vara U_i}_{\bar\rho}$ because it is in a list
        $\vec{\ve t'}$ where
        $F(\vec{\ve t'})\in\sum_i\denot{\vara U_i}_{\bar\rho}$ for some
        alg. context $F$. Induction hypothesis says that
        $F(\vec{\ve t'})\in\sum_i\denot{\vara U_i}_{\rho}$. We then
        get $\ve t\in\sum_i\denot{\vara U_i}_{\rho}$ using $\CC_1$.
      \item $\CC_2$.
        Let $\ve t=F(\vec{\ve r})$ where $F$ is an algebraic context
        and $\ve r_i\in\denot{\vara U_i}_{\bar{\rho}}$.  By
        induction hypothesis $\forall \ve r\in\denot{\vara U_i}_{\bar{\rho}}$,
        $\ve r\in\denot{\vara U_i}_{\rho}$ and so the result holds by $\CC_2$.
      \item $\RCn_2$. 
        Let
        $\ve t\in\sum_i\denot{\vara U_i}_{\bar{\rho}}$
        and
        $\ve t\to\ve t'$.
        By the induction hypothesis
        $\ve t\in\sum_i\denot{\vara U_i}_{\rho}$,
        hence by $\RCn_2$,
        $\ve t'\in\sum_i\denot{\vara U_i}_{\rho}$.
      \item $\RCn_3$.
        Let
        $\Red(\ve t)\in\sum_i\denot{\vara U_i}_{\bar{\rho}}$
        and
        $\ve t\in\Neutral$.
        By the induction hypothesis
        $\Red(\ve t)\in\sum_i\denot{\vara U_i}_\rho$
        and so, by $\RCn_3$,
        $\ve t\in\sum_i\denot{\vara U_i}_\rho$.\qedhere
    \end{itemize}
\end{itemize}
\end{proof}

\begin{lemma}
  \label{lem:sumrcappendix}\label{lem:alphatinsigma}
  Let $\{\bcal A_{i}\}_{i=1\cdots n}$ be a family of reducibility
  candidates.  If $\ve s$ and $\ve t$ both belongs to
  $\sui{n}\bcal A_{i}$, then so does $\ve s+\ve t$. Similarly, if
  $\ve t\in\sui{n}\bcal A_i$, then for any $\alpha$,
  $\alpha\cdot\ve t\in\sui{n}\bcal A_i$.
\end{lemma}

\begin{proof}
  Direct corollary of the closure under $\CC_2$.
\end{proof}

\begin{lemma}
  \label{lem:applircappendix}
  Suppose that $\lambda x.\ve s\in\bcal A\to\bcal B$ and $\ve b\in\bcal A$, then
  $(\lambda x.\ve s)\,\ve b\in\bcal B$.
\end{lemma}

\begin{proof}
  Induction on the definition of $\bcal A\to\bcal B$.
  \begin{itemize}
    \item If $\lambda x.\ve s$ is in
      $\{\ve t~|~(\ve t)~\ve 0\in\bcal B$ and 
      $\forall\ve b\in\bcal A, (\ve t)\,\ve b \in \bcal B\}$, then it is trivial
    \item $\lambda x.\ve s$ cannot be in $\bcal A\to\bcal B$ by the closure under 
      $\RCn_3$, because it is not neutral, neither by the closure under $\RCn_4$, 
      because it is not the term $\ve 0$.\qedhere
  \end{itemize}
\end{proof}
\medskip

\begin{remark}\rm
  For the proof of adequacy, we show in the following lemma
  that $\denot{\forall\vec X.U}_\rho$ can be equivalently defined as a
  more general intersection, provided that $\rho$ is valid.
\end{remark}

\begin{lemma}\label{lem:betterforall}
  Suppose that $\rho=(\rho_+,\rho_-)$ is a valid valuation. 
  Then
  \[
    \cap_{\bcal{B}\subseteq\bcal{A}\in\RC}\denot{U}_{\rho,(X_+,X_-)\mapsto (\bcal{A},\bcal{B})} = \cap_{\bcal{A}\in\RC}\denot{U}_{\rho,(X_+,X_-)\mapsto(\bcal{A},\bcal{A})}.
  \]
\end{lemma}

\begin{proof}
  Suppose that
  $\ve t \in
  \cap_{\bcal{B}\subseteq\bcal{A}\in\RC}\denot{U}_{\rho,(X_+,X_-)\mapsto (\bcal{A},\bcal{B})}$,
  and pick any $\bcal{A}\in\RC$. Let $\bcal{B}:=\bcal{A}$: we have
  $\bcal{B}\subseteq\bcal{A}$, so
  $\ve t\in \denot{U}_{\rho,(X_+,X_-)\mapsto (\bcal{A},\bcal{B})}$,
  and then $\ve t\in \denot{U}_{\rho,(X_+,X_-)\mapsto (\bcal{A},\bcal{A})}$. Since this is the case for all
  $\bcal{A}\in\RC$, we conclude that 
  \[
    \cap_{\bcal{B}\subseteq\bcal{A}\in\RC}\denot{U}_{\rho,(X_+,X_-)\mapsto (\bcal{A},\bcal{B})} \subseteq \cap_{\bcal{A}\in\RC}\denot{U}_{\rho,(X_+,X_-)\mapsto(\bcal{A},\bcal{A})}.
  \]
  Now, suppose that 
  $\ve t \in
  \cap_{\bcal{A}\in\RC}\denot{U}_{\rho,(X_+,X_-)\mapsto(\bcal{A},\bcal{A})}$. Pick
  any pair $\bcal{B}\subseteq\bcal{A}\in\RC$.
  Then $\ve t
  \in\denot{U}_{\rho,(X_+,X_-)\mapsto(\bcal{A},\bcal{A})}$.
  From Lemma~\ref{lem:polar2appendix} and because
  $\bcal{B}\subseteq\bcal{A}$,
  \[
    \denot{U}_{\rho,(X_+,X_-)\mapsto(\bcal{A},\bcal{A})} \subseteq
    \denot{U}_{\rho,(X_+,X_-)\mapsto(\bcal{A},\bcal{B})}.
  \]
  Since this is true for any pair $\bcal{B}\subseteq\bcal{A}\in\RC$,
  we deduce that
  \[
    \cap_{\bcal{A}\in\RC}\denot{U}_{\rho,(X_+,X_-)\mapsto(\bcal{A},\bcal{A})} \subseteq \cap_{\bcal{B}\subseteq\bcal{A}\in\RC}\denot{U}_{\rho,(X_+,X_-)\mapsto (\bcal{A},\bcal{B})}.
  \]
  We therefore have the required equality.
\end{proof}

\begin{remark}\rm
  Now, we can prove the Adequacy Lemma. In the proof, we replace the
  definition of $\denot{\forall\vec X.U}_\rho$ with the more precise
  intersection proposed in Lemma \ref{lem:betterforall}.
\end{remark}

\recap{Lemma}{Adequacy Lemma}{lem:SNadeq}
Every derivable typing judgement is valid: For every valid sequent
$\Gamma\vdash\ve t:T$, we have $\Gamma\models\ve t:T$.

\begin{proof}
  The proof of the adequacy lemma is made by induction on the size of
  the typing derivation of $\Gamma\vdash\ve t: T$. We look
  at the last typing rule that is used, and show in each case that
  $\Gamma\models\ve t: T$, \ie if $T\equiv\vara U$, then $\ve 
  t_\sigma\in\denot{\vara U}_\rho$ or if 
  $T\equiv\sui{n}\alpha_i.\vara U_i$ in
  the sense of Lemma~\ref{lem:typedecomp}, then 
  $\ve t_{\sigma}\in\sum_{i=1}^n\denot{\vara U_i}_{\rho,\rho_i}$,
  for every valid valuation~$\rho$, set of valid valuations $\{\rho_i\}_n$, and
  substitution~$\sigma\in\denot{\Gamma}_\rho$ (\ie substitution $\sigma$
  such that $(x:V)\in\Gamma$ implies $x_\sigma\in\denot{V}_{\bar\rho}$).
\item \parbox{3.4cm}{
    \prooftree
    \justifies\Gamma,x: U\vdash x: U
    \using ax
\endprooftree}
\parbox{10cm}{ Then for any $\rho$,
  $\forall\sigma\in\denot{\Gamma,x: U}_{\rho}$ by definition we
  have $x_\sigma\in\denot{U}_{\bar\rho}$.From
  Lemma~\ref{lem:polar1appendix}, we deduce that
  $x_\sigma\in\denot{U}_{\rho}$. 
}

\item\parbox{3cm}{
    \prooftree\Gamma\vdash\ve t: T
    \justifies\Gamma\vdash\ve 0: 0\cdot T
    \using 0_I
  \endprooftree
}
\parbox{10,4cm}{
  Note that $\forall \sigma,\,\ve 0_\sigma=\ve 0$, and $\ve 0$ is in any reducibility candidate by $\RCn_4$.
}

\item\parbox{4cm}{
    \prooftree\Gamma,x: U\vdash\ve t: T
    \justifies\Gamma\vdash\lambda x.\ve t: U\to T
    \using\to_I
  \endprooftree
}
\parbox{9,4cm}{Let $T\equiv\vara V$ or
  $T\equiv\sui{n}\alpha_i\cdot\vara U_i$ with $n>1$.
  Then by the induction hypothesis, for any $\rho$,
  set $\{\rho_i\}_n$ not acting on $FV(\Gamma)\cup FV(U)$, and
  $\forall\sigma\in\denot{\Gamma,x: U}_\rho$, we have
  $\ve t_\sigma\in\sum_{i=1}^n\denot{\vara U_i}_{\rho,\rho_i}$, or
  simply $\ve t_\sigma\in\denot{\vara V}_\rho$ if $T\equiv\vara V$.}

  In any case, we must prove that $\forall\sigma\in\denot{\Gamma}_\rho$,
  $(\lambda x.\ve t)_\sigma\in\denot{U\to T}_{\rho,\rho'}$, or
  what is the same $\lambda x.\ve
  t_\sigma\in\denot{U}_{\bar\rho,\bar\rho'}\to\denot{T}_{\rho,\rho'}$,
  where $\rho'$ does not act on $FV(\Gamma)$. If we can show
  that $\ve b\in\denot{U}_{\bar\rho,\bar\rho'}$ implies
  $(\lambda x.\ve t_\sigma)~\ve b\in\denot{T}_{\rho,\rho'}$,
  then we are done.
  Notice that
  $\denot{T}_{\rho,\rho'}=\sum_{i=1}^n\denot{\vara U_i}_{\rho,\rho'}$, or
  $\denot{T}_{\rho,\rho'}=\denot{\vara V}_{\rho,\rho'}$         
  Since
  $(\lambda x.\ve t_\sigma)~\ve b$
  is a neutral term, we just need to prove that every
  one-step reduction of it is in $\denot{T}_\rho$, which by
  $\RCn_3$ closes the case.
  By $\RCn_1$, $\ve t_\sigma$ and $\ve b$ are strongly normalising, and so is 
  $\lambda x.\ve t_\sigma$. Then we proceed by induction on the sum of the 
  lengths of all the reduction paths starting from $(\lambda x.\ve 
  t_{\sigma})$ plus the same sum starting from $\ve b$:
  \begin{description}
    \item[$(\lambda x.\ve t_{\sigma})~\ve b\to(\lambda x.\ve
      t_{\sigma})~\ve b'$] with $\ve b\to \ve b'$. Then $\ve
      b'\in\denot{U}_{\bar\rho,\bar\rho'}$ and we close by induction
      hypothesis.
    \item[$(\lambda x.\ve t_{\sigma})~\ve b\to(\lambda x.\ve t')~\ve b$]
      with $\ve t_{\sigma}\to\ve t'$. If $T\equiv\vara V$, then
      $\ve t_\sigma\in\denot{\vara V}_{\rho,\rho'}$, and by $\RCn_2$ so is $\ve t'$. 
      In other case 
      $\ve t_{\sigma}\in\sum_{i=1}^n\denot{\vara U_i}_{\rho,\rho_i}$
      for any $\{\rho_i\}_n$ not acting on $FV(\Gamma)$, take $\forall
      i,\,\rho_i=\rho'$, so $\ve t_{\sigma}\in\denot{T}_{\rho,\rho'}$
      and so are its reducts, such as $\ve t'$. We close by induction
      hypothesis.
    \item[$(\lambda x.\ve t_{\sigma})~\ve b\to\ve t_{\sigma}\subst{\ve b}{x}$]
      Let $\sigma'=\sigma;x\mapsto\ve b$.
      Then $\sigma'\in\denot{\Gamma,x: U}_{\rho,\rho'}$, so 
      $\ve t_{\sigma'}\in\denot{T}_{\rho,\rho_i}$. Notice that
      $\ve t_{\sigma}[\ve b/x]=\ve t_{\sigma'}$.
  \end{description}

\item\prooftree
  \Gamma\vdash\ve t:\sui{n}\alpha_i\cdot \forall\vec{X}.(U\to T_i)
  \qquad
  \Gamma\vdash\ve r:\suj{m}\beta_j\cdot U[\vec{A}_j/\vec{X}]
  \justifies
  \Gamma\vdash(\ve t)~\ve r:
  \sui{n}\suj{m}\alpha_i\times\beta_j\cdot T_i[\vec{A}_j/\vec{X}]
  \using\to_E
\endprooftree

Without loss of generality, assume that the $T_i$'s are different from 
each other (similarly for $\vec{A}_j$).
By the induction hypothesis, for any $\rho$, 
$\{\rho_{i,j}\}_{n,m}$ not acting on $FV(\Gamma)$, and
$\forall\sigma\in\denot{\Gamma}_\rho$ we have
$\ve t_\sigma\in\sum_{i=1}^n
\cap_{\vec{\bcal{A}}\subseteq\vec{\bcal{B}}\in\RC}
\denot{(U\to T_i)}_{\rho,\rho_i,{(\vec{X}_+,\vec{X}_-)
\mapsto(\vec{\bcal{A}},\vec{\bcal{B}})}}$
and
$\ve r_\sigma\in
\sum_{j=1}^m\denot{U[\vec{A}_j/\vec{X}]}_{\rho,\rho_j}$,
or if $n=\alpha_1=1$, 
$\ve t_\sigma\in
\cap_{\vec{\bcal{A}}\subseteq\vec{\bcal{B}}\in\RC}
\denot{(U\to T_1)}_{\rho,{(\vec{X}_+,\vec{X}_-)
\mapsto(\vec{\bcal{A}},\vec{\bcal{B}})}}$
and if $m=1$ and $\beta_1=1$,
$\ve r_\sigma\in\denot{U[\vec{A}_j/\vec{X}]}_{\rho}$.
Notice that for any $\vec{A}_j$, if $U$ is a unit type,
$U[\vec{A}_j/\vec{X}]$ is still unit.

For every $i,j$, let
$T_i[\vec{A}_j/\vec{X}]\equiv\suk{r^{ij}}\delta^{ij}_k\cdot\vara W^{ij}_k$.
We must show that for any $\rho$, sets
$\{\rho'_{i,j,k}\}_{r_{i,j}}$ not acting on $FV(\Gamma)$ and
$\forall\sigma\in\denot{\Gamma}_\rho$, the term
$((\ve t)~\ve r)_\sigma$ is in the set 
$\sum_{i=1\cdots n, j=1\cdots m,k=1\cdots r^{ij}}
\denot{\vara W^{ij}_k}_{\rho,\rho_{ijk}}$,
or in case of
$n=m=\alpha_1=\beta_1=r^{11}=1$,
$((\ve t)~\ve r)_\sigma\in\denot{\vara W^{11}_1}_\rho$.

Since both
$\ve t_\sigma$ and $\ve r_\sigma$ are strongly normalising,
we proceed by induction on the sum of the lengths of their
rewrite sequence. The set $\Red(((\ve t)~\ve r)_\sigma)$
contains:
\begin{itemize}
  \item $(\ve t_\sigma)~\ve r'$ or $(\ve t')~\ve r_\sigma$ when
    $\ve t_\sigma\to\ve t'$ or $\ve r_\sigma\to\ve r'$. By
    $\RCn_2$, the term $\ve t'$ is in the set $\sum_{i=1}^n
    \cap_{\vec{\bcal{A}}\subseteq\vec{\bcal{B}}\in\RC}\denot{(U\to
    T_i)}_{\rho,\rho_i,{(\vec{X}_+,\vec{X}_-)
    \mapsto(\vec{\bcal{A}},\vec{\bcal{B}})}}$
    (or if $n=\alpha_1=1$, the term $\ve t'$ is in
    $\cap_{\vec{\bcal{A}}\subseteq\vec{\bcal{B}}\in\RC}
    \denot{(U\to T_1)}_{\rho,{(\vec{X}_+,\vec{X}_-)
    \mapsto(\vec{\bcal{A}},\vec{\bcal{B}})}}$),
    and
    $\ve r'\in\sum_{j=1}^m\denot{U[\vec{A}_j/\vec{X}]}_{\rho,\rho_j}$
    (or in $\denot{U[\vec{A}_1/\vec{X}]}_\rho$ if $m=\beta_1=1$).
    In any case, we conclude by the induction hypothesis.
  \item $({\ve t_1}_\sigma)~\ve r_\sigma+({\ve t_2}_\sigma)~\ve r_\sigma$
    with $\ve t_\sigma={\ve t_1}_\sigma+{\ve t_2}_\sigma$,
    where, $\ve t=\ve t_1+\ve t_2$.
    Let $s$ be the size of the derivation of
    $\Gamma\vdash\ve t:\sui{n}\alpha_i\cdot \forall\vec{X}.(U\to T_i)$.
    By Lemma~\ref{lem:sums}, there exists
    $R_1+R_2\equiv\sui{n}\alpha_i\cdot \forall\vec{X}.(U\to T_i)$
    such that
    $\Gamma\vdash{\ve t_1}_\sigma:R_1$ and
    $\Gamma\vdash{\ve t_2}_\sigma:R_2$
    can be derived with a derivation tree of size $s-1$ if
    $R_1+R_2=\sui{n}\alpha_i\cdot \forall\vec{X}.(U\to T_i)$,
    or of size $s-2$ in other case. In such case, there exists\
    $N_1,N_2\subseteq\{1,\dots,n\}$
    with $N_1\cup N_2=\{1,\dots,n\}$
    such that
    \begin{align*}
      R_1\equiv
      \sum\limits_{i\in N_1\setminus N_2}
      \alpha_i\cdot\forall\vec X.(U\to T_i)+
      \sum\limits_{i\in N_1\cap N_2}
      \alpha'_i\cdot\forall\vec X.(U\to T_i)
      & \mbox{\quad and}\\
      R_2\equiv
      \sum\limits_{i\in N_2\setminus N_1}
      \alpha_i\cdot\forall\vec X.(U\to T_i)+
      \sum\limits_{i\in N_1\cap N_2}
      \alpha''_i\cdot\forall\vec X.(U\to T_i) &
    \end{align*}
    where $\forall i\in N_1\cap N_2$,
    $\alpha'_i+\alpha''_i=\alpha_i$. 
    Therefore, using $\equiv$ we get
    \begin{align*}
      \Gamma\vdash{\ve t_1}:
      \sum\limits_{i\in N_1\setminus N_2}
      \alpha_i\cdot\forall\vec X.(U\to T_i)+
      \sum\limits_{i\in N_1\cap N_2}
      \alpha'_i\cdot\forall\vec X.(U\to T_i) 
      & 		\mbox{\quad and}\\
      \Gamma\vdash{\ve t_2}:
      \sum\limits_{i\in N_2\setminus N_1}
      \alpha_i\cdot\forall\vec X.(U\to T_i)+
      \sum\limits_{i\in N_1\cap N_2}
      \alpha''_i\cdot\forall\vec X.(U\to T_i) &
    \end{align*} 
    with a derivation three of size $s-1$.
    So, using rule $\to_E$, we get
    \begin{align*}
      \Gamma\vdash(\ve t_1)~\ve r:
      \sum\limits_{i\in N_1\setminus N_2}
      \suj{m}\alpha_i\times\beta_j\cdot T_i[\vec A_j/\vec X]+
      \sum\limits_{i\in N_1\cap N_2}
      \suj{m}\alpha'_i\times\beta_j\cdot T_i[\vec A_j/\vec X]
      & \mbox{\qquad and}\\
      \Gamma\vdash(\ve t_2)~\ve r:
      \sum\limits_{i\in N_2\setminus N_1}
      \suj{m}\alpha_i\times\beta_j\cdot T_i[\vec A_j/\vec X]+
      \sum\limits_{i\in N_1\cap N_2}
      \suj{m}\alpha''_i\times\beta_j\cdot T_i[\vec A_j/\vec X] &
    \end{align*}
    with a derivation three of size $s$.
    Hence, by the induction hypothesis the term
    $({\ve t_1}_\sigma)~\ve r_\sigma$ is in the set
    $\sum_{i=N_1, j=1\cdots m,k=1\cdots r^{ij}}
    \denot{\vara W^{ij}_k}_{\rho,\rho_{ijk}}$,
    and the term
    $({\ve t_2}_\sigma)~\ve r_\sigma$ is in
    $\sum_{i=N_2, j=1\cdots m,k=1\cdots r^{ij}}
    \denot{\vara W^{ij}_k}_{\rho,\rho_{ijk}}$.
    Hence, by Lemma~\ref{lem:sumrcappendix} the term
    $({\ve t_1}_\sigma)~\ve r_\sigma+({\ve t_2}_\sigma)~\ve r_\sigma$
    is in the set
    $\sum_{i=1,\dots,n, j=1\cdots m,k=1\cdots r^{ij}}
    \denot{\vara W^{ij}_k}_{\rho,\rho_{ijk}}$.
    The case where
    $m=\alpha_1=\beta_1=r^{11}=1$, and $card(N_1)$ or $card(N_2)$
    is equal to $1$ follows analogously.

  \item $(\ve t_\sigma)~{\ve r_1}_\sigma+(\ve t_\sigma)~{\ve r_2}_\sigma$
    with $\ve r_\sigma={\ve r_1}_\sigma+{\ve r_2}_\sigma$.
    Analogous to previous case.
  \item $\gamma\cdot(\ve t'_\sigma)~\ve r_\sigma$
    with $\ve t_\sigma=\gamma\cdot\ve t'_\sigma$,
    where $\ve t=\gamma\cdot\ve t'$.
    Let $s$ be the size of the derivation of
    $\Gamma\vdash\gamma\cdot\ve t':
    \sui{n}\alpha_i\cdot\forall\vec{X}.(U\to T_i)$.
    Then by Lemma~\ref{lem:scalars},
    $\sui{n}\alpha_i\cdot\forall\vec X.(U\to T_i)\equiv\alpha\cdot R$ 
    and
    $\Gamma\vdash\ve t':R$.
    If
    $\sui{n}\alpha_i\cdot\forall\vec X.(U\to T_i)=\alpha\cdot R$,
    such a derivation is obtained with size $s-1$, in other case it is 
    obtained in size $s-2$ and by Lemma~\ref{lem:typecharact}, 
    $R\equiv\sui{n'}\gamma_i\cdot V_i
    +\suk{h}\eta_k\cdot\vara X_k$, however it is easy to see that
    $h=0$ because $R$ is equivalent to a sum of terms, where none of them is $\vara X$.
    So $R\equiv\sui{n'}\gamma_i\cdot V_i$. Notice that
    $\sui{n}\alpha_i\cdot\forall\vec X.(U\to T_i)
    \equiv
    \sui{n'}\alpha\times\gamma_i\cdot V_i$.
    Then by Lemma~\ref{lem:equivdistinctscalars}, there exists a 
    permutation $p$ such that $\alpha_i=\alpha\times\gamma_{p(i)}$ and
    $\forall\vec X.(U\to T_i)\equiv V_{p(i)}$. Then by rule $\equiv$, 
    in size $s-1$ we can derive 
    $\Gamma\vdash\ve t':\sui{n}\gamma_i\cdot\forall\vec X.(U\to T_i)$.
    Using rule $\to_E$, we get
    $\Gamma\vdash(\ve t')~\ve r:
    \sui{n}\suj{m}\gamma_i\times\beta_j\cdot T_i[\vec A_j/\vec X]$
    in size $s$. Therefore, by the induction hypothesis,
    $({\ve t'}_\sigma)~\ve r_\sigma$
    is in the set
    $\sum_{i=1,\dots,n, j=1\cdots m,k=1\cdots r^{ij}}
    \denot{\vara W^{ij}_k}_{\rho,\rho_{ijk}}$.
    We conclude with Lemma~\ref{lem:alphatinsigma}.

  \item $\gamma\cdot(\ve t_\sigma)~\ve r'_\sigma$
    with $\ve r_\sigma=\gamma\cdot\ve r'_\sigma$.
    Analogous to previous case.
  \item $\ve 0$
    with $\ve t_\sigma=\ve 0$, or $\ve r_\sigma=\ve 0$.
    By $\RCn_4$, $\ve 0$ is in every candidate.

  \item The term $\ve t'_\sigma[\ve r_\sigma/x]$, when $\ve
    t_\sigma=\lambda x.\ve t'$ and $\ve r$ is a base
    term. Note that this term is of the form $\ve
    t'_{\sigma'}$ where $\sigma'=\sigma;x\mapsto\ve r$. We are
    in the situation where the types of $\ve t$ and $\ve r$
    are respectively $\forall\vec{X}.(U\to T)$ and 
    $U[\vec{A}/\vec{X}]$, and so
    $\sum_{i,j,k}\denot{\vara W^{ij}_k}_{\rho,\rho_{ijk}}=
    \sum_{k=1}^{r}\denot{\vara W_k}_{\rho,\rho_{k}}$,
    where we omit the index ``${11}$'' (or directly $\denot{\vara W}_\rho$ if 
    $r=1$).
    Note that
    \[
      \lambda x.\ve
      t'_\sigma\in\denot{\forall\vec{X}.(U\to
    T)}_{\rho,\rho'}=\cap_{\vec{\bcal{A}}\subseteq\vec{\bcal{B}}\in\RC}\denot{U\to
    T}_{\rho,\rho',(\vec{X}_+,\vec{X}_-)\mapsto(\vec{A},\vec{B})}
  \]
  for all possible $\rho'$ such that $|\rho'|$ does not
  intersect $\FV\Gamma$.  Choose $\vec{\bcal{A}}$ and
  $\vec{\bcal{B}}$ equal to $\denot{\vec A}_{\rho,\rho'}$
  and choose $\rho'_-$ to send every $X$ in its domain to
  $\cap_{k}\rho_{k-}(X)$ and $\rho'_+$ to send all the $X$
  in its domain to $\sum_{k}\rho_{k+}(X)$. Then
  by definition of $\to$ and Lemma~\ref{lem:substRed},
  \begin{align*}
    \lambda x.\ve t'_\sigma &
    \in \denot{U\to T}_{\rho,\rho',(\vec X_+,\vec X_-)\mapsto (\denot{\vec A}_{\bar\rho,\bar\rho'}, \denot{\vec A}_{\rho,\rho'})} \\
    &=
    \denot{U[\vec{A}/\vec{X}]}_{\bar\rho,\bar\rho'}\to \denot{T}_{\rho,\rho',(\vec X_+,\vec X_-)\mapsto (\denot{\vec A}_{\bar\rho,\bar\rho'}, \denot{\vec A}_{\rho,\rho'})}.
  \end{align*}
  Since $\ve r\in\denot{U[\vec{A}/\vec{X}]}_{\bar\rho,\bar\rho'}$,
  using Lemmas~\ref{lem:applircappendix} and \ref{lem:substRed},
  \begin{align*}
    (\lambda x.\ve t_\sigma)~\ve r 
    &\in\denot{T}_{\rho,\rho',(\vec X_+,\vec X_-)\mapsto 
      (\denot{\vec A}_{\bar\rho,\bar\rho'},
      \denot{\vec A}_{\rho,\rho'})
    }\\
    &=
    \denot{T[\vec{A}/\vec{X}]}_{\rho,\rho'}\\
    &=
    \sum_{k=1}^n\denot{\vara W_k}_{\rho,\rho'}\quad\mbox{ or just 
    }\quad\denot{\vara W_1}_{\rho,\rho'}\mbox{ if }n=1.
  \end{align*}
  Now, from Lemma~\ref{lem:polar2appendix}, for all $k$ we have
  $\denot{\vara W_k}_{\rho,\rho'}\subseteq\denot{\vara W_k}_{\rho,\rho_k}$.
  Therefore
  \[
    (\lambda x.\ve
    t_\sigma)~\ve
    r\in\sum_{k=1}^n\denot{\vara W_k}_{\rho,\rho_k}\enspace.
  \]
\end{itemize}
Since the set $\Red(((\ve t)~\ve
r)_\sigma)\subseteq\sui n\suj m\suk{r^{ij}} \denot{\vara W^{ij}_k}_{\rho,\rho_{ijk}}$, we can conclude by $\RCn_3$.

\item\parbox{5.7cm}{
    \prooftree\Gamma\vdash\ve t:\sui{n}\alpha_i\cdot U_i\quad X\notin FV(\Gamma)
    \justifies\Gamma\vdash\ve t:\sui{n}\alpha_i\cdot \forall X.U_i
    \using\forall_I
  \endprooftree
}
\parbox{7.7cm}{ By the induction hypothesis,
  for any $\rho$, set $\{\rho_i\}_n$ not acting on
  $FV(\Gamma)$, we have $\forall\sigma\in\denot{\Gamma}_\rho$,
  $\ve t_\sigma\in\sum_{i=1}^n\denot{U_i}_{\rho,\rho_i}$
  (or $\ve t_\sigma\in\denot{U_1}_{\rho,\rho_1}$ if $n=\alpha_1=1$).
  Since
  $X\notin FV(\Gamma)$, we can take
  $\rho_i=\rho'_i,(X_+,X_-)\mapsto(\bcal{A},\bcal{B})$, then
  for any $\bcal{B}\subseteq\bcal{A}$, we have
  $\ve t_\sigma
  \in\sum_{i=1}^n\denot{U_i}_{\rho,\rho'_i,(X_+,X_-)\mapsto (\bcal{A},\bcal {B})}$ (or $\ve t_\sigma\in \denot{U_1}_{\rho,\rho'_1,(X_+,X_-)\mapsto (\bcal{A},\bcal {B})}$ if $n=\alpha_1=1$).}

  {Since it is valid for any $\bcal{B}\subseteq\bcal{A}$, we can take all the intersections, thus we have $\ve t_\sigma\in\sum_{i=1}^n\cap_{\bcal{B}\subseteq\bcal{A}\in\RC}\denot{U_i}_{\rho,\rho'_i,(X_+,X_-)\mapsto(\bcal{A},\bcal{B})}=\sum_{i=1}^n\denot{\forall X.U_i}_{\rho,\rho'_i}$ (or  if $n=\alpha_1=1$ simply $\ve t_\sigma\in \cap_{\bcal{B}\subseteq\bcal{A}\in\RC}\denot{U_1}_{\rho, \rho'_1,(X_+,X_-)\mapsto(\bcal{A},\bcal{B})} =\denot{\forall X.U_1}_{\rho,\rho'_1}$).  }

\item\parbox{4.5cm}{
    \prooftree\Gamma\vdash\ve t:\sui{n}\alpha_i\cdot \forall X.U_i
    \justifies\Gamma\vdash\ve t:\sui{n}\alpha_i\cdot U_i[A/X]
    \using\forall_E
  \endprooftree
}
\parbox{8.9cm}{ By the induction hypothesis,
  for any $\rho$ and for all families $\{\rho_i\}_n$, we have
  for all $\sigma$ in $\denot{\Gamma}_\rho$ that the term
  $\ve t_\sigma$ is in
  $\sum_{i=1}^n\denot{\forall X.U_i}_{\rho,\rho_i}
  =\sum_{i=1}^n\cap_{\bcal{B}\subseteq\bcal{A}\in\RC}\denot{
  U_i}_{\rho,\rho'_i,(X_+,X_-)\mapsto(\bcal{A},\bcal{B})}$
  (or  if $n=\alpha_1=1$, 
  $\ve t_\sigma$ is in the set
  $\denot{\forall X.U_1}_{\rho,\rho_1}
  =\cap_{\bcal{B}\subseteq\bcal{A}\in\RC}\denot{
  U_1}_{\rho,\rho'_1,(X_+,X_-)\mapsto(\bcal{A},\bcal{B})}$).
  Since it is in the intersections, we can chose
  $\bcal{A}=\denot{A}_{\bar\rho,\bar\rho_i}$ and 
  $\bcal{B}=\denot{A}_{\rho,\rho_i}$, and
  then  $\ve t_\sigma\in
  \sum_{i=1}^n\denot{U_i}_{\rho,\rho'_i,X\mapsto
  A}=\sum_{i=1}^n\denot{U_i[A/X]}_{\rho,\rho'_i}$
  (or
  $\ve t_\sigma
  \in\denot{U_1}_{\rho,\rho'_1,X\mapsto A}
  =\denot{U_i[A/X]}_{\rho,\rho'_1}$,
if $n=\alpha_1=1$).}

\item\parbox{3.2cm}{
    \prooftree\Gamma\vdash\ve t: T
    \justifies\Gamma\vdash\alpha\cdot \ve t:\alpha\cdot T
    \using\alpha_I
  \endprooftree
}
\parbox{10.2cm}{ Let
  $T\equiv\sui{n}\beta_i\cdot\vara U_i$, so
  $\alpha\cdot T\equiv\sui{n}\alpha\times\beta_i\cdot\vara U_i$. By the
  induction hypothesis, for any $\rho$, we have
  $\forall\sigma\in\denot{\Gamma}_\rho$, $\ve
  t_\sigma\in\sum_{i=1}^n\denot{\vara U_i}_{\rho,\rho_i}$. 
  By Lemma~\ref{lem:alphatinsigma}, $(\alpha\cdot \ve 
  t)_\sigma=\alpha\cdot \ve
  t_\sigma\in\sum_{i=1}^n\denot{\vara U_i}_{\rho,\rho_i}$.
  Analogous if $n=\beta_1=1$.
}

  \item\parbox{4.5cm}{
      \prooftree\Gamma\vdash\ve t: T\qquad\Gamma\vdash\ve r: R
      \justifies\Gamma\vdash\ve t+\ve r: T+R
      \using +_I
    \endprooftree
  }
  \parbox{8.9cm}{ Let
    $T\equiv\sui{n}\alpha_i\cdot\vara U_{i1}$ and
    $R\equiv\suj{m}\beta_j\cdot\vara U_{j2}$. By the induction hypothesis,
    for any $\rho$, $\{\rho_i\}_{n}$, $\{\rho'_j\}_{m}$, we have
    $\forall\sigma\in\denot{\Gamma}_\rho$, $\ve
    t_\sigma\in\sum_{i=1}^n\denot{\vara U_{i1}}_{\rho,\rho_i}$
    and $\ve
    r_\sigma\in\sum_{j=1}^m\denot{\vara U_{j2}}_{\rho,{\rho'}_j}$. Then
    by Lemma~\ref{lem:sumrcappendix},
    $(\ve t+\ve r)_\sigma=\ve t_\sigma+\ve
    r_\sigma\in\sum_{i,k}\denot{\vara U_{ik}}_{\rho,\rho_i}$.
    Analogous if $n=\beta_1=1$ and/or $m=\beta_1=1$.
  }

\item\parbox{4cm}{
    \prooftree\Gamma\vdash\ve t: T\qquad T\equiv R
    \justifies\Gamma\vdash\ve t: R
    \using\equiv
  \endprooftree
}
\parbox{9.4cm}{
  Let $T\equiv\sui{n}\alpha_i\cdot\vara U_i$ in the sense of Lemma~\ref{lem:typedecomp}, then since $T\equiv R$, $R$ is also equivalent to $\sui{n}\alpha_i\cdot\vara U_i$, so $\Gamma\vDash\ve t: T\Rightarrow\Gamma\vDash\ve t: R$.\qedhere
}
\end{proof}

\section{Detailed proofs of lemmas and theorems in Section~\ref{sec:examples}}\label{app:examples}
\recap{Theorem}{Characterisation of terms}{thm:termcharact}
Let $T$ be a generic type with canonical decomposition $\sui{n}\alpha_i.\vara{U}_i$, in the sense of Lemma~\ref{lem:typedecomp}. If ${}\vdash\ve t:T$,
then $\ve t\to^*\sui{n}\suj{m_i}\beta_{ij}\cdot\ve b_{ij}$, where
for all $i$, $\vdash\ve b_{ij}:\vara{U}_i$ and
$\suj{m_i}\beta_{ij}=\alpha_i$, and with the convention that $\suj{0}\beta_{ij}=0$ and $\suj{0}\beta_{ij}\cdot\ve b_{ij}=\ve 0$.

\begin{proof}
  We proceed by induction on the maximal length of reduction from $\ve t$.
  \begin{itemize}
    \item Let $\ve t=\ve b$ or $\ve t=\ve 0$. Trivial using Lemma~\ref{lem:basevectors} or \ref{lem:termzero},  and Lemma~\ref{lem:typedecomp}.
    \item Let $\ve t=(\ve t_1)~\ve t_2$. Then by Lemma~\ref{lem:app},
      $\vdash\ve t_1:\suk{o}\gamma_k\cdot\forall\vec{X}.(U\to T_k)$ and
      $\vdash\ve t_2:\sum_{l=1}^{p}\delta_l\cdot U[\vec{A}_l/\vec{X}]$, where
      $\suk{o}\sum_{l=1}^p\gamma_k\times\delta_l\cdot T_k[\vec{A}_l/\vec{X}]\succeq^{(\ve t_1)\ve t_2}_{\V,\emptyset}T$, for some $\V$.
      Without loss of generality, consider these two types to be already canonical decompositions, that is, for all $k_1,k_2$, $T_{k_1}\not\equiv T_{k_2}$ and for all $l_1,l_2$, $U[\vec{A}_{l_1}/\vec{X}]\not\equiv U[\vec{A}_{l_2}/\vec{X}]$ (in other case, it suffices to sum up the equal types).
      Hence, by the induction hypothesis,
      $\ve t_1\to^*\suk{o}\sum_{s=1}^{q_k}\psi_{ks}\cdot\ve b_{ks}$ and
      $\ve t_2\to^*\sum_{l=1}^p\sum_{r=1}^{t_l}\phi_{lr}\cdot\ve b'_{lr}$, where
      for all $k$,
      $\vdash\ve b_{ks}:\forall\vec{X}.(U\to T_k)$ and $\sum_{s=1}^{q_k}\psi_{ks}=\gamma_k$,
      and for all $l$,
      $\vdash\ve b'_{lr}: U[\vec{A}_l/\vec{X}]$ and $\sum_{r=1}^{t_l}\phi_{lr}=\delta_l$.
      By rule $\to_E$, for each $k,s,l,r$ we have $\vdash(\ve b_{ks})~\ve b'_{lr}:T_k[\vec{A}_l/\vec{X}]$, where the induction hypothesis also apply, and notice that
      $(\ve t_1)~\ve t_2\to^*(\suk{o}\sum_{s=1}^{q_k}\psi_{ks}\cdot\ve b_{ks})~\sum_{l=1}^p\sum_{r=1}^{t_l}\phi_{lr}\cdot\ve b'_{lr}
      \to^*
      \suk{o}\sum_{s=1}^{q_k}\sum_{l=1}^p\sum_{r=1}^{t_l}\psi_{ks}\times\phi_{lr}\cdot(\ve b_{ks})~\ve b'_{lr}$. Therefore, we conclude with the induction hypothesis.
    \item Let $\ve t=\alpha\cdot\ve r$. Then by Lemma~\ref{lem:scalars}, $\vdash\ve r:R$, with $\alpha\cdot R\equiv T$. Hence, using Lemmas~\ref{lem:typedecomp} and \ref{lem:equivdistinctscalars},
      $R$ has a type decomposition $\sui{n}\gamma_i\cdot\vara U_i$, where $\alpha\times\gamma_i=\alpha_i$.
      Hence, by the induction hypothesis,
      $\ve r\to^*\sui{n}\suj{m_i}\beta_{ij}\cdot\ve b_{ij}$, where for all $i$,
      $\vdash\ve b_{ij}:\vara U_i$ and
      $\suj{m_i}\beta_{ij}=\gamma_i$.
      Notice that
      $\ve t=\alpha\cdot\ve r\to^*\alpha\cdot\sui{n}\suj{m_i}\beta_{ij}\cdot\ve b_{ij}\to^*\sui{n}\suj{m_i}\alpha\times\beta_{ij}\cdot\ve b_{ij}$,
      and
      $\alpha\cdot\suj{m_i}\beta_{ij}=\suj{m_i}\alpha\times\beta_{ij}=\alpha\times\gamma_i=\alpha_i$.
    \item Let $\ve t=\ve t_1+\ve t_2$. Then by Lemma~\ref{lem:sums}, $\vdash\ve t_1:T_1$ and $\vdash\ve t_2:T_2$, with $T_1+T_2\equiv T$. By Lemma~\ref{lem:typedecomp}, $T_1$ has canonical decomposition $\suj{m}\beta_j\cdot\vara{V}_j$ and $T_2$ has canonical decomposition $\suk{o}\gamma_k\cdot\vara{W}_k$. Hence
      by the induction hypothesis
      $\ve t_1\to^*\suj{m}\sum_{l=1}^{p_j}\delta_{jl}\cdot\ve b_{jl}$ and
      $\ve t_2\to^*\suk{o}\sum_{s=1}^{q_k}\epsilon_{ks}\cdot\ve b'_{ks}$, where
      for all $j$, $\vdash\ve b_{jl}:\vara{V}_{j}$ and $\sum_{l=1}^{p_j}\delta_{jl}=\beta_j$, and
      for all $k$, $\vdash\ve b'_{ks}:\vara{W}_{k}$ and $\sum_{s=1}^{q_k}\epsilon_{ks}=\gamma_k$.
      In for all $j,k$ we have $\vara{V}_j\neq\vara{W}_k$, then we are done since the canonical decomposition of $T$ is $\suj{m}\beta_j\cdot\vara{V}_j+\suk{o}\gamma_k\cdot\vara{W}_k$. In other case, suppose there exists $j',k'$ such that  $\vara{V}_{j'}=\vara{W}_{k'}$, then the canonical decomposition of $T$ would be
      $\sum_{j=1,j\neq j'}^{m}\beta_j\cdot\vara{V}_j+\sum_{k=1,k\neq k'}^{o}\gamma_k\cdot\vara{W}_k+(\beta_{j'}+\gamma_{k'})\cdot\vara{V}_{j'}$.
      Notice that $\sum_{l=1}^{p_{j'}}\delta_{j'l}+\sum_{s=1}^{q_{k'}}\epsilon_{k's}=\beta_{j'}+\gamma_{k'}$.
      \qedhere
  \end{itemize}
\end{proof}
\end{document}